\documentclass{article}

    \PassOptionsToPackage{round,compress}{natbib}

\usepackage[main, final]{neurips_2026}

\usepackage[utf8]{inputenc} 
\usepackage[T1]{fontenc}    
\usepackage{hyperref}       
\hypersetup{linktoc=page,hidelinks}
\usepackage{url}            
\usepackage{booktabs}       
\usepackage{nicefrac}       
\usepackage{microtype}      
\usepackage{xcolor}         

\usepackage{float}
\usepackage{dblfloatfix}
\usepackage{paralist}

\usepackage[toc,page,header]{appendix}
\usepackage{minitoc}
\usepackage{tocloft}

\usepackage{amsmath}
\usepackage{amssymb}
\usepackage{mathtools}
\usepackage{amsthm}
\usepackage{thmtools}
\usepackage{thm-restate}

\usepackage{wrapfig}
\usepackage{algorithm}
\usepackage[noend]{algorithmic}

\usepackage[capitalize,noabbrev]{cleveref}

\newcommand{\para}[1]{\textbf{#1}}
\newcommand{\minisection}[1]{%
    \phantomsection%
    \addcontentsline{toc}{subsection}{#1}%
    \textbf{#1.}%
}

\theoremstyle{plain}
\newtheorem{theorem}{Theorem}[section]
\newtheorem{proposition}[theorem]{Proposition}
\newtheorem{lemma}[theorem]{Lemma}
\newtheorem{corollary}[theorem]{Corollary}

\theoremstyle{definition}
\newtheorem{definition}[theorem]{Definition}

\newtheorem{example}[theorem]{Example}
\theoremstyle{remark}
\newtheorem{remark}[theorem]{Remark}

\makeatletter
\def\thmhead@plain#1#2#3{%
  \thmname{#1}\thmnumber{\@ifnotempty{#1}{ }\@upn{#2}}%
  \thmnote{ {\the\thm@notefont#3}}}
\let\thmhead\thmhead@plain
\makeatother

\makeatletter
\newcommand{\SHORTIFELSE}[3]{%
  \ALC@it#2\ \algorithmicif\ #1\ \algorithmicelse\ #3%
}
\makeatother

\newcommand{\M}{\mathcal{M}}
\renewcommand{\O}{\mathcal{O}}
\newcommand{\midd}{\!\mid\mid\!}
\newcommand{\F}{\mathcal{F}}
\newcommand{\Na}{\mathbb{N}}

\renewcommand{\L}{\mathcal{L}}
\newcommand{\real}{\mathbb{R}}
\renewcommand{\H}{\mathcal{H}}
\newcommand{\E}{\mathbb{E}}
\newcommand{\norm}[1]{\left\Vert #1\right\Vert }
\renewcommand{\d}{\mathrm{d}}
\newcommand{\cemin}{\textup{ce}\,\textup{min}}
\newcommand{\ceinf}{\textup{ce}\,\textup{inf}}
\renewcommand{\_}{\rule{4pt}{0.4pt}}
\newcommand{\G}{\mathcal{G}}
\newcommand{\e}{\mathrm{e}}
\newcommand{\N}{\mathcal{N}}
\newcommand{\Var}{\textup{Var}}
\newcommand{\sym}{\textup{sym}}

\title{$f$-Differential Privacy Filters: Validity and Approximate Solutions}

\author{%
  \begin{tabular}{cccc}
    Long Tran$^{1}$ & Antti Koskela$^{2}$ & Ossi Räisä$^{3}\thanks{Work primarily done while at the University of Helsinki.}$ & Antti Honkela$^{1}$
  \end{tabular}\\[0.5em]
  $^{1}$Department of Computer Science, University of Helsinki, Finland \\
  $^{2}$Nokia Bell Labs, Espoo, Finland \\
  $^{3}$CISPA Helmholtz Center for Information Security, Saarbrücken, Germany \\[0.5em]
  \texttt{\{thanh-long.tran,antti.honkela\}@helsinki.fi}, \\
  \texttt{antti.h.koskela@nokia-bell-labs.com},
  \texttt{ossi.raisa@cispa.de} \\
}

\begin{document}
\raggedbottom

\maketitle

\begin{abstract}
  Accounting for privacy loss under fully adaptive composition---where mechanism choice and privacy parameters may depend on the history of prior outputs---is a central challenge in differential privacy (DP). Here, privacy filters are stopping rules ensuring a prescribed global budget is not exceeded. A leading candidate for optimal filter design is $f$-DP, which characterizes the full extent of adversarial hypothesis testing and recovers $(\varepsilon,\delta)$-DP through piece-wise linear trade-off functions, while enabling tight $(\varepsilon,\delta)$-DP accounting in standard compositions via tensor products. Yet whether such filters can be correctly defined under $f$-DP remains unclear. We show that the natural $f$-DP filter---tracking path-wise accumulating tensor products and stopping when the prescribed curve is crossed---is fundamentally invalid, precluding the direct use of standard efficient numerical Fast-Fourier-Transform accounting in the fully adaptive setting. We characterize this failure, establishing necessary and sufficient conditions for the natural filter's validity. Furthermore, we prove a fully adaptive central limit theorem for $f$-DP, establishing Gaussian convergence of cumulative privacy losses under full adaptivity. As a demonstration, we construct a closed-form approximate GDP filter for subsampled Gaussian mechanisms that provably outperforms RDP-based accounting in asymptotic regimes ($q\ll 1$ and $q\approx 1$) without tracking the full trade-off function, demonstrating that the slack in RDP is not intrinsic to adaptive composition---though CLT-based approximations are known to be optimistic at realistic subsampling rates, a limitation that remains an open challenge.
\end{abstract}

\section{Introduction}

Differential privacy (DP) rigorously quantifies privacy loss from computations on sensitive data~\citep{dwork2006calibrating,dwork2014algorithmic}. A key property is composition: complex algorithms can be analyzed by composing privacy guarantees of simpler mechanisms. Common composition results require knowing privacy parameters beforehand. However, fully adaptive compositions—where mechanisms and their privacy parameters depend on prior outputs—are essential for iterative algorithms like differentially private stochastic gradient descent and adaptive data analysis. Accounting for privacy loss in this setting is a central challenge.

Privacy filters and odometers formalize fully adaptive composition for $(\varepsilon,\delta)$-DP~\citep{rogers_privacy_2016}, and 
Rényi DP (RDP) admits simple, tight filters under full adaptivity~\citep{feldman_individual_2021}, making RDP-based accounting the current state of the art. In contrast, much less is understood about fully adaptive composition for $f$-DP \citep{dong_gaussian_2022}. Only in the special case of Gaussian DP (GDP) can accurate fully adaptive bounds be obtained~\citep{koskela_individual_2023}. Applying the $f$-DP tensor product gives the tightest composition results for non-adaptive compositions---and since any $(\varepsilon,\delta)$-DP guarantee admits an exact equivalent piece-wise linear trade-off function~\citep{dong_gaussian_2022}, this tightness extends to the $(\varepsilon,\delta)$-DP framework as well. This raises a natural question:

\emph{Does $f$-DP admit valid filters under fully adaptive composition, in analogy to RDP and GDP?}

A natural candidate filter tracks the accumulating tensor product of trade-off functions and halts once it crosses the prescribed $f$-DP budget, mirroring standard $f$-DP composition and adaptive RDP/GDP accounting schemes. Our first main result shows that this approach is fundamentally invalid.

\para{Failure of natural $f$-DP filters.} We show that composing trade-off functions and stopping at a threshold does not, in general, guarantee that the resulting stopped interaction satisfies the same $f$-DP bound. This failure occurs even when each individual mechanism is validly $f$-DP and even when the composed trade-off function never exceeds the threshold along the realized history of outcomes. We give explicit negative examples and a general construction that explains when and why such failures arise. Moreover, we show that this phenomenon occurs for fully adaptive compositions of subsampled Gaussian mechanisms, the main workhorse of private machine learning. As a consequence, the natural filter rules out the most direct route to efficient numerical Fast-Fourier-Transform (FFT) accounting \citep{pmlr-v108-koskela20b,NEURIPS2021_6097d8f3,Doroshenko2022ConnectTD} in the fully adaptive setting, blocking the standard approach for accurate $(\varepsilon, \delta)$ bounds.

\para{When $f$-DP filters are valid.} While the natural $f$-DP filter fails in general, we show that it becomes valid under specific structural conditions on the trade-off functions. We provide necessary and sufficient conditions under which stopping based on composed trade-off curves yields a correct $f$-DP guarantee. These conditions clarify why divergence-based notions such as RDP admit simple filters under full adaptivity, while trade-off-function-based notions do not.

\para{Fully adaptive $f$-DP central limit theorem.} We prove a central limit theorem (CLT) for adaptive privacy loss processes, extending the non-fully adaptive CLT of~\citet{dong_gaussian_2022} to the fully adaptive setting via a modern variant of the martingale Berry--Esseen bound \citep{FAN20191028}. This CLT suggests a general methodology for constructing approximate GDP filters: any mechanism family admitting closed-form PLRV moment approximations yields a closed-form approximate filter, with the CLT error quantifying approximation quality. As a demonstration, we instantiate this methodology for subsampled Gaussian mechanisms. At realistic subsampling rates, the resulting filter underestimates privacy loss---a known limitation of CLT-based Gaussian approximations even in the non-adaptive setting \citep{Zheng2020SharpCB,NEURIPS2021_6097d8f3,pmlr-v202-alghamdi23a}. In the extreme low- and high-subsampling regimes, however, it yields asymptotically valid guarantees strictly tighter than the fully adaptive RDP accounting of~\citet{feldman_individual_2021}, without ever tracking the full trade-off function, demonstrating that the slack in RDP stems from the lossy reduction to Rényi divergences rather than any fundamental barrier in adaptive composition.

\para{Contributions.}
The main contributions of this paper are as follows:
\begin{compactenum}
    
    \item We show that the natural $f$-DP filter --- which composes individual trade-off functions and stops at a prescribed threshold --- is not valid under fully adaptive composition, precluding the direct extension of FFT numerical accounting to such setting. We provide explicit counterexamples and a general construction of this failure. We show that this failure occurs even when we restrict to subsampled Gaussian mechanisms. (Corollary \ref{cor:SG_counterexample}, Figure \ref{fig:SG_counterexample})
    
    \item We establish necessary and sufficient conditions under which the natural $f$-DP filter becomes valid, characterizing when trade-off-function-based accounting is possible under full adaptivity. (Theorems \ref{thm:validity} and \ref{thm:nochain=>invalid})
    
    \item We prove a fully adaptive $f$-DP central limit theorem, establishing that cumulative privacy losses converge to Gaussian distributions, yielding GDP bounds under full adaptivity. This extends the non-adaptive CLT of~\citet{dong_gaussian_2022} via martingale Berry--Esseen techniques, and suggests a general methodology for constructing approximate GDP filters for any mechanism family admitting closed-form PLRV moment approximations. (Theorem \ref{thm:fully-adaptive-CLT})
    
    \item As a demonstration of our CLT, we instantiate this methodology for 
    subsampled Gaussian mechanisms, constructing a closed-form approximate GDP filter that yields strictly tighter guarantees than fully adaptive RDP accounting in asymptotic regimes $q\ll1$ and $q\approx1$, without tracking the full trade-off function. (Algorithm \ref{alg:approx-GDP-filter})

\end{compactenum}

\para{Related work.}
Privacy filters were introduced by \citet{rogers_privacy_2016}, who formalized fully adaptive composition for $(\varepsilon,\delta)$-DP and established the foundations of filter-based accounting. \citet{feldman_individual_2021} showed that R\'{e}nyi DP admits simple and tight privacy filters under full adaptivity, with the added flexibility of individual privacy accounting. \citet{whitehouse_fully-adaptive_2023} subsequently constructed approximate R\'{e}nyi DP filters achieving essentially no loss in tightness relative to advanced standard composition. In the GDP setting, \citet{koskela_individual_2023} developed an accurate fully adaptive filter exploiting the Gaussian structure. \citet{haney} studied the transfer of filters from the noninteractive to the interactive setting under a general framework that includes $f$-DP, but did not analyze specific filter constructions. Independently and concurrently, \citet{regehr2026} also prove that the natural $f$-DP filter is invalid and provide complementary validity conditions. Beyond this shared result, we additionally prove a fully adaptive $f$-DP central limit theorem that suggests a general methodology for constructing approximate GDP filters via mechanism-specific PLRV moment approximations, of which the subsampled Gaussian filter is one instantiation.

\section{Background}

\minisection{Differential Privacy}
A mechanism \(\M\) is a function that maps a dataset input \(S=\!\left(X_1,\ldots,X_n\right)\) to a distribution \(\M\!\left(S\right)\) on a space \(\O\) of possible outputs. We call \(\left(S,S^-\right)\) remove neighbours if \(S^-\) is obtained by removing one datapoint from \(S\).
\begin{definition}[\citep{dwork2006calibrating}]
    A mechanism \(\M\) is \(\left(\varepsilon,\delta\right)\)-DP if for any remove pair \(S\) and \(S^-\),
    \begin{align*}
        P\!\left[\M\!\left(S\right)\in A\right] \leq e^\varepsilon P\!\left[\M\!\left(S^-\right)\in A\right] + \delta \quad\text{and}\quad
        P\!\left[\M\!\left(S^-\right)\in A\right] \leq e^\varepsilon P\!\left[\M\!\left(S\right)\in A\right] + \delta
    \end{align*} for all set \(A\) in the output space \(\O\). When \(\delta=0\), we shorten \(\left(\varepsilon,0\right)\)-DP into \(\varepsilon\)-DP.
\end{definition}

\begin{definition}[\citep{dong_gaussian_2022}]
    Given distributions \(P,Q\) on a common space \(\O\), the trade-off function is defined as
    \begin{align*}
        T[P,Q](\alpha)=\inf_{\phi:\O\to[0,1]}\left\{1-\E_Q[\phi]:\E_P[\phi]\le \alpha\right\}.
    \end{align*}
\end{definition}
\vspace{-5pt}
When treating the pair \(\left(P,Q\right)\) as a binary testing problem of \(H_0:y\sim P\) vs.~\(H_1:y\sim Q\), the function
\(\phi\) is a randomized rejection rule that maps \({y\in\O}\) to the probability \(\phi\!\left(y\right)\) of accepting \(H_1\). Then, \(\E_P\!\left[\phi\right]\) is the false positive rate (FPR), while \(1-\E_Q\!\left[\phi\right]\) is the false negative rate (FNR), and hence \(T\!\left[P,Q\right]\) represents the optimal FNR at any FPR \(\alpha\in\left[0,1\right]\).
\begin{theorem}[\citep{dong_gaussian_2022}]
    A function \(f:\left[0,1\right]\to\left[0,1\right]\) is the trade-off function of some distribution pair if and only if \(f\!\left(\alpha\right)\leq1-\alpha\), \(f\) is decreasing, convex, continuous.
\end{theorem}
An everywhere smaller trade-off function means a uniformly easier testing problem, and hence weaker privacy protection.
\begin{definition}
    Two trade-off functions \(f_1,f_2\) are Blackwell ordered if \(f_1\geq f_2\) or \(f_1\leq f_2\).
\end{definition}
\begin{definition}[\citep{dong_gaussian_2022}]\label{def:fDP}
    A mechanism \(\M\) is \(f\)-DP where \(f\) is a trade-off function, if \begin{align*}
        T\!\left[\M\!\left(S\right)\!,\M\!\left(S^-\right)\right] \geq f \quad\text{and}\quad
        T\!\left[\M\!\left(S^-\right)\!,\M\!\left(S\right)\right] \geq f
    \end{align*} point-wise for any remove pair \(S,S^-\). In particular, \(\M\) is \(\mu\)-GDP if it is \(G_\mu\)-DP:
    \begin{align*}
        G_\mu = T\!\left[\N\!\left(0,1\right)\!,\N\!\left(\mu,1\right)\right]
            = T\!\left[\N\!\left(0,\nicefrac{1}{\mu^2}\right)\!,\N\!\left(1,\nicefrac{1}{\mu^2}\right)\right].
    \end{align*}
\end{definition}
\vspace{-5pt}
Following \cite{zhu_optimal_2022}, we also consider $f$-DP under remove adjacency, requiring only the first inequality above. Composing directional guarantees before symmetrizing can yield tighter bounds than symmetrizing each step; see Appendix \ref{appx:composition} and \cite{lebeda_pitfalls_2025}.

The subsampled Gaussian trade-off function of a remove pair is \(SG_{q,\mu}^- = \left(SG_{q,\mu}^+\right)^{-1}\) where
\begin{align*}
    SG_{q,\mu}^+\!\left(\alpha\right) = T\!\left[\N\!\left(0,1\right)\!,\left(1-q\right)\!\N\!\left(0,1\right)+q\N\!\left(\mu,1\right)\right] = \left(1-q\right)\!\left(1-\alpha\right) + qG_\mu\!\left(\alpha\right).
\end{align*}
\vspace{-6pt}
\begin{proposition}[\citep{dong_gaussian_2022}]
    A mechanism \(\M\) is \(\left(\varepsilon,\delta\right)\)-DP if and only if it is \(f_{\varepsilon,\delta}\)-DP for
    \begin{align*}
        f_{\varepsilon,\delta}\!\left(\alpha\right) = \max\!\left(0,1-\delta-\e^{\varepsilon}\alpha,\e^{-\varepsilon}\!\left(1-\delta-\alpha\right)\right).
    \end{align*}
\end{proposition}

By deploying \(\left(\varepsilon,\delta\right)\)-DP as \(f_{\varepsilon,\delta}\)-DP, all results for \(f\)-DP immediately apply to \(\left(\varepsilon,\delta\right)\)-DP as well. In particular, the \(f\)-DP standard composition results is expressed via tensor products
\begin{align*}
    T\!\left[P_1,Q_1\right] \otimes T\!\left[P_2,Q_2\right] = T\!\left[P_1\times P_2, Q_1\times Q_2\right],
\end{align*}
from which we can derive the tight \(\left(\varepsilon,\delta\right)\)-DP of standard compositions.
\begin{theorem}[\citep{dong_gaussian_2022}]\label{thm:composition_fDP}
    If \(\M_1\) is tightly \(f_1\)-DP, and the adaptively chosen \(\M_2\) is \(f_2\)-DP where \(f_2\) is fixed and tight, then \(\left(\M_1,\M_2\right)\) is \(f_1\otimes f_2\)-DP. In particular, if \(\M_1\) is \(\mu_1\)-GDP and \(\M_2\) is \(\mu_2\)-GDP, then \(\left(\M_1,\M_2\right)\) is \(\sqrt{\mu_1^2+\mu_2^2}\)-GDP. Moreover, \(f_1\otimes f_2\) can be converted into the tight \(\left(\varepsilon,\delta\right)\)-DP guarantee at each \(\varepsilon>0\),
    \begin{align*}
        \min\!\left\{\delta\geq 0: \left(\M_1,\M_2\right)\,\textup{is}\,\left(\varepsilon,\delta\right)\!\textup{-DP} \right\} = \min\!\left\{\delta\geq 0: f_1\otimes f_2 \geq f_{\varepsilon,\delta} \right\}.
    \end{align*}
\end{theorem}
\begin{definition}
    For distributions \(\left(P,Q\right)\) on a common space \(\O\), denote \(L=\ln\frac{P}{Q}\). The PLRVs are defined as \(L\!\left(y\right)\), \(y\sim P\) and \(L\!\left(y\right)\), \(y\sim Q\). The PLRV density is called the privacy loss density (PLD).
\end{definition}
The cdf of each PLRV characterizes the FPR and FNR of likelihood ratio tests,  directly representing the trade-off function via the Neyman-Pearson lemma. Under compositions, PLRVs sum as $\ln\frac{\M_{1:T}\left(S\right)}{\M_{1:T}\left(S^-\right)} = \sum\ln\!\frac{\M_t\left(S\right)}{\M_t\left(S^-\right)}$, indicating CLT-based convergence to Gaussian distributions.

\minisection{Extended DP Notions} In adaptive composition, the next distributions \(P_2\!\left(y_1\right)\) and \(Q_2\!\left(y_1\right)\) are chosen based on the realization \(y_1\sim P_1\) and \(y_1\sim Q_1\), respectively. We denote then the adaptive product as
\begin{align}
    T\!\left[P_1,Q_1\right] \otimes_{\textup{adapt}} T\!\left[P_2\!\left(y_1\right)\!,Q_2\!\left(y_1\right)\right] = T\!\left[P_1\times P_2\!\left(y_1\right)\!, Q_1\times Q_2\!\left(y_1\right)\right].\label{eq:f-adapt-product}
\end{align}

Let \(\F\) be a family of trade-off functions. The greatest lower bound of $\F$ is not the point-wise infimum, but rather the lower convex envelope of that infimum, denoted $f^\downarrow=\ceinf_{f\in\F} f$. For a single trade-off function \(f\), we define its symmetrization as $\sym(f) = \cemin(f, f^{-1})$.

Not all trade-off functions are Blackwell ordered, so the bound \(f^\downarrow\) of a family \(\F\) can deviate significantly from every \(f\in\F\). \citet{kaissis_beyond_2024} proposed \(\Delta\)-divergence to measure the distance between two arbitrary trade-off functions. We then define convergence of trade-off functions via \(\Delta\)-divergence, and prove its equivalence to other convergence notions (see Theorem \ref{thm:unifying-fconv}).
\begin{definition}
    The lower closure \(\bar{\F}\) of a family \(\F\) of trade-off functions is the set of all limits of any everywhere decreasing sequence in \(\F\). We furthermore define \(\F^{\_\otimes\F}=\left\{f\otimes g:f\in\F,g\in\bar{\F}\right\}\).
\end{definition}
In addition, we use \(\Delta\)-divergence to extend GDP to privacy guarantees that converge to it.
\begin{definition}
    A mechanism \(\M\) is \(\Delta\)-approximately \(\mu\)-GDP if its \(f\)-DP guarantee satisfies
    \begin{align*}
        G_\mu\!\left(\alpha+\Delta\right)-\Delta \leq f\!\left(\alpha\right) \leq G_\mu\!\left(\alpha-\Delta\right)+\Delta \text{ for all }\alpha.
    \end{align*}
\end{definition}

\minisection{Fully Adaptive Composition}
Fully adaptive composition extends standard composition by allowing mechanism choice and privacy parameters to depend arbitrarily on the history of prior outputs. This is crucial in practice: in DP-SGD, for instance, committing to a fixed schedule of $q_t$ and $\sigma_t$ in advance is often suboptimal for utility, yet adapting them based on observed model behaviour or remaining privacy budget requires exactly this framework.

\setlength{\intextsep}{1pt}
\begin{wrapfigure}{r}{0.5\textwidth}
\begin{minipage}{\linewidth}
    \begin{algorithm}[H]
        \caption{Privacy filter in fully adaptive composition}
    \begin{algorithmic}\label{alg:fully-adaptive-composition}
        \STATE\algorithmicinput\  Adaptively chosen mechanisms \(\M_t^{\left(y_{1:t-1}\right)}\!\left(\cdot\right)\) for \(t = 1,\ldots,T\), privacy budget \(\omega_B\), dataset \(S\).
        \FOR{\(t=1\) to \(T\)}
            \STATE \textbf{compute} privacy level \(\omega_t\) of \(\M_t^{\left(y_{1:t-1}\right)}\)
            \IF{\(F\!\left(\omega_B,\omega_1,\ldots,\omega_t\right)=\text{HALT}\)}
                \STATE \textbf{return} \(\left(y_1,\ldots,y_{t-1}\right)\)
            \ENDIF
            \STATE \textbf{compute} \(y_t\gets\M_t^{\left(y_{1:t-1}\right)}\!\left(S\right)\)
        \ENDFOR
        \STATE \textbf{return} \(\left(y_1,\ldots,y_T\right)\)
    \end{algorithmic}
    \end{algorithm}
\end{minipage}
\end{wrapfigure}

Under full adaptivity, privacy parameters become random variables dependent on the stochastic outcomes of previous mechanisms, requiring path-dependent privacy accounting methods where the composition length itself may be random. To manage this complexity, \citet{rogers_privacy_2016} introduced privacy filters and privacy odometers as two complementary approaches. Privacy filters are adaptive stopping rules that halt mechanism releases to ensure pre-specified privacy guarantees, while privacy odometers track cumulative privacy loss without requiring a fixed budget in advance. Since odometers must ensure valid privacy bounds at any analyst-chosen stopping time, they present significant theoretical and practical difficulties. In this work, we focus on privacy filters (Algorithm \ref{alg:fully-adaptive-composition}), which maintain fixed privacy guarantees at the cost of random stopping times.
Here, the privacy level is denoted \(\omega_t\) due to the freedom to choose DP definitions. Central to the composition, the privacy filter \(F\) compares the privacy parameters \(\omega_{1:t}\) realized so far to \mbox{the budget \(\omega_B\)} to determine whether the upcoming mechanism release is allowed. 

\begin{definition}
    A privacy filter \(F\) is valid for a family \(\Omega\) of per-step guarantees ($\omega_t\in\Omega\ \forall t$) if, for every sequence of mechanisms \(\M_1,\ldots,\M_T\), each satisfying its selected guarantee $\omega_t\in\Omega$, and every budget \(\omega_B\), Algorithm \ref{alg:fully-adaptive-composition} produces a composition satisfying the prescribed budget \(\omega_B\).
\end{definition}
Among several valid filters developed in previous work (e.g.~\citealp{feldman_individual_2021}; \citealp{whitehouse_fully-adaptive_2023}), the GDP filter by \citet{koskela_individual_2023} directly adapts the standard GDP composition result to the filtering framework, and provides a foundation for our more general \(f\)-DP results.\negthickspace
\begin{theorem}[\citep{koskela_individual_2023}]\label{thm:filter-GDP}
    With GDP guarantees \(\omega_t=\mu_t\), the following filter is valid (cf.~Theorem \ref{thm:composition_fDP}):
        \begin{align*}
            F\!\left(\mu_B,\mu_1,\ldots,\mu_t\right) = \begin{cases}
                                \textup{CONT} & \text{if }\sum_{i=1}^t \mu_i^2 \leq \mu_B^2,\\
                                \textup{HALT} & \text{otherwise.}                        
                            \end{cases}
        \end{align*}
\end{theorem}

\section{Natural \(f\)-DP Filter}
We first give the formal definition of the natural $f$-DP filter construction.
\begin{definition}\label{def:fDP-filter}
    The natural \(f\)-DP filter is defined as
    \begin{align*}
        F\!\left(f_B,f_1,\ldots,f_t\right) = \begin{cases}
                            \text{CONT} & \text{if } f_1\otimes\ldots\otimes f_t \geq f_B,\\
                            \text{HALT} & \text{otherwise.}                        
                        \end{cases}
    \end{align*}
\end{definition}
This filter keeps track of the selected per-step guarantees \(f_t\) and compares their tensor product with the budget \(f_B\). Throughout this section, $f$-DP guarantees need not be symmetric. For Poisson-subsampled DP-SGD, the filter can be used to track $f$-DP guarantees under remove adjacency: if it successfully guarantees the budget $f_B$, then symmetrizing $f_B$ yields the add/remove guarantee.

Such a procedure is motivated by the standard \(f\)-DP composition theorem (Theorem \ref{thm:composition_fDP}), and it is a direct generalization of the valid GDP filter (Theorem \ref{thm:filter-GDP}). It is therefore tempting to expect validity under full adaptivity as well---which would also allow numerical PLD composition methods \citep{pmlr-v108-koskela20b,NEURIPS2021_6097d8f3,Doroshenko2022ConnectTD} to extend to this setting. We show that this validity claim is false in general.

\minisection{Calibration Point and Filter Invalidity} To understand this failure, we reduce to a two-step adaptive composition, since the general case can be viewed through this setting. The first mechanism has a trade-off of \(f_1\), and the second adaptively chosen mechanism has a trade-off of \(f_2^{(y_1)}\) which can be selected arbitrarily from a family \(\F_2\). The true composition is described by the adaptive product \(f_1\otimes_{\textup{adapt}}f_2^{(y_1)}\), whereas the natural \(f\)-DP filter only tracks the standard product \(f_1\otimes f_2^{(\hat{y}_1)}\) when $y_1=\hat{y}_1$ is realized.

\begin{theorem}[(Adaptive Tightness of Theorem \ref{thm:composition_fDP})]
    Let \(f_1\) be a fixed trade-off function, and let \(\F_2\) be a family of trade-off functions. Denote the greatest lower bound
    \(f_2^\downarrow=\ceinf_{f_2 \in \F_2} f_2\).
    Then, for any selection rule \(y_1\mapsto f_2^{\left(y_1\right)}\in\F_2\), the adaptive product of \(f_1\) and \(f_2^{\left(y_1\right)}\) is lower bounded by \(f_1\otimes f_2^\downarrow\). Moreover, this bound is tight:
    \begin{align*}
        f_1\otimes f_2^\downarrow = \ceinf_{\textup{selection rule}} \left(f_1\otimes_{\textup{adapt}} f_2^{\left(y_1\right)}\right).
    \end{align*}
    More precisely, at each FPR \(\alpha\in\left(0,1\right)\), there exists a selection rule \(y_1\mapsto f_2^{\left(y_1\right)}\) such that the adaptive product \(f_1\otimes_{\textup{adapt}} f_2^{\left(y_1\right)}\) is arbitrarily close to the lower bound \(f_1\otimes f_2^\downarrow\) at \(\alpha\) (but not necessarily at other FPRs).
\end{theorem}

Suppose there exist \(f_{21},f_{22}\in\F_2\) that intersect at an FPR \(\alpha_0\) (a calibration point, see \citealp{kaissis_beyond_2024}). Their lower envelope \(f_2^\downarrow=\cemin\!\left(f_{21},f_{22}\right)\) then falls below both \(f_{21}\) and \(f_{22}\) on some interval containing \(\alpha_0\). Consequently, \(f_1\otimes f_2^\downarrow\) can be strictly smaller than both \(f_1\otimes f_{21}\) and \(f_1\otimes f_{22}\): the former captures all extent of the true effect achievable under adaptive selection over \(\F_2\), whereas the latter are the per-branch bounds monitored by the natural \(f\)-DP filter.
\begin{corollary}\label{cor:SG_counterexample}
    The \(f\)-DP filter is not valid for subsampled Gaussian mechanisms (and is thus not valid in general): There exist \(f_1,f_{21},f_{22}\) that are (tensor products of) subsampled Gaussian trade-off functions, and a budget \(f_B\), such that \(f_1\otimes f_{21}\geq f_B\) and \(f_1\otimes f_{22}\geq f_B\), but \(\left(f_1\otimes \cemin\!\left(f_{21},f_{22}\right)\right)\!\left(\alpha\right) < f_B\!\left(\alpha\right)\) at some FPR \(\alpha\in\left(0,1\right)\). See Figure \ref{fig:SG_counterexample}.
\end{corollary}
\begin{figure}
    \centering
    \includegraphics[width=0.95\linewidth]{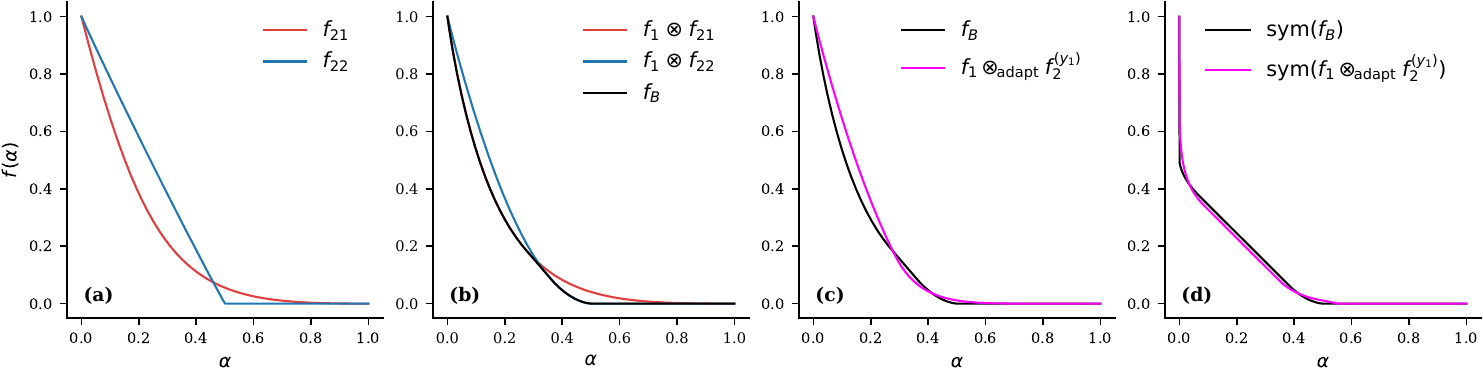}
    \caption{Counterexample showing that the natural \(f\)-DP filter can fail in fully adaptive composition of subsampled Gaussian mechanisms (Corollary~\ref{cor:SG_counterexample}). We fix \(q=0.5\), set \(f_1=SG^-_{q,1.3}\), \(f_{21}=SG^-_{q,2}\otimes SG^-_{q,2}\), \(f_{22}=SG^-_{q,0.1}\otimes SG^-_{q,10}\), and define the budget \(f_B=\cemin\!\left(f_1\otimes f_{21},\,f_1\otimes f_{22}\right)\). \textbf{(a)} The two candidate future trade-off functions cross (not Blackwell-ordered). \textbf{(b)} The budget \(f_B\) is formed as the convex lower envelope of the two fixed-branch products. \textbf{(c)} Output-dependent selection between the continuations produces an adaptive composition below \(f_B\), violating validity. \textbf{(d)} The violation persists after symmetrization into an appropriate \(f\)-DP guarantee.}
    \label{fig:SG_counterexample}
\end{figure}
Such crossings are generic in common mechanism families, including Laplace mechanisms. The absence of crossings for (tensor products of) Gaussian trade-off functions is a notable special case. Additional details and proofs for the material up to this point are deferred to Appendix~\ref{appx:counterexample}.

\minisection{Blackwell Chains and Filter Characterization} By contrast, this issue does not occur for RDP or GDP filters, since their guarantees are scalar-valued and therefore always comparable. This motivates imposing structural assumptions under which the natural \(f\)-DP filter becomes valid.

\begin{definition}[(Blackwell Chain)]\label{def:chain}
    A family \(\F\) of trade-off functions is a Blackwell chain if \(\F\) is totally ordered with respect to the Blackwell order: \(f_1\geq f_2\) or \(f_1\leq f_2\) for any \(f_1,f_2\in\F\).
\end{definition}
For any Blackwell chain \(\F\), the trade-off functions \(f\in\F\) can approach the greatest lower bound \(f^\downarrow=\ceinf_{f\in\F} f\) arbitrarily close FPR-uniformly. Gaussian trade-off functions form such a chain, with \(G_\mu\) decreasing in \(\mu\) and their lower bound approximated uniformly. Thus, the Blackwell chain assumption removes the gap between adaptive products and standard products (as in Corollary \ref{cor:SG_counterexample}) and ensures validity of the \(f\)-DP filter.

\begin{restatable}{thm}{filtervalidity}\label{thm:validity}
    Let \(\F\) be the family of trade-off functions used as per-step \(f\)-DP guarantees. The natural \(f\)-DP filter is valid for \(\F\) in fully adaptive composition of up to \(T\geq2\) steps if \(\F^{\left(\_\otimes\F\right)\left(T-2\right)}\) is a Blackwell chain.
\end{restatable}
The form \(\F^{\left(\_\otimes\F\right)\left(T-2\right)}\) arises from reducing \(T\)-step adaptive composition to a two-step setting, where the first and second trade-off functions are drawn from \(\F\) and \(\F^{\left(\_\otimes\F\right)\left(T-2\right)}\), respectively. The proof (\mbox{Appendix \ref{appx:fDP-filter}}) shows that the Blackwell chain assumption allows one to conclude that if every \(f_2\in\F^{\left(\_\otimes\F\right)\left(T-2\right)}\) satisfies \(f_1\otimes f_2\geq f_B\), then \(f_2^\downarrow\) does as well. This deduction is a generalization of the RDP filter proof in \cite{feldman_individual_2021} and the GDP filter proof in \cite{koskela_individual_2023}, extended from scalar to function objects. As Gaussian trade-off functions are closed under tensor products, Theorem \ref{thm:validity} straightforwardly recovers the valid GDP filter as a special case.

The notion of lower closure \(\bar{\F}\) within \(\F^{\_\otimes\F}\) used above relies on convergence of trade-off functions, for which we adopt \(\Delta\)-divergence \citep{kaissis_beyond_2024} as the underlying metric. As an independent technical contribution, we prove that convergence in \(\Delta\)-divergence is equivalent to pointwise convergence on \((0,1]\) and to uniform convergence on compact subsets within \((0,1]\) (Theorem \ref{thm:unifying-fconv}, Appendix \ref{appx:convergences}), providing a flexible toolkit for theoretical analyses of trade-off function approximations.

We conclude this section by sharpening Corollary~\ref{cor:SG_counterexample} into a general converse statement.
\begin{restatable}{thm}{nochainisinvalid}\label{thm:nochain=>invalid}
    Let \(\F\) be the family of trade-offs used as per-step $f$-DP guarantees. Suppose that:
    \begin{compactenum}
        \item Some $f_1\in\F$ is differentiable on $(\alpha_1,\alpha_2)\subset(0,1)$ and has a nonconstant derivative there;
        \item \(\F^{\left(\_\otimes\F\right)\left(T-2\right)}\) contains two trade-off functions $f_{21}$ and $f_{22}$ that cross at some FPR $\alpha_0\in(0,1)$ with different slopes $f_{21}'(\alpha_0)\neq f_{22}'(\alpha_0)$.
    \end{compactenum} then the natural \(f\)-DP filter is invalid for $\F$: there exists a fully adaptive composition of at most $T$ steps, defined on a single neighboring input pair whose trade-off function at each step belongs to $\F$, that passes every filter check but violates its prescribed budget.\footnote{We thank Matthew Regehr for pointing out a gap in an earlier version of Theorem \ref{thm:nochain=>invalid}.}
\end{restatable}
The first condition in Theorem \ref{thm:nochain=>invalid} includes Laplace and Poisson subsampled Gaussian mechanisms, but does not generally cover discrete-noise mechanisms such as discrete Gaussian. Theorem \ref{thm:validity} and Theorem \ref{thm:nochain=>invalid} together yield a sharp characterisation: 
    \textit{Under regularity assumptions, the natural \(f\)-DP filter is valid  if and only if \(\F^{\left(\_\otimes\F\right)\left(T-2\right)}\) is a Blackwell chain.}

\section{Fully Adaptive $f$-DP Central Limit Theorem}

\minisection{Fully Adaptive Central Limit Theorem} Having established when the natural exact $f$-DP filter is and is not valid, we move on to proving a central limit theorem for fully adaptive privacy loss processes. In contrast to the previous section, we analyze the composition structure through the stochastic process of PLRVs. At each step \(t\), the output \(y_t\) extends the observed history \(y_{1:t-1}\), which in turn governs the choice of the next pair of distributions \(\left(P_{t+1},Q_{t+1}\right)\) and hence the next PLRV \(L_{t+1}\). If a filter halts the composition early, all subsequent PLRVs are set to zero. The total privacy loss is thus the terminal value of a stochastic process — a running sum of sequentially realized log-likelihood ratios, each drawn from a distribution that may depend on all preceding outputs. Despite this adaptive feedback, the cumulative sum remains approximately Gaussian under some control and regularity conditions.
\begin{restatable}{thm}{adaptiveCLT}\label{thm:fully-adaptive-CLT}
    Let \(P_1,\ldots,P_T\) and \(Q_1,\ldots,Q_T\) be adaptive sequences of distributions, where for each \(t\), the choices of \(P_t\) and \(Q_t\) may depend on the past realizations \(y_1,\ldots,y_{t-1}\) drawn under \(P_1,\ldots,P_{t-1}\) and \(Q_1,\ldots,Q_{t-1}\), respectively. Define \(L_t=\ln\frac{P_t}{Q_t}\) as the PLRV at step \(t\), and let \(P=\bigotimes_{t=1}^T P_t\), \(Q=\bigotimes_{t=1}^T Q_t\), and \(y=\left(y_1,\ldots,y_T\right)\). We assume that the following conditions hold with probability 1:

        \begin{compactenum}
            \item \(\left|\sum_{t=1}^T \E_{y\sim P}\!\left[L_t\!\left(y_t\right)\!\mid\!y_{1:t-1}\right] - m_1\right| \leq \eta_1\),
            \item \(\left|\sum_{t=1}^T \E_{y\sim Q}\!\left[L_t\!\left(y_t\right)\!\mid\!y_{1:t-1}\right] + m_2\right| \leq \eta_2\),
            \item \(\left|\sum_{t=1}^T \Var_{y\sim P}\!\left[L_t\!\left(y_t\right)\!\mid\!y_{1:t-1}\right] - v\right| \leq \kappa\),
            \item \(\left|\sum_{t=1}^T \Var_{y\sim Q}\!\left[L_t\!\left(y_t\right)\!\mid\!y_{1:t-1}\right] - v\right| \leq \kappa\),
            \item \(\E_{y\sim P}\!\left[\left|L_t-\E_{y\sim P}\!\left[L_t\!\left(y_t\right)\!\mid\!y_{1:t-1}\right]\right|^3\,\middle|\,y_{1:t-1}\right] \leq \rho\cdot\Var_{y\sim P}\!\left[L_t\!\left(y_t\right)\!\mid\!y_{1:t-1}\right]\) for all \(t\),
            \item \(\E_{y\sim Q}\!\left[\left|L_t-\E_{y\sim Q}\!\left[L_t\!\left(y_t\right)\!\mid\!y_{1:t-1}\right]\right|^3\,\middle|\,y_{1:t-1}\right] \leq \rho\cdot\Var_{y\sim Q}\!\left[L_t\!\left(y_t\right)\!\mid\!y_{1:t-1}\right]\) for all \(t\).
        \end{compactenum}

        Here, we require that \(4\rho^2\leq v\) and \(4\kappa\leq v\). Under these assumptions, the trade-off function \(T\!\left[P,Q\right]\) satisfies the following inequalities for all \(\alpha\):
        \begin{align*}
            G_{\mu+\phi}\!\left(\alpha+\Delta\right) - \Delta \leq T\!\left[P,Q\right]\!\left(\alpha\right) \leq G_{\mu-\phi}\!\left(\alpha-\Delta\right) + \Delta,
        \end{align*}
        where \(\mu=\frac{m_1+m_2}{\sqrt{v}}\), \(\phi=\frac{\eta_1+\eta_2}{\sqrt{v}}\), and \(\Delta=C\!\left(\frac{\rho}{\sqrt{v}}\!\left|\ln\!\frac{\rho}{\sqrt{v}}\right|+\sqrt{\frac{\kappa}{v}}\right)\) for some universal constant \(C\).
\end{restatable}
Our fully adaptive CLT generalizes the non-adaptive CLT of \citet{dong_gaussian_2022}. Instead of relying on the classical Berry--Esseen theorem, the proof's main idea (detailed in Appendix \ref{appx:adaptCLT-approxGDPfilter}) applies the martingale Berry--Esseen theorem of \citet{FAN20191028} to the martingale differences
\begin{align*}
    \frac{L_t\!\left(y_t\right) - \E_{y\sim P}\!\left[L_t\!\left(y_t\right)\!\mid\!y_{1:t-1}\right]}{\sqrt{v}}, y\sim P 
    \quad\textup{and}\quad 
    \frac{L_t\!\left(y_t\right) - \E_{y\sim Q}\!\left[L_t\!\left(y_t\right)\!\mid\!y_{1:t-1}\right]}{\sqrt{v}}, y\sim Q.
\end{align*}
Conditions 5--6 control higher-order moments relative to variance, ensuring no small number of steps dominates privacy loss and allowing a Gaussian approximation with error \(\Delta\) of such form---where the constant \(C\), inherited from \citet{FAN20191028}, is not tracked explicitly as is standard in this literature. Conditions 1--4 ensure path-realized sums of conditional means and variances remain approximately constant, reflecting moment computation from \citet{dong_gaussian_2022}.

More broadly, Conditions 1--4 suggest a general methodology for constructing approximate GDP filters: any mechanism family for which the cumulative conditional means and variances of PLRVs can be approximated analytically admits a closed-form approximate filter via Theorem \ref{thm:fully-adaptive-CLT}. The resulting filter tracks these moment approximations in place of the full trade-off function, with the CLT error \(\Delta\) quantifying the approximation quality.

\minisection{Demonstration: Approximate GDP Filter for Subsampled Gaussian Mechanisms}
We derive closed-form PLRV moment approximations for subsampled Gaussian mechanisms in the asymptotic regimes \(q\ll1\) and \(q\approx1\). For an arbitrary remove pair \(S,S^-\), let \(P=\M\!\left(S\right)\) and \(Q=\M\!\left(S^-\right)\) denote the mechanism outputs, with \(\sigma\) the Gaussian noise and \(\mu\) the norm of the removed datapoint's contribution. The leading-order Taylor expansions give:
{\begin{compactenum}\setlength{\leftskip}{-1em}
    \item For \(q\ll1\), \(\E_{P}\!\left[\ln\frac{P}{Q}\right] \approx -\E_{Q}\!\left[\ln\frac{P}{Q}\right] \approx \frac{1}{2}\Var_{P}\!\left[\ln\frac{P}{Q}\right] \approx \frac{1}{2} \Var_{Q}\!\left[\ln\frac{P}{Q}\right] \approx \frac{1}{2}q^2\!\left(\e^{\frac{\mu^2}{\sigma^2}}-1\right)\).
    \item For \(q\approx1\), \(\E_{P}\!\left[\ln\frac{P}{Q}\right] \approx -\E_{Q}\!\left[\ln\frac{P}{Q}\right] \approx \frac{1}{2}\Var_{P}\!\left[\ln\frac{P}{Q}\right] \approx \frac{1}{2} \Var_{Q}\!\left[\ln\frac{P}{Q}\right] \approx \frac{1}{2}q^2\frac{\mu^2}{\sigma^2}\).
\end{compactenum}}
Crucially, these terms depend on the remove pair only through \(\mu\), making them pair-invariant and enabling a path-wise filter that certifies privacy uniformly over all neighbouring pairs. We denote these per-step approximations \(\operatorname{Budg}(\tilde{q},q,\sigma,\mu)\), with \(\tilde{q}\in\{0,1\}\) indicating the regime. Approximation errors are numerically negligible for \(q<0.2\) and \(q>0.8\); proofs and error bounds are in Appendix~\ref{appx:approximations}.

Algorithm \ref{alg:approx-GDP-filter} instantiates one approximate GDP filter, allowing \(q_t\in(0,\bar{q}]\) with \(\bar{q}\ll1\) (regime \(\tilde{q}=0\)) or \(q_t\in[\bar{q},1]\) with \(\bar{q}\approx1\) (regime \(\tilde{q}=1\)). The function \(\operatorname{invBudg}\) is the algebraic inverse of \(\operatorname{Budg}\) in \(\mu\), returning the clipping level that exactly uses the remaining budget. Since the PLRV mean equals the Rényi divergence of order one, the filter can be interpreted as an approximate order-one RDP-style filter, while providing a final guarantee of approximate GDP.

\begin{figure*}[b]
  \noindent\hrule
  \smallskip
  \begin{minipage}[t]{0.48\linewidth}
    \refstepcounter{algorithm}\label{alg:approx-GDP-filter}
    \addcontentsline{loa}{algorithm}{%
      \protect\numberline{\thealgorithm}%
      Approximate GDP Filter in Fully Adaptive Composition
      of Subsampled Gaussian Mechanisms}
    \noindent\textbf{Algorithm~\thealgorithm:}
    Approximate GDP Filter in Fully Adaptive Composition
    of Subsampled Gaussian Mechanisms
  \end{minipage}
  \hfill\vline\hfill
  \begin{minipage}[t]{0.48\linewidth}
    \refstepcounter{algorithm}\label{alg:individual-approx-GDP-filter}
    \addcontentsline{loa}{algorithm}{%
      \protect\numberline{\thealgorithm}%
      Individual Privacy Accounting for Subsampled Gaussian
      Mechanisms via Approximate GDP Filtering}
    \noindent\textbf{Algorithm~\thealgorithm:}
    Individual Privacy Accounting for Subsampled Gaussian
    Mechanisms via Approximate GDP Filtering
  \end{minipage}
  \smallskip\noindent\hrule
  \begin{algorithmic}
    \STATE\textbf{common input}\ Boolean flag \(\tilde{q}\),
      dataset \(S\).
      Adaptively chosen sequence of objects for \(t=1,\ldots,T\):
      functions \(g_t\!\left(\cdot,y_{1:t-1}\right)\),
      sampling rates \(q_t\!\left(y_{1:t-1}\right)\),
      noise scales \(\sigma_t\!\left(y_{1:t-1}\right)\) in
      \(\left[\bar{\sigma},\infty\right)\),
      clips \(C_t\!\left(y_{1:t-1}\right)\).
  \end{algorithmic}
  \noindent\hrule
  \smallskip
  \begin{minipage}[t]{0.48\linewidth}
    \raggedright
    \parbox[t][0.5\baselineskip][t]{\linewidth}{%
      \textbf{additional input}\ privacy budget \(B\).%
    }
    \begin{algorithmic}
      \FOR{\(t=1\) to \(T\)}
        \STATE \textbf{choose}
          \(g_t,q_t,\sigma_t,C_t\) based on \(y_{1:t-1}\)
        \STATE \algorithmicif\ \(t=1\) \algorithmicthen\
          \(B_t\gets B\)
        \STATE \algorithmicelse\ \(B_t\gets B_{t-1} -
          \operatorname{Budg}\!\left(\tilde{q},q_{t-1},
          \sigma_{t-1},1\right)\)
        \STATE \(C_{B,t}\gets
          C_t\cdot\operatorname{invBudg}\!\left(\tilde{q},q_t,
          \sigma_t,B_t\right)\)
        \STATE \textbf{sample} \(S_t\) from \(S\) at rate \(q_t\)
        \FOR{\(X_j\in S_t\)}
        \STATE \(\grave{g}_{t,j} \gets g_t\!\left(X_j\right)\cdot
            \min\!\left[1,\frac{\min\left(C_t,C_{B,t}\right)}
            {\norm{g_t\left(X_j\right)}}\right]\)
        \ENDFOR
        \STATE \(y_t\gets \sum_{X_j\in S_t}\grave{g}_{t,j} +
          \N\!\left(0,\sigma_t^2C_t^2\mathbf{I}\right)\)
        \STATE \algorithmicif\ \(C_{B,t}\leq C_t\)
          \algorithmicthen\ \textbf{break}
      \ENDFOR
      \STATE \textbf{return} \(y_1,\ldots,y_t\)
    \end{algorithmic}
  \end{minipage}
  \hfill\vline\hfill
  \begin{minipage}[t]{0.48\linewidth}
    \raggedright
    \parbox[t][0.5\baselineskip][t]{\linewidth}{%
      \textbf{additional input}\ dataset size \(n\),
        per-datapoint budgets \(B_j\) for \(j=1,\ldots,n\).%
    }
    \begin{algorithmic}
      \FOR{\(t=1\) to \(T\)}
        \STATE \textbf{choose}
          \(g_t,q_t,\sigma_t,C_t\) based on \(y_{1:t-1}\)
        \FOR{\(X_j\) in \(S\)}
          \STATE \algorithmicif\ \(t=1\) \algorithmicthen\
            \(B_{t,j}\gets B_j\)
          \STATE \algorithmicelse\ \(B_{t,j}\gets B_{t-1,j} -\)
           \newline\hspace*{\fill} \({} \operatorname{Budg}\!\left(\tilde{q},q_{t-1},
                     \sigma_{t-1}C_{t-1},
                     \norm{\grave{g}_{t-1,j}}\right)\)
          \STATE \(C_{B,t,j}\gets
            \operatorname{invBudg}\!\left(\tilde{q},q_t,
            \sigma_tC_t,B_{t,j}\right)\)
          \STATE \(\grave{g}_{t,j} \gets g_t\!\left(X_j\right)
            \cdot\min\!\left[1,
            \frac{\min\left(C_t,C_{B,t,j}\right)}
            {\norm{g_t\left(X_j\right)}}\right]\)
        \ENDFOR
        \STATE \textbf{sample} \(S_t\) from \(S\) at rate \(q_t\)
        \STATE \(y_t\gets \sum_{X_j\in S_t}\grave{g}_{t,j} +
          \N\!\left(0,\sigma_t^2C_t^2\mathbf{I}\right)\)
      \ENDFOR
      \STATE \textbf{return} \(y_1,\ldots,y_T\)
    \end{algorithmic}
  \end{minipage}
  \smallskip\noindent\hrule
\end{figure*}

The following scaling conditions ensure the total privacy spend remains non-degenerate as \(T\) grows: \(T\bar{q}^2=\Theta(1)\) as \(\bar{q}\to0\) for regime \(\tilde{q}=0\), and \(T\bar{\sigma}^{-2}=\Theta(1)\) with \(\hat{q}:=1-\bar{q}=O(\bar{\sigma}^{-2})\) as \(\bar{\sigma}\to\infty\) for regime \(\tilde{q}=1\).
\begin{restatable}{thm}{validapproxGDPfilter}\label{thm:validity-approx-GDP-filter}
Under the above scaling assumptions, the composition of subsampled Gaussian mechanisms produced by Algorithm \ref{alg:approx-GDP-filter} satisfies:
{\begin{compactenum}\setlength{\leftskip}{-1em}
    \item Near \(\tilde q=0\): \(\Delta\)-approx.~\(\bigl(\sqrt{2B}+O(\bar q)\bigr)\)-GDP, with an error \(\Delta=O(\sqrt{\bar q})\) as \(\bar q\to 0\).
    \item Near \(\tilde q=1\): \(\Delta\)-approx.~\(\bigl(\sqrt{2B}+O(\bar\sigma^{-2})\bigr)\)-GDP, with an error \(\Delta=O\!\left(\frac{\ln\bar\sigma}{\bar\sigma}\right)\) as \(\bar\sigma\to\infty\).
\end{compactenum}}
\end{restatable}

Theorem \ref{thm:validity-approx-GDP-filter} is derived by applying Theorem \ref{thm:fully-adaptive-CLT} with \(m_1=m_2=B\) and \(v=2B\) to an arbitrary remove pair (Appendix \ref{appx:adaptCLT-approxGDPfilter}). Verification that the CLT conditions hold in each regime is given by Theorems \ref{thm:approx-q=0} and \ref{thm:approx-q=1}. The errors here vanish as \(\bar{q}\to0\) or \(\bar{\sigma}\to\infty\), so the guarantee approaches exact \(\sqrt{2B}\)-GDP in each regime. Since the RDP-to-$f$-DP conversion is inherently lossy by a nonvanishing factor \citep{zhu_optimal_2022}, appropriate parameter selection yields strictly tighter guarantees than fully adaptive RDP accounting.

The CLT methodology extends naturally to individual privacy accounting. Algorithm \ref{alg:individual-approx-GDP-filter} replaces the single global budget of Algorithm \ref{alg:approx-GDP-filter} with per-individual budgets \(B_{t,j}\) for each datapoint \(X_j\in S\), updating privacy costs using \(\norm{\grave{g}_{t-1,j}}\) rather than the shared worst-case bound \(C_{t-1}\). The budget-aware clipping threshold \(C_{B,t,j}\) is chosen so that the next update cannot overspend \(B_{t,j}\); once depleted, \(C_{B,t,j}=0\) and \(X_j\) contributes no further. In fact, Algorithm \ref{alg:individual-approx-GDP-filter} generalizes the individual accounting of \citet{feldman_individual_2021} from \(q=1\) to both asymptotic regimes \(q\ll1\) and \(q\approx1\).
\begin{restatable}{thm}{individualapproxGDPfilter}
    Suppose the initial individual budgets are uniform, i.e.\ \(B_j = B\) for all \(j=1,\ldots,n\). Then Algorithm \ref{alg:individual-approx-GDP-filter} satisfies identical asymptotic privacy guarantees as stated in Theorem \ref{thm:validity-approx-GDP-filter}.
\end{restatable}



\minisection{Numerical Validation via Monte Carlo Simulation}
We validate Theorem~\ref{thm:validity-approx-GDP-filter} through Monte Carlo simulation of DP-SGD-like fully adaptive compositions with Poisson subsampling. At a fixed subsampling rate $q$, adaptive noise and clipping selections follow an AdaGrad-inspired rule depending on previous outputs, see Appendix \ref{appx:experiments}. We note that we choose and evaluate one sensible adaptive selection policy; searching for the worst-case policy is outside the scope of this paper.

For each fixed $q$ and each add/remove direction, we simulate $10^{10}$ trajectories and compute empirical $(\varepsilon,\delta)$ curves using EVR-style pessimistic estimation \citep{wang_randomized_2023}, with a simultaneous $1-10^{-9}$ confidence band over the $\varepsilon$-grid used. We compare these curves with the $(\varepsilon,\delta)$ curves obtained by converting from the theoretical \(\sqrt{2B}\)-GDP bounds and from the single-order \(\left(1,B\right)\)-RDP filter---the structurally equivalent baseline, as the approximate GDP filter operates by tracking no more than order-one Rényi divergence (see below).

\begin{figure}[!b]
    \centering
    \includegraphics[width=0.875\linewidth]{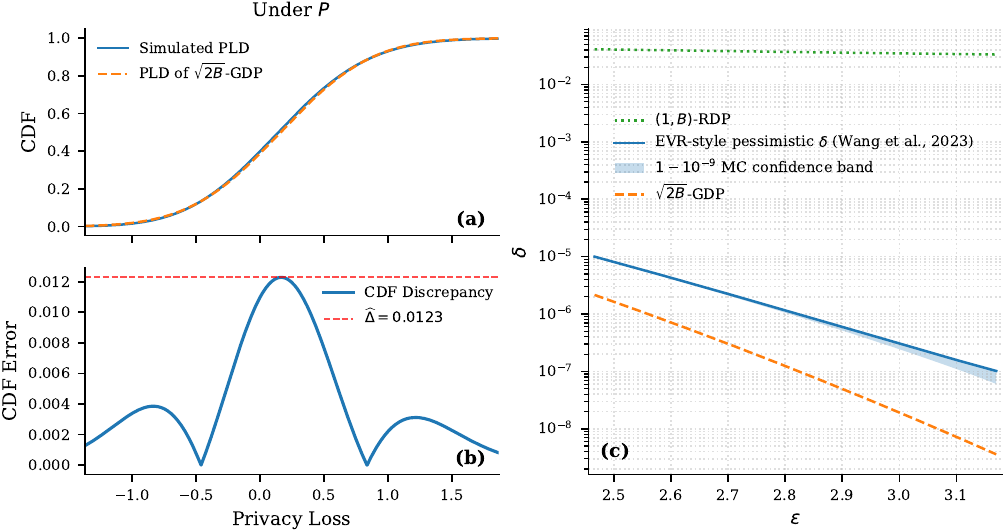}
    \caption{Monte Carlo evaluation of the approximate $\sqrt{2B}$-GDP filter (Theorems \ref{thm:validity-approx-GDP-filter} and \ref{thm:fully-adaptive-CLT}) for DP-SGD-like fully adaptive composition with Poisson subsampling at $q=0.1$. Adaptive noise and clipping follow an Adagrad-inspired rule depending on previous outputs; we simulate $10^{10}$ trajectories of privacy losses (Appendix \ref{appx:experiments}). 
        \textbf{(a)} Empirical CDF of privacy loss vs.~its Gaussian approximation, under $P$.
        \textbf{(b)} Absolute CDF discrepancy, with maximum $\Delta\approx0.0123$, a large deviation from $\sqrt{2B}$-GDP bound.
        \textbf{(c)} Empirical EVR-style pessimistic upper $\delta$ curve \citep{wang_randomized_2023} maximized over add and remove directions, with a simultaneous $1-10^{-9}$ confidence band over the $\varepsilon$-grid, compared with the baseline $(1,B)$-RDP guarantee and the exact $\sqrt{2B}$-GDP used when applying Theorem \ref{thm:validity-approx-GDP-filter} in practice with $\Delta$ omitted. 
    } 
    \label{fig:CLT-validation}
\end{figure}

Figure~\ref{fig:CLT-validation} illustrates Gaussian convergence of the empirical PLD and the corresponding privacy bound comparison. At \(q=0.1\), the \(\sqrt{2B}\)-GDP approximation underestimates privacy loss---a known limitation of CLT-based approximations even in the non-adaptive case \citep{Zheng2020SharpCB,NEURIPS2021_6097d8f3,pmlr-v202-alghamdi23a}---while the \((1,B)\)-RDP bound is notably extremely pessimistic.

\setlength{\intextsep}{0pt}
\begin{wraptable}{r}{0.367\textwidth}
\begin{minipage}{0.367\textwidth}
\vspace{-10pt}
\caption{Empirical error \(\Delta\) and \(\sqrt{2B}\)-GDP 
values across subsampling rates. \(\Delta\) decreases 
as \(q\to0\) or \(q\to1\).}
\label{tab:delta-gdp-conv}
\centering
\begin{tabular}{c c cc c}
\toprule
$q$  & $\Delta$ Error & $\sqrt{2B}$-GDP \\
\midrule
$0.01$   & $0.003533$ & $0.1775$ \\
$0.1$    & $0.01450$ & $0.5612$ \\
$0.199$  & $0.02584$ & $0.7897$ \\
$0.801$  & $0.01461$ & $1.416$ \\
$0.95$   & $0.004735$ & $1.454$ \\
\bottomrule
\end{tabular}
\end{minipage}
\vspace{-10pt}
\end{wraptable}

The tightness over RDP has a clear interpretation: since the PLRV mean equals the order-one Rényi divergence, the approximate GDP filter effectively operates as a single order-one RDP filter. Compared to more restrictive multi-order RDP filters, this single-order structure can permit more mechanisms to proceed, while yielding tighter guarantees than the $(1,B)$-RDP guarantee given by the single-order RDP filter---though at the cost of validity outside the asymptotic regimes.

Table \ref{tab:delta-gdp-conv} shows smaller $\Delta$ toward the low- and high-subsampling regimes, consistent with Theorem \ref{thm:validity-approx-GDP-filter}. The theorem further shows that this error can be made arbitrarily small by selecting sufficiently extreme $q$, making the $\sqrt{2B}$-GDP approximation arbitrarily accurate. Together with the improvement over the matched $(1,B)$-RDP guarantee, this demonstrates that the slack in RDP stems from its lossy Rényi divergence representation rather than any fundamental barrier in adaptive composition. Whether tighter multi-order-equivalent approximate filters---constructed via the same CLT framework by tracking higher-order PLRV moments---can close the remaining gap at moderate subsampling rates is a natural direction for future work.

\section{Discussion and Conclusions}
This work resolves a central open question in fully adaptive privacy accounting for $f$-DP. We show that the natural $f$-DP filter is invalid under full adaptivity and establish necessary and sufficient conditions for validity---a sharp characterization that explains why divergence-based notions such as RDP admit simple filters while trade-off-function-based notions do not. On the constructive side, we prove a fully adaptive $f$-DP central limit theorem, establishing Gaussian convergence of cumulative privacy losses under full adaptivity, suggesting a general methodology for constructing approximate GDP filters via mechanism-specific PLRV moment approximations. As one instantiation, we construct a closed-form approximate GDP filter for subsampled Gaussian mechanisms that, in asymptotic regimes, yields valid guarantees strictly tighter than fully adaptive RDP accounting without tracking the full trade-off function---demonstrating that the slack in RDP stems from its lossy Rényi divergence representation rather than any fundamental barrier in adaptive composition. 

\section*{Acknowledgements}
\addcontentsline{toc}{section}{Acknowledgements}
We would like to thank Matthew Regehr for identifying a gap in an earlier version of Theorem \ref{thm:nochain=>invalid}, which prompted the revised assumptions and proof. This research was conducted in collaboration with Nokia Bell Labs, as part of Nokia's Veturi program. The program is supported by funding from Business Finland, whose contribution is gratefully acknowledged. This work was also supported by the Research Council of Finland (Flagship programme: Finnish Center for Artificial Intelligence, FCAI, Grant 356499 and Grant 359111), the Strategic Research Council at the Research Council of Finland (Grant 358247), and the European Union (Project 101070617). Views and opinions expressed are however those of the author(s) only and do not necessarily reflect those of the European Union or the European Commission. Neither the European Union nor the granting authority can be held responsible for them. The authors acknowledge the research environment provided by ELLIS Institute Finland.

\para{Use of LLMs and AI agents.} The authors developed the research and wrote the initial manuscript, with LLM assistance limited to language polishing. During subsequent revisions, LLMs and coding agents assisted with implementing and refining experiments, discussing and checking mathematical arguments, literature searches, and manuscript revision. The authors determined the scientific direction and final decisions and take responsibility for the paper.

\bibliography{references}
\bibliographystyle{abbrvnat}
\addcontentsline{toc}{section}{References}


\newpage
\appendix
\onecolumn

\tableofcontents
\newpage

\addcontentsline{toc}{section}{Appendix}

\section{Additional Background: Privacy Profile}\label{appx:privacy_profile}
\subsection{Privacy Profile vs.~Trade-off Function}
The results here are standard; proofs can be found in \citet{dong_gaussian_2022}, \citet{zhu_optimal_2022}, \citet{kaissis_beyond_2024}. The hockey-stick divergence is an integral that directly captures \(\left(\varepsilon,\delta\right)\)-DP bounds. Define \(\left[x\right]_+ = \max\!\left(x,0\right)\).
\begin{definition}
    The hockey-stick divergence at order \(\gamma\geq0\) of a pair of distributions \(P,Q\) with densities \(p,q\) is
    \begin{align*}
        D_\gamma\!\left[P\midd Q\right] = \int_\O \left[p\!\left(y\right) - \gamma q\!\left(y\right)\right]_+ \d y
        = \max_{A\subset\O} P\!\left(A\right) - \gamma Q\!\left(A\right).
    \end{align*}
    A mechanism \(\M\) is \(\left(\varepsilon,\delta\right)\)-DP if \(D_{e^\varepsilon}\!\left[\M\!\left(S\right)\midd\M\!\left(S^-\right)\right]\leq\delta\) and \(D_{e^\varepsilon}\!\left[\M\!\left(S^-\right)\midd\M\!\left(S\right)\right]\leq\delta\) for all remove pair \(S,S^-\).
\end{definition}
The privacy profile on a mechanism level defined by \citet{Balle_Barthe_Gaboardi_2020}  is given as Definition \ref{def:tight-H} below. We first give a definition on the distribution level.
\begin{definition}
    The privacy profile of a pair of distributions \(\left(P,Q\right)\) is a function \(H:\left[0,\infty\right)\to\left[0,1\right]\) that represents the hockey-stick divergence \(H\!\left(\gamma\right)=D_\gamma\!\left[P\midd Q\right]\) of all orders \(\gamma\geq0\).
\end{definition}
\begin{theorem}\label{thm:H}
    A function \(H:\left[0,\infty\right)\to\left[0,1\right]\) is the privacy profile of some pair of distributions if and only if \(H\!\left(0\right)=1\), \(H\) is decreasing, convex, and \(H\!\left(\gamma\right)\geq\left[1-\gamma\right]_+\).
\end{theorem}
We can describe the spectrum of all \(\left(\varepsilon,\delta\right)\)-DP bounds using privacy profiles.
\begin{definition}\label{def:dominating}
    A mechanism \(\M\) is dominated by a privacy profile \(H\) if for all remove pair \(S,S^-\),
    \begin{align*}
        D_\gamma \!\left[\M\!\left(S\right)\midd\M\!\left(S^-\right)\right] \leq H\!\left(\gamma\right) \quad\text{and}\quad D_\gamma \!\left[\M\!\left(S^-\right)\midd\M\!\left(S\right)\right] \leq H\!\left(\gamma\right)
    \end{align*}
    at all \(\gamma\geq0\). \(\M\) is dominated by a privacy profile \(H\) under remove adjacency if only the first inequality holds. \(\M\) is dominated by a pair of distributions \(\left(P,Q\right)\) (under remove adjacency) if it is dominated by the privacy profile of \(\left(P,Q\right)\) (under remove adjacency).
\end{definition}
For a family of privacy profiles \(\H\), the point-wise supremum \(H^\uparrow\) of all \(H\in\H\) is another privacy profile by Theorem \ref{thm:H}, and therefore \(H^\uparrow\) is the smallest upper bound of \(\H\).
\begin{definition}\label{def:tight-H}
    The privacy profile of a mechanism \(\M\) is the one that tightly dominates \(\M\):
    \begin{align*}
        H_\M\!\left(\gamma\right) = \sup_{S,S^-\,\textup{remove pair}} \max\!{\left(D_\gamma\!\left[\M\!\left(S\right)\midd\M\!\left(S^-\right)\right]\!, D_\gamma\!\left[\M\!\left(S^-\right)\midd\M\!\left(S\right)\right]\right)}.
    \end{align*}
    \(H_\M\) represents tight \(\delta\)-bounds in \(\left(\varepsilon,\delta\right)\)-DP for all \(\varepsilon\in\real\). In addition, the privacy profile of \(\M\) under remove adjacency is
    \begin{align*}
        H_\M^-\!\left(\gamma\right) = \sup_{S,S^-\,\textup{remove pair}} D_\gamma\!\left[\M\!\left(S\right)\midd\M\!\left(S^-\right)\right].
    \end{align*}
\end{definition}

In fact, tight \(f\)-DP is equivalent to a privacy profile of \mbox{\(\left(\epsilon,\delta\right)\)-DP} bounds. This relationship justifies the conventional reason why the privacy profile includes the parameters \mbox{\(\gamma=e^\varepsilon\in\left(0,1\right)\)}, or \(\varepsilon<0\).
\begin{theorem}\label{thm:f<->H}
    There is a bijection between the trade-off functions and the privacy profiles. Specifically, let \(f\) and \(H\) be the trade-off function and privacy profile of \(\left(P,Q\right)\) respectively, then
    \begin{enumerate}
        \item \(H\!\left(\gamma\right) = 1+\left(f^{-1}\right)^\star\!\left(-\gamma\right)\) for all \(\gamma\geq0\),
        \item \(f^{-1}\!\left(\alpha\right) = 1+H^\star\!\left(-\alpha\right)\) for all \(\alpha\in\left[0,1\right]\),
    \end{enumerate}
    where $^\star$ denotes the convex conjugate operation and \(f^{-1}=T\!\left[Q,P\right]\) is the inverse function of \(f\) (see Proposition \ref{prop:symmetry-distribution-level}).
\end{theorem}

Furthermore, this bijection carries over the Blackwell ordering.
\begin{theorem}\label{thm:order-f-H}
    Let \(f_1,f_2\) be trade-off functions, and let \(H_1,H_2\) be the corresponding privacy profiles respectively. Then the following are equivalent.
    \begin{enumerate}
        \item \(f_1\!\left(\alpha\right)\geq f_2\!\left(\alpha\right)\) for all \(\alpha\in\left[0,1\right]\).
        \item \(H_1\!\left(\gamma\right)\leq H_2\!\left(\gamma\right)\) for all \(\gamma\geq0\).
    \end{enumerate}
\end{theorem}
Similar to trade-off functions, we say that the privacy profiles \(H_1\) and \(H_2\) are Blackwell ordered if \(H_1\leq H_2\) or \(H_1\geq H_2\) point-wise. In light of Theorem \ref{thm:f<->H} and Theorem \ref{thm:order-f-H}, Definition \ref{def:fDP} of \(f\)-DP is equivalent to Definition \ref{def:dominating} of the dominating privacy profile.

For a family of trade-off functions \(\F\), their point-wise infimum does not maintain convexity in general, and hence it is not necessarily a trade-off function. Instead, the greatest lower bound \(f^\downarrow\) of all \(f\in\F\) can be obtained by first converting \(\F\) to the corresponding family of privacy profiles \(\H\) by using \mbox{Theorem \ref{thm:f<->H}}, taking the point-wise supremum \(H^\uparrow\) and converting \(H^\uparrow\) to \(f^\downarrow\) using Theorem \ref{thm:f<->H} again. This is equivalently the lower convex envelope of the point-wise infimum of \(\F\).
\begin{proposition}\label{prop:symmetry-distribution-level}
    For a pair of distributions \(\left(P,Q\right)\),
    \begin{enumerate}
        \item \(T\!\left[Q,P\right]\) is the reflection of \(T\!\left[P,Q\right]\) across the identity line.
        \item For all \(\gamma>0\),
        \begin{align*}
            D_\gamma\!\left[Q\midd P\right] = 1 - \gamma + \gamma D_{\frac{1}{\gamma}}\!\left[P\midd Q\right].
        \end{align*}
    \end{enumerate}
\end{proposition}
\begin{corollary}\label{cor:symmetry-f-H}
    Let \(H_\M\) be the privacy profile of a mechanism \(\M\), and let \(f_\M\) be the corresponding tightest trade-off function such that \(\M\) is \(f_\M\)-DP. Then the following holds.
    \begin{enumerate}
        \item The graph of \(f_\M\!\left(\alpha\right)\) is symmetric w.r.t.~the identity line.
        \item For all \(\gamma>0\), \(H_\M\!\left(\gamma\right) = 1 - \gamma + \gamma H_\M\!\left(\frac{1}{\gamma}\right)\).
    \end{enumerate}
\end{corollary}
A privacy profile \(H_\M^-\) of \(\M\) under remove adjacency implies that the privacy profile \(H_\M\) is
\begin{align*}
    H_\M\!\left(\gamma\right) = \max\!\left(H_\M^-\!\left(\gamma\right),1 - \gamma + \gamma H_\M^-\!\left(\frac{1}{\gamma}\right)\right).
\end{align*}
It is therefore sufficient to analyze privacy profiles under remove adjacency.
\begin{definition}
    A privacy profile \(H\) is symmetric if \(H\!\left(\gamma\right) = 1-\gamma + \gamma H\!\left(\frac{1}{\gamma}\right)\).
\end{definition}
    For an arbitrary privacy profile \(H\), denote \(\hat{H}\!\left(\gamma\right)=1-\gamma + \gamma H\!\left(\frac{1}{\gamma}\right)\) and \(\sym\!\left(H\right) = \max\!\left(H,\hat{H}\right)\). Note that if \(H\) is symmetric, then \(H = \hat{H} = \sym\!\left(H\right)\).

\subsection{Tensor Product of Trade-off Functions/Privacy Profiles}\label{appx:composition}

We give a self-contained introduction to tensor products of trade-off functions and privacy profiles. Building on \cite{dong_gaussian_2022}, we present the theory from the privacy-profile perspective.

\citet{zhu_optimal_2022} showed that the dominating pairs of distributions transform nicely under standard compositions.
\begin{theorem}\label{thm:product_dominating}
    If \(\M_1\) is dominated by \(\left(P_1,Q_1\right)\), and the adaptively chosen \(\M_2\) is dominated by a fixed pair \(\left(P_2,Q_2\right)\), then \(\left(\M_1,\M_2\right)\) is dominated by \(\left(P_1\times P_2,Q_1\times Q_2\right)\).
\end{theorem}

This motivates the definition of the tensor product of privacy profiles, or equivalently, trade-off functions.
\begin{definition}\label{def:H-product}
    Let \(H_1,H_2\) be privacy profiles. The tensor product \(H_1\otimes H_2\) is defined as
    \begin{align*}
        \left(H_1\otimes H_2\right)\!\left(\gamma\right) = D_\gamma\!\left[P_1\times P_2 \midd Q_1\times Q_2\right],
    \end{align*}
    where \(\left(P_1,Q_1\right)\) is any pair of distributions that has a privacy profile of \(H_1\), and \(\left(P_2,Q_2\right)\) is any pair of distributions that has a privacy profile of \(H_2\). Using Fubini's theorem,
    \begin{align*}
        \left(H_1\otimes H_2\right)\!\left(\gamma\right) 
        &= \int_{\O_1\times\O_2} \left[
            p_1\!\left(y_1\right)p_2\!\left(y_2\right)
            - \gamma q_1\!\left(y_1\right)q_2\!\left(y_2\right)
        \right]_+ \d y_1\d y_2\\
        &= \int_{\O_1} H_2\!\left(\gamma\frac{q_1\!\left(y_1\right)}{p_1\!\left(y_1\right)}\right) p_1\!\left(y_1\right)\d y_1.
    \end{align*}
\end{definition}
\begin{proposition}\label{prop:properties-H-product}
    The tensor product of privacy profiles satisfies all following properties.
    \begin{enumerate}\setcounter{enumi}{-1}
        \item It is well-defined.
        \item It is commutative and associative.
        \item\label{propitem:H-product-ordering} It carries over the Blackwell ordering: If \(H_{21}\geq H_{22}\), then \(H_1\otimes H_{21}\geq H_1\otimes H_{22}\).
        \item\label{propitem:H-product-with-zero} \(H\otimes H_{\textup{Id}} = H_{\textup{Id}}\otimes H = H\), where \(H_{\textup{Id}}\!\left(\gamma\right)=\left[1-\gamma\right]_+\) is the perfectly private privacy profile.
        \item \(\widehat{\left(H_1\otimes H_2\right)} = \hat{H}_1\otimes\hat{H}_2\).
    \end{enumerate}
\end{proposition}
\begin{proof}
    \textit{Proof for 0}. We need to show that the resulting privacy profile \(H_1\otimes H_2\) is independent of the distribution pair \(\left(P_1,Q_1\right)\) and \(\left(P_2,Q_2\right)\) behind \(H_1\) and \(H_2\). It is sufficient to show that property \ref{propitem:H-product-ordering} holds: By property \ref{propitem:H-product-ordering}, \(D_\gamma\!\left[P_{21},Q_{21}\right] = D_\gamma\!\left[P_{22},Q_{22}\right]\) for all \(\gamma\) implies that \(D_\gamma\!\left[P_1\times P_{21},Q_1\times Q_{21}\right] = D_\gamma\!\left[P_1\times P_{22},Q_1\times Q_{22}\right]\) for all \(\gamma\).

    \textit{Proof for 2}. Using Fubini's theorem,
    \begin{multline*}
        \left(H_1\otimes H_{21}\right)\!\left(\gamma\right) = \int_{\O_1} H_{21}\!\left(\gamma\frac{q_1\!\left(y_1\right)}{p_1\!\left(y_1\right)}\right) p_1\!\left(y_1\right)\d y_1 \\ \geq \int_{\O_1} H_{22}\!\left(\gamma\frac{q_1\!\left(y_1\right)}{p_1\!\left(y_1\right)}\right) p_1\!\left(y_1\right)\d y_1 = \left(H_1\otimes H_{22}\right)\!\left(\gamma\right).
    \end{multline*}

    \textit{Proof for 1}. \(H_1\otimes H_2 = H_2\otimes H_1\) follows from 
    \begin{align*}
        p_1\!\left(y_1\right)p_2\!\left(y_2\right)&=p_2\!\left(y_2\right)p_1\!\left(y_1\right), \\
        q_1\!\left(y_1\right)q_2\!\left(y_2\right)&=q_2\!\left(y_2\right)q_1\!\left(y_1\right).
    \end{align*}
    \(\left(H_1\otimes H_2\right)\otimes H_3 = H_1\otimes\left(H_2\otimes H_3\right)\) follows from the Fubini's theorem combined with
    \begin{align*}
        p_1\!\left(y_1\right)p_2\!\left(y_2\right)p_3\!\left(y_3\right)&=p_1\!\left(y_1\right)\!\left[p_2\!\left(y_2\right)p_3\!\left(y_3\right)\right], \\
        q_1\!\left(y_1\right)q_2\!\left(y_2\right)q_3\!\left(y_3\right)&=q_1\!\left(y_1\right)\!\left[q_2\!\left(y_2\right)q_3\!\left(y_3\right)\right].
    \end{align*}

    \textit{Proof for 3}. \(H_\textup{Id}\) is the privacy profile of any pair of identical distributions \(\left(P_2,P_2\right)\). Therefore
    \begin{multline*}
        \left(H_1\otimes H_2\right)\!\left(\gamma\right) 
        = \int_{\O_1\times\O_2} \left[
            p_1\!\left(y_1\right)p_2\!\left(y_2\right)
            - \gamma q_1\!\left(y_1\right)p_2\!\left(y_2\right)
            \right]_+ \d y_1\d y_2 \\
        = \int_{\O_1\times\O_2} \left[
            p_1\!\left(y_1\right)
            - \gamma q_1\!\left(y_1\right)
            \right]_+ p_2\!\left(y_2\right) \d y_1\d y_2
        = \int_{\O_1} \left[
            p_1\!\left(y_1\right)
            - \gamma q_1\!\left(y_1\right)
            \right]_+ \d y_1 = H_1\!\left(\gamma\right).
    \end{multline*}

    \text{Proof for 4}. By the definition of tensor product and Proposition \ref{prop:symmetry-distribution-level},
    \begin{align*}
        \widehat{\left(H_1\otimes H_2\right)}\!\left(\gamma\right)
        = D_\gamma\!\left[Q_1\times Q_2 \midd P_1\times P_2\right]
        = \left(\hat{H}_1\otimes\hat{H}_2\right)\!\left(\gamma\right).
    \end{align*}
\end{proof}
Properties \ref{propitem:H-product-ordering} and \ref{propitem:H-product-with-zero} from Proposition \ref{prop:properties-H-product} together imply that \(H_1\otimes H_2\geq H_1\), i.e.~tensor product weakens privacy protections. The following Theorem \ref{thm:H_product_dominating} is equivalent to Theorem \ref{thm:product_dominating}
\begin{theorem}\label{thm:H_product_dominating}
    If \(\M_1\) is dominated by \(H_1\), and the adaptively chosen \(\M_2\) is dominated by a fixed \(H_2\), then \(\left(\M_1,\M_2\right)\) is dominated by \(H_1 \otimes H_2\).
\end{theorem}

Alternatively, we can define the tensor product of trade-off functions. We list results by \citet{dong_gaussian_2022}, which are equivalent to the corresponding versions in terms of privacy profiles presented above.
\begin{definition}
    Let \(f_1,f_2\) be trade-off functions. The tensor product \(f_1\otimes f_2\) is defined as
    \begin{align*}
        f_1\otimes f_2 = T\!\left[P_1\times P_2,Q_1 \times Q_2\right],
    \end{align*}
    where \(\left(P_1,Q_1\right)\) is any pair of distributions that satisfies \(f_1=T\!\left[P_1,Q_1\right]\), and \(\left(P_2,Q_2\right)\) is any pair of distributions that satisfies \(f_2=T\!\left[P_2,Q_2\right]\).
\end{definition}
\begin{proposition}\label{prop:properties-f-product}
    The tensor product of trade-off functions satisfies all the following properties.
    \begin{enumerate}\setcounter{enumi}{-1}
        \item It is well-defined.
        \item It is commutative and associative.
        \item\label{propitem:f-product-ordering} It carries over the Blackwell ordering: If \(f_{21}\leq f_{22}\), then \(f_1\otimes f_{21}\leq f_1\otimes f_{22}\).
        \item\label{propitem:f-product-with-zero} \(f\otimes \textup{Id} = \textup{Id}\otimes f = f\), where \(\textup{Id}\!\left(\alpha\right)=\left[1-\alpha\right]_+\) is the perfectly private privacy profile.
        \item \(\left(f_1\otimes f_2\right)^{-1} = f_1^{-1}\otimes f_2^{-1}\).
    \end{enumerate}
\end{proposition}
\begin{proof}
    The proposition is equivalent to Proposition \ref{prop:properties-H-product}.
\end{proof}
\begin{theorem}\label{thm:f_product_dominating}
    If \(\M_1\) is \(f_1\)-DP, and the adaptively chosen \(\M_2\) is \(f_2\)-DP with a fixed \(f_2\), then \(\left(\M_1,\M_2\right)\) is \(f_1\otimes f_2\)-DP.
\end{theorem}

The standard adaptive composition results in GDP follow by using a simple relationship.
\begin{proposition}\label{prop:product-G}
    Gaussian trade-off functions satisfy \(G_{\mu_1} \otimes G_{\mu_2} = G_{\sqrt{\mu_1^2+\mu_2^2}}\).
\end{proposition}
\begin{proof}
    The hockey-stick divergence is an \(F\)-divergence, constructed using the convex function \(\left[x-\gamma\right]_+ = \max\!\left(x-\gamma,0\right)\). Therefore, it satisfies the data-processing inequality: for any randomized function \(\mathcal{T}\), \(D_\gamma\!\left[\mathcal{T}\!\left(P\right)\!,\mathcal{T}\!\left(Q\right)\right] \leq D_\gamma\!\left[P,Q\right]\) for all \(\gamma\). In particular, if \(\mathcal{T}\) is a deterministic bijection, then \(D_\gamma\!\left[\mathcal{T}\!\left(P\right)\!,\mathcal{T}\!\left(Q\right)\right] = D_\gamma\!\left[P,Q\right]\) for all \(\gamma\), and equivalently \(T\!\left[\mathcal{T}\!\left(P\right)\!,\mathcal{T}\!\left(Q\right)\right] = T\!\left[P,Q\right]\) by Theorem \ref{thm:f<->H}.

    For Gaussian trade-off functions \(G_{\mu_1}\) and \(G_{\mu_2}\), we apply the rotation transformation as a bijection,
    \begin{align*}
        G_{\mu_1}\otimes G_{\mu_2} 
        = T\!\left[\N\!\left(0,\mathbf{I}_2\right)\!,\N\!\left(\left(\mu_1,\mu_2\right)\!,\mathbf{I}_2\right)\right] 
        &= T\!\left[\N\!\left(0,\mathbf{I}_2\right)\!,\N\!\left(\left(\sqrt{\mu_1^2+\mu_2^2},0\right)\!,\mathbf{I}_2\right)\right] \\
        &= T\!\left[\N\!\left(0,1\right)\!,\N\!\left(\sqrt{\mu_1^2+\mu_2^2},1\right)\right]
        = G_{\sqrt{\mu_1^2+\mu_2^2}}.
    \end{align*}
\end{proof}
\begin{corollary}\label{cor:product-GDP}
    If \(\M_1\) is \(\mu_1\)-GDP, and the adaptively chosen \(\M_2\) is \(\mu_2\)-GDP with a fixed \(\mu_2\), then \(\left(\M_1,\M_2\right)\) is \(\sqrt{\mu_1^2+\mu_2^2}\)-GDP.
\end{corollary}
The following Lemma \ref{lem:accounting-per-pair} shows that performing analysis per pair---applying tensor products on the pair's trade-off functions before taking the infimum---can yield tighter bounds than the standard composition Theorem \ref{thm:f_product_dominating}, which directly composes the mechanism-level privacy bounds.
\begin{lemma}\label{lem:accounting-per-pair}
    Let \(\F_1=\left\{f_{1,i}\right\}_{i\in I}\) and \(\F_2=\left\{f_{2,i}\right\}_{i\in I}\) be families of trade-off functions with common indexing by \(i\in I\). Then
    \begin{align*}
        \textup{ce}\inf_{i\in I} f_{1,i}\otimes f_{2,i} \geq \left(\textup{ce}\inf_{i\in I}f_{1,i}\right) \otimes \left(\textup{ce}\inf_{i\in I}f_{2,i}\right),
    \end{align*}
    where \(\textup{ce}\inf\) denotes the convex envelope of the point-wise infimum. There exists an example in which the inequality is strict.
\end{lemma}
\begin{proof}
    For all \(i\in I\), \(f_{1,i}\geq\textup{ce}\inf_{i\in I}f_{1,i}\) and \(f_{2,i}\geq\textup{ce}\inf_{i\in I}f_{2,i}\). By Property 2 of Proposition~\ref{prop:properties-f-product}, \(f_{1,i}\otimes f_{2,i}\geq\left(\textup{ce}\inf_{i\in I}f_{1,i}\right) \otimes \left(\textup{ce}\inf_{i\in I}f_{2,i}\right)\) for all \(i\in I\), and hence we can take the \(\textup{ce}\inf\) of the left hand side.

    An example in which the inequality is strict: For \(\mu\in\left[0,1\right]\), let \(f_{1,\mu}=G_\mu\) and \(f_{2,\mu}=G_{1-\mu}\), then
    \begin{align*}
        \textup{ce}\inf_{i\in I} f_{1,i}\otimes f_{2,i} = G_{1} \geq G_{\sqrt{2}} = \left(\textup{ce}\inf_{i\in I}f_{1,i}\right) \otimes \left(\textup{ce}\inf_{i\in I}f_{2,i}\right).
    \end{align*}
\end{proof}
\begin{corollary}\label{cor:accounting-per-pair}
    For arbitrary trade-off functions \(f_1,f_2\), \(\sym\!\left(f_1\otimes f_2\right)\geq\sym\!\left(f_1\right)\otimes\sym\!\left(f_2\right)\).
\end{corollary}

\para{Adaptive Tensor Product.} Translating the adaptive tensor product \eqref{eq:f-adapt-product} via Theorem~\ref{thm:f<->H} and Definition~\ref{def:H-product}, the adaptive tensor product of privacy profiles takes the form
\begin{align*}
    \bigl(H_1 \otimes_{\textup{adapt}} H_2^{\left(y_1\right)}\bigr)(\gamma)
    = \int_{\O_1} H_2^{\left(y_1\right)}\!\left(\gamma\,\frac{q_1(y_1)}{p_1(y_1)}\right) p_1(y_1)\,\d y_1,
\end{align*}
where the selection rule $H_2^{\left(y_1\right)}$ picks a privacy profile for the second mechanism based on the first output $y_1$. Analogously to Property~\ref{propitem:H-product-ordering} in Proposition~\ref{prop:properties-H-product}, if $H_{21}^{\left(y_1\right)} \leq H_{22}^{\left(y_1\right)}$ for every $y_1$, then 
\begin{align*}
    H_1 \otimes_{\textup{adapt}} H_{21}^{\left(y_1\right)} \leq H_1 \otimes_{\textup{adapt}} H_{22}^{\left(y_1\right)}
\end{align*}
immediately by Fubini's theorem. However, ordering in the first component behaves very differently.

\begin{remark}
    The adaptive tensor product is not well-defined as an operation on privacy profiles/trade-off functions alone. Unlike the standard tensor product, the implication
    \begin{align*}
        H_{11} \leq H_{12} \;\not\Rightarrow\; H_{11} \otimes_{\textup{adapt}} H_2^{\left(y_1\right)} \leq H_{12} \otimes_{\textup{adapt}} H_2^{\left(y_1\right)}
    \end{align*}
    fails in general. The reason is that the selection rule maps each first-mechanism output $y_1$ to a second-mechanism profile via the likelihood ratio $q_1(y_1)/p_1(y_1)$, so the result depends on how $(P_1, Q_1)$ is organised over $y_1$, not merely on the profile $H_1$ it induces. Two distribution pairs sharing the same privacy profile can therefore yield different adaptive tensor products under the same selection rule.
    To see this concretely, let $(P_{11}, Q_{11})$ be $0.1$-GDP and $(P_{12}, Q_{12})$ be $0.3$-GDP, so that $H_{11} \leq H_{12}$. Fix the selection rule
    \begin{align*}
        H_2^{(y_1)} = \begin{cases} H_{G_{0.01}} & \text{if } y_1 < 0.3, \\ H_{G_2} & \text{otherwise,} \end{cases}
    \end{align*}
    where $H_{G_\mu}$ denotes the privacy profile corresponding to the trade-off function $G_\mu$ of $\mu$-GDP via Theorem~\ref{thm:f<->H}. Figure~\ref{fig:adapt_prod_notdefined} plots $H_{11} \otimes_{\textup{adapt}} H_2^{(y_1)}$ and $H_{12} \otimes_{\textup{adapt}} H_2^{(y_1)}$: the two curves cross, so neither dominates the other pointwise despite the first profiles being Blackwell ordered.
    
    Nevertheless, all results on adaptive composition in this work are valid at the level of distribution pairs, and we abuse the notation $\otimes_{\textup{adapt}}$ as a shorthand in the main text to better communicate the structure of these results.
\end{remark}

\begin{figure}
    \centering
    \includegraphics[width=\linewidth]{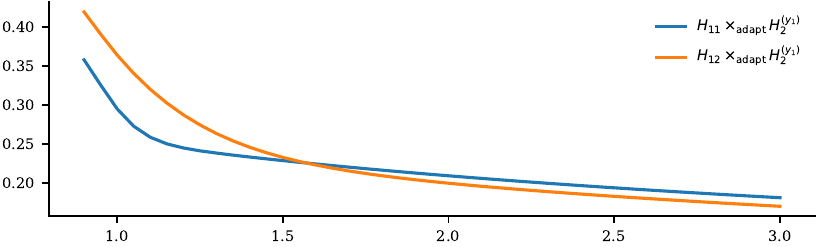}
    \caption{Privacy profiles $H_{11} \otimes_{\textup{adapt}} H_2^{(y_1)}$ and $H_{12} \otimes_{\textup{adapt}} H_2^{(y_1)}$ for $(P_{11}, Q_{11})$ being $0.1$-GDP and $(P_{12}, Q_{12})$ being $0.3$-GDP, with the same selection rule for step 2 that assigns $H_{G_{0.01}}$ when $y_1 < 0.3$ and $H_{G_2}$ otherwise. The two curves cross near $\gamma \approx 1.5$, demonstrating that the Blackwell ordering $H_{11} \leq H_{12}$ of the first profiles is not preserved under adaptive tensor product.}
    \label{fig:adapt_prod_notdefined}
\end{figure}

\section{Counterexamples Against \(f\)-DP Filter Validity}\label{appx:counterexample}

We compute tensor products in terms of privacy profiles, as Definition \ref{def:H-product} offers a direct integral form. We can then convert between privacy profiles and trade-off functions using Theorem \ref{thm:f<->H}. By Theorem \ref{thm:order-f-H}, the natural \(f\)-DP filter can be equivalently expressed using privacy profiles as
\begin{align}\label{eq:H-filter}
    F\!\left(H_B,H_1,\ldots,H_t\right) = \begin{cases}
            \textup{CONT} & \text{if } H_1\otimes\ldots\otimes H_t\leq H_B, \\
            \textup{HALT} & \text{otherwise}.
        \end{cases}
\end{align}
Motivated by Corollary \ref{cor:accounting-per-pair}, we use the \(f\)-DP filter with a budget \(H_B^-\) for privacy profiles \(H_1^-,\ldots,H_t^-\) under remove adjacency, since the validity of the filter would imply that the produced composition has a privacy profile \(\sym\!\left(H_B^-\right)\).

\subsection{Details of the Counterexample in Corollary~\ref{cor:SG_counterexample}}\label{appx:details-counterexample}

We work throughout with remove-adjacency privacy profiles, as motivated in the paragraph above. Let the datapoint space be \(\O=\{X\}\), so the only remove pair is \(S=(X)\) and \(S^-=\emptyset\). Fix composition length \(T=3\), subsampling rate \(q=0.5\), and additive Gaussian noise \(\sigma=1\) for all mechanisms \(\mathcal{M}_{1:3}\). Under input \(S^-\), every mechanism outputs mean \(0\); under input \(S\), the output means are
\begin{align*}
    \mu_1 = 1.3,\qquad
    (\mu_2,\mu_3) = \begin{cases}
        (0.1,\,10) & \text{if } y_1 \leq 0.65,\\
        (2,\,2) & \text{if } y_1 > 0.65.
    \end{cases}
\end{align*}
Note that \(\mathcal{M}_3\) depends only on \(y_1\), the output of \(\mathcal{M}_1\), and not on \(y_2\).

\para{Budget construction.} Denote by \(H^-_{SG,q,\mu}\) the remove-adjacency privacy profile of the subsampled Gaussian mechanism with parameters \((q,\mu)\), converted from \(SG_{q,\mu}^-\) as given by Theorem~\ref{thm:f<->H}. Set
\[
    H_{2\times3,1} = H^-_{SG,0.5,0.1} \otimes H^-_{SG,0.5,10}, \qquad
    H_{2\times3,2} = H^-_{SG,0.5,2}   \otimes H^-_{SG,0.5,2},
\]
which correspond to the future privacy profiles of \((\mathcal{M}_2,\mathcal{M}_3)\) conditional on \(y_1\leq 0.65\) and \(y_1>0.65\), respectively. These are illustrated in Figure~\ref{fig:H231xH232_remove}; they cross at a calibration point \(\gamma_0\approx 0.971\), confirming there is no Blackwell order between them.

\begin{figure}
    \centering
    \includegraphics[width=\linewidth]{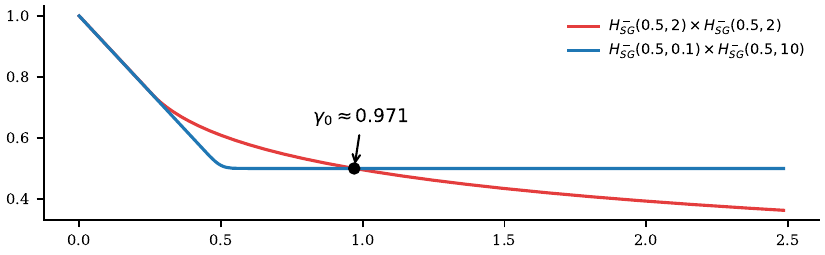}
    \caption{Tensor products of future privacy profiles, crossing at calibration point \(\gamma_0\approx 0.971\). These correspond to the trade-off functions in the left panel of Figure~\ref{fig:SG_counterexample}.}
    \label{fig:H231xH232_remove}
\end{figure}

The tightest budget that permits release of \(y_2\) and \(y_3\) regardless of \(y_1\) is, by~\eqref{eq:H-filter},
\[
    H_{B,\textup{tight}}
    = \max\!\left(
        H^-_{SG,0.5,1.3} \otimes H_{2\times3,1},\;
        H^-_{SG,0.5,1.3} \otimes H_{2\times3,2}
    \right),
\]
which passes the filter check for each branch by construction.

\para{True adaptive guarantee.} Let \(p_t\) and \(q_t\) denote the output densities of \(\mathcal{M}_t\) under inputs \(S\) and \(S^-\), respectively. Using Definition~\ref{def:H-product}, the tight privacy profile of the fully adaptive composition is
\begin{align*}
    H_{\textup{adapt}}(\gamma)
    &= \int_{\mathbb{R}^3}
        \Bigl[p_1(y_1)\,p_2(y_2\mid y_1)\,p_3(y_3\mid y_1)
        - \gamma\, q_1(y_1)\,q_2(y_2\mid y_1)\,q_3(y_3\mid y_1)\Bigr]_+
        \,\mathrm{d}y_1\,\mathrm{d}y_2\,\mathrm{d}y_3 \\
    &= \int_{\mathbb{R}} \left(
        \int_{\mathbb{R}^2}
        \Bigl[p_2(y_2\mid y_1)\,p_3(y_3\mid y_1)
        - \gamma\tfrac{q_1(y_1)}{p_1(y_1)}\,q_2(y_2\mid y_1)\,q_3(y_3\mid y_1)\Bigr]_+
        \mathrm{d}y_2\,\mathrm{d}y_3
    \right) p_1(y_1)\,\mathrm{d}y_1 \\
    &= \int_{-\infty}^{0.65}
        H_{2\times3,1}\!\left(\gamma\tfrac{q_1(y_1)}{p_1(y_1)}\right) p_1(y_1)\,\mathrm{d}y_1
      + \int_{0.65}^{\infty}
        H_{2\times3,2}\!\left(\gamma\tfrac{q_1(y_1)}{p_1(y_1)}\right) p_1(y_1)\,\mathrm{d}y_1.
\end{align*}

All privacy profiles are computed as integrals numerically; the resulting \(H_{B,\textup{tight}}\) and \(H_{\textup{adapt}}\) are compared in Figure~\ref{fig:H_B_vs_H_adapt}.

\begin{figure}
    \centering
    \includegraphics[width=\linewidth]{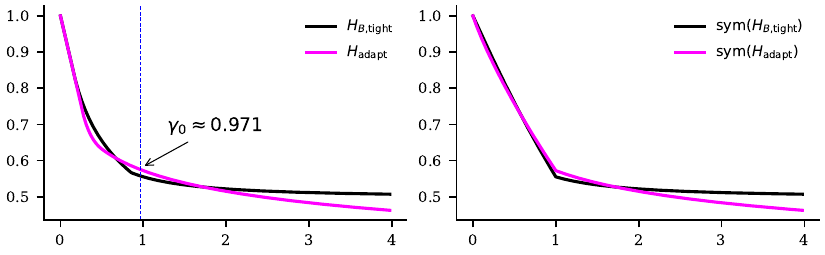}
    \caption{Left: \(H_{\textup{adapt}}\) exceeds \(H_{B,\textup{tight}}\) in the region surrounding \(\gamma_0\approx 0.971\). Right: the violation persists after symmetrization. These privacy profiles correspond to panels (c) and (d) in Figures~\ref{fig:SG_counterexample}.}
    \label{fig:H_B_vs_H_adapt}
\end{figure}

\para{Why the failure occurs.} We now show that \(H_{\textup{adapt}}(\gamma_0) > H_{B,\textup{tight}}(\gamma_0)\), which establishes the invalidity claimed in Corollary~\ref{cor:SG_counterexample}.

The two tensor products underlying \(H_{B,\textup{tight}}\) each average a \emph{fixed} future profile over all outputs \(y_1\):
\begin{align*}
    \bigl(H^-_{SG,0.5,1.3} \otimes H_{2\times3,i}\bigr)(\gamma)
    = \int_{\mathbb{R}} H_{2\times3,i}\!\left(\gamma\tfrac{q_1(y_1)}{p_1(y_1)}\right) p_1(y_1)\,\mathrm{d}y_1,
    \quad i=1,2.
\end{align*}
By contrast, \(H_{\textup{adapt}}\) is free to select a different profile per output. At \(\gamma=\gamma_0\), the ratio \(q_1(y_1)/p_1(y_1)\geq 1\) when \(y_1\leq 0.65\) and \(<1\) when \(y_1>0.65\), so the rescaled argument \(\gamma_0 q_1(y_1)/p_1(y_1)\) lies above \(\gamma_0\) in the first \(y_1\)-region and below in the second. Because \(H_{2\times3,1}>\!H_{2\times3,2}\) above \(\gamma_0\) and \(H_{2\times3,2}>\!H_{2\times3,1}\) below \(\gamma_0\) (the calibration point), the adaptive integrand takes the larger profile in each region simultaneously:
\begin{align*}
    H_{\textup{adapt}}(\gamma_0)
    &= \int_{-\infty}^{0.65} H_{2\times3,1}\!\left(\gamma_0\frac{q_1(y_1)}{p_1(y_1)}\right) p_1(y_1)\,\mathrm{d}y_1
     + \int_{0.65}^{\infty} H_{2\times3,2}\!\left(\gamma_0\frac{q_1(y_1)}{p_1(y_1)}\right) p_1(y_1)\,\mathrm{d}y_1,
\end{align*}
whereas each fixed-branch integral must use the \emph{same} profile across both regions. It follows that
\[
    H_{\textup{adapt}}(\gamma_0)
    > \max\!\left[
        \bigl(H^-_{SG,0.5,1.3}\otimes H_{2\times3,1}\bigr)(\gamma_0),\;
        \bigl(H^-_{SG,0.5,1.3}\otimes H_{2\times3,2}\bigr)(\gamma_0)
    \right]
    = H_{B,\textup{tight}}(\gamma_0).
\]
By continuity of privacy profiles, the strict inequality holds on a neighbourhood of \(\gamma_0\). This confirms that the filter is violated, with the calibration point as the precise cause.

\subsection{Proof of Adaptive Tightness of Standard Composition Results}

We prove the result in terms of privacy profiles, then convert to trade-off functions via Theorem~\ref{thm:f<->H}.

\begin{theorem}\label{thm:standard-composition-optimal}
    Let $(P_1,Q_1)$ be a pair of distributions on $\mathcal{O}_1$ with privacy profile $H_1$. For each $y_1\in\mathcal{O}_1$, let $(P_2^{(y_1)},Q_2^{(y_1)})$ be freely chosen from a family of pairs of distributions on $\mathcal{O}_2$, and let $\mathcal{H}_2$ denote the family of all achievable privacy profiles at step 2. Set $H_2^\uparrow=\sup_{H_2\in\mathcal{H}_2}H_2$ (pointwise). Then the adaptive product $(P_1\times P_2^{(y_1)},\,Q_1\times Q_2^{(y_1)})$ for any selection rule is dominated by $H_1\otimes H_2^\uparrow$. Furthermore, this bound is tight: for every $\gamma\geq 0$ and $\eta>0$, there exists a selection rule $y_1\mapsto(P_2^{(y_1)},Q_2^{(y_1)})$ such that
    \[
        \left(H_1\otimes H_2^\uparrow\right)(\gamma) - \eta
        \;\leq\;
        D_\gamma\!\left[P_1\times P_2^{(y_1)}\;\middle\|\;Q_1\times Q_2^{(y_1)}\right]
        \;\leq\;
        \left(H_1\otimes H_2^\uparrow\right)(\gamma).
    \]
\end{theorem}

\begin{remark}
    The tightness statement is pointwise in $\gamma$: a single selection rule need not bring the full privacy profile uniformly close to $H_1\otimes H_2^\uparrow$.
\end{remark}

\begin{proof}
    The upper bound follows immediately from Property~\ref{propitem:H-product-ordering} of Proposition~\ref{prop:properties-H-product}.

    For the lower tightness bound, fix $\gamma>0$ and $\eta>0$. Since $H_2^\uparrow=\sup_{H_2\in\mathcal{H}_2}H_2$, for each $y_1\in\mathcal{O}_1$ we may choose $H_2^{(y_1)}\in\mathcal{H}_2$ such that
    \[
        H_2^\uparrow\!\left(\gamma\frac{q_1(y_1)}{p_1(y_1)}\right) - \eta
        \;<\;
        H_2^{(y_1)}\!\left(\gamma\frac{q_1(y_1)}{p_1(y_1)}\right)
        \;\leq\;
        H_2^\uparrow\!\left(\gamma\frac{q_1(y_1)}{p_1(y_1)}\right).
    \]
    Integrating against $p_1(y_1)\,\mathrm{d}y_1$ and applying Definition~\ref{def:H-product} gives
    \begin{multline*}
        \left(H_1\otimes H_2^\uparrow\right)\!(\gamma) - \eta
        \;\leq\;
        \int_{\mathcal{O}_1} H_2^{(y_1)}\!\left(\gamma\frac{q_1(y_1)}{p_1(y_1)}\right)p_1(y_1)\,\mathrm{d}y_1 \\
        \;=\;
        D_\gamma\!\left[P_1\times P_2^{(y_1)}\;\middle\|\;Q_1\times Q_2^{(y_1)}\right]
        \;\leq\;
        \left(H_1\otimes H_2^\uparrow\right)\!(\gamma).
    \end{multline*}
\end{proof}

The trade-off function counterpart follows by translating between privacy profiles and trade-off functions.

\begin{corollary}\label{cor:tightness-fDP}
    Under the same setup, let $f_2^\downarrow = \ceinf_{f_2\in\mathcal{F}_2}f_2$. Then the adaptive product is lower bounded by $f_1\otimes f_2^\downarrow$, and for every $\alpha\in[0,1]$ and $\eta>0$ there exists a selection rule such that
    \[
        \left(f_1\otimes f_2^\downarrow\right)(\alpha)
        \;\leq\;
        T\!\left[P_1\times P_2^{(y_1)}\;\middle\|\;Q_1\times Q_2^{(y_1)}\right](\alpha)
        \;\leq\;
        \left(f_1\otimes f_2^\downarrow\right)(\alpha) + \eta.
    \]
\end{corollary}

\begin{proof}
    Combine Theorem~\ref{thm:standard-composition-optimal} with the correspondence in Theorem~\ref{thm:f<->H}, using the fact that $H_2^\uparrow$ in the privacy profile formulation corresponds to $f_2^\downarrow$ in the trade-off function formulation under Theorem~\ref{thm:order-f-H}.
\end{proof}

\para{Connection to the counterexample.} Corollary~\ref{cor:tightness-fDP} implies that the $f$-DP filter can fail whenever $H_1\otimes H_2^\uparrow > H_B$ at some point, since the adaptive composition can approach $H_1\otimes H_2^\uparrow$ arbitrarily closely there.

\begin{corollary}\label{cor:counterexample-2step}
    In fully adaptive composition of up to 2 steps,\footnote{If $y_1$ causes the filter to halt, we set $H_2^{(y_1)}(\gamma)=H_{\mathrm{Id}}(\gamma)=[1-\gamma]_+$.} a counterexample to filter validity exists whenever there is some $\gamma_0>0$ such that $\left(H_1\otimes H_2^\uparrow\right)(\gamma_0) > H_B(\gamma_0)$.
\end{corollary}

The proof of Theorem~\ref{thm:standard-composition-optimal} also makes explicit how calibration points lead to filter failure: when $H_{2,1},H_{2,2}\in\mathcal{H}_2$ cross at $\gamma_0$, the pointwise maximum $H_2^\uparrow=\max(H_{2,1},H_{2,2})$ satisfies
\[
    \left(H_1\otimes\max(H_{2,1},H_{2,2})\right)(\gamma_0)
    \;>\;
    \max\!\left(\left(H_1\otimes H_{2,1}\right)(\gamma_0),\,\left(H_1\otimes H_{2,2}\right)(\gamma_0)\right),
\]
and the extremizing selection rule is simply $y_1\mapsto\arg\max_{H_2\in\{H_{2,1},H_{2,2}\}}H_2\!\left(\gamma_0\tfrac{q_1(y_1)}{p_1(y_1)}\right)$. This is precisely the mechanism used in Appendix~\ref{appx:details-counterexample}, where $H_{2\times3}^\uparrow=\max(H_{2\times3,1},H_{2\times3,2})$. As shown in Figure~\ref{fig:H_B_vs_H_adapt_vs_H_crit}, $H_{\mathrm{adapt}}$ touches $H_1\otimes H_{2\times3}^\uparrow$ at $\gamma_0\approx 0.971$ but not everywhere, consistent with the pointwise nature of the tightness result.

\begin{figure}[H]
    \centering
    \includegraphics[width=\linewidth]{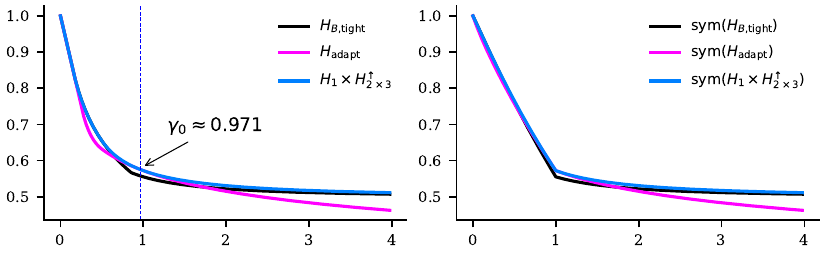}
    \caption{$H_1\otimes H_{2\times3}^\uparrow$ compared to $H_{\mathrm{adapt}}$ and $H_{B,\mathrm{tight}}$, illustrating Corollary~\ref{cor:counterexample-2step}. The adaptive composition reaches $H_1\otimes H_{2\times3}^\uparrow$ at the calibration point $\gamma_0$ but not globally.}
    \label{fig:H_B_vs_H_adapt_vs_H_crit}
\end{figure}

\section{Convergence of Privacy Profiles/Trade-off Functions}\label{appx:convergences}
\subsection{Existence of Convergent Sequence in Blackwell Chains}
We can impose the structural assumption of Blackwell chain on privacy profiles, in a similar way to trade-off functions.
\begin{definition}
    A family \(\H\) of privacy profiles is a Blackwell chain if any pair \(H_1,H_2\in\H\) is Blackwell ordered, i.e.~\(H_1\leq H_2\) or \(H_1\geq H_2\).
\end{definition}
By Theorem \ref{thm:order-f-H}, \(\H\) is a Blackwell chain if and only if its corresponding family \(\F\) of trade-off functions is a Blackwell chain.

\begin{lemma}\label{lem:ineq-expanding-coverage-resolution}
    In a Blackwell chain \(\H\) of privacy profiles with the point-wise supremum \(H^\uparrow\), there exists an everywhere increasing sequence \(\left(H_n\right)_{n\in\Na}\) such that for all \(n\in\Na\),
    \begin{align}\label{eq:ineq-expanding-coverage-resolution}
        H^\uparrow\!\left(\frac{i}{2^n}\right) - \frac{1}{n} < H_n\!\left(\frac{i}{2^n}\right) \leq H^\uparrow\!\left(\frac{i}{2^n}\right)\,\text{for }i=1,2,\ldots,n2^n.
    \end{align}
\end{lemma}
Lemma \ref{lem:ineq-expanding-coverage-resolution} enables inequalities at finitely many points covering regions \(\left[0,n\right]\) of increasing sizes \(n\) and increasing \mbox{resolutions \(\frac{1}{2^n}\)}. As \(\cup_{n\in\Na}\left[0,n\right]=\left[0,\infty\right)\) and \(\frac{1}{2^n}\searrow0\), these inequalities serve to connect point-wise properties to set-wise properties.
\begin{proof}
    Since \(H^\uparrow\!\left(\gamma\right) = \sup_{H\in\H} H\!\left(\gamma\right)\), for all \(n\in\Na\), for all \(i=1,2,\ldots,n2^n\), there exists \(H_{ni}\in\H\) such that at \(\gamma=\frac{i}{2^n}\),
    \begin{align*}
        H^\uparrow\!\left(\frac{i}{2^n}\right) - \frac{1}{n} < H_{ni}\!\left(\frac{i}{2^n}\right) \leq H^\uparrow\!\left(\frac{i}{2^n}\right).
    \end{align*}
    Let \(\hat{H}_n = \max_{1\leq i \leq n2^n} H_{ni}\), then \(\hat{H}_n\) satisfies the required inequality \eqref{eq:ineq-expanding-coverage-resolution}, and \(\hat{H}_n\) is in \(\H\) since \(\H\) is a Blackwell chain. Finally, since \(\left(\frac{i}{2^n}\right)_{i=1}^{n2^n} \subset \left(\frac{i}{2^{n+1}}\right)_{i=1}^{\left(n+1\right)2^{n+1}}\), we choose \(H_n = \max_{k\leq n} \hat{H}_k\) so that \(\left(H_n\right)_{n\in\Na}\) is increasing while maintaining inequality \eqref{eq:ineq-expanding-coverage-resolution}.
\end{proof}
\begin{lemma}\label{lem:chain->inc_seq_conv}
    In a Blackwell chain of privacy profiles \(\H\) with a smallest upper bound \(H^\uparrow=\sup_{H\in\H}H\), there exists an everywhere increasing sequence \(\left(H_n\right)_{n\in\Na}\) such that \(H_n\nearrow_{n\to\infty}H^\uparrow\) point-wise.
\end{lemma}
\begin{proof}
    Apply Lemma \ref{lem:ineq-expanding-coverage-resolution} to find an everywhere increasing sequence \(H_n\) such that for all \(n\in\Na\), for \(i=1,2,\ldots,n2^n\),
    \begin{align*}
        H^\uparrow\!\left(\frac{i}{2^n}\right) - \frac{1}{n} < H_n\!\left(\frac{i}{2^n}\right) \leq H^\uparrow\!\left(\frac{i}{2^n}\right).
    \end{align*}
    We show that this sequence converges point-wise to \(H^\uparrow\). Fix \(\gamma>0\), then \(n\geq \gamma\) once \(n\) is big enough, say, \(n\geq N\). If \(\gamma\) has a finite binary representation, then once \(n\) is big enough,
    \begin{align*}
        H^\uparrow\!\left(\gamma\right) \geq H_n\!\left(\gamma\right) > H^\uparrow\!\left(\gamma\right) - \frac{1}{n} \nearrow_{n\to\infty} H^\uparrow\!\left(\gamma\right).
    \end{align*}
    Otherwise if \(\gamma\) has an infinite binary representation, let \(\left(\gamma_n\right)_{n\geq N}\) be a decreasing sequence that converges to \(\gamma\), where each \(\gamma_n\) has a finite binary representation and \(H_n\!\left(\gamma_n\right)>H^\uparrow\!\left(\gamma_n\right)-\frac{1}{n}\). Then
    \begin{align*}
        H^\uparrow\!\left(\gamma\right) \geq H_n\!\left(\gamma\right) \geq  H_n\!\left(x_n\right) > H^\uparrow\!\left(\gamma_n\right) - \frac{1}{n} \nearrow_{n\to\infty} H^\uparrow\!\left(\gamma\right).
    \end{align*}
    Here, we used the fact that \(H^\uparrow\!\left(\gamma_n\right)\nearrow_{n\to\infty} H^\uparrow\!\left(\gamma\right)\), which is due to the continuity of \(H^\uparrow\).
\end{proof}

\subsection{Unifying Various Notions of Convergence}
In addition to the point-wise convergence shown in Lemma \ref{lem:chain->inc_seq_conv}, we discuss other notions of convergence.
\begin{enumerate}
    \item Uniform convergence \(\sup_{\gamma\geq0} \left|H_n\!\left(\gamma\right) - H\!\left(\gamma\right)\right| \xrightarrow[]{n\to\infty} 0\), much stronger than point-wise convergence.
    \item Uniform convergence over any compact set: \(\sup_{\gamma\in\left[0,\gamma_{\max}\right]} \left|H_n\!\left(\gamma\right) - H\!\left(\gamma\right)\right| \xrightarrow[]{n\to\infty} 0\) for all \(\gamma_{\max}>0\).
    \item Convergence in the symmetrized \(\Delta\)-divergence introduced by \citet{kaissis_beyond_2024}:
    \begin{align*}
        \sup_{\gamma\geq0} \frac{\left|H_n\!\left(\gamma\right) - H\!\left(\gamma\right)\right|}{1+\gamma} \xrightarrow[]{n\to\infty} 0
    \end{align*}
\end{enumerate}
\begin{proposition}\label{prop:conv->conv}
    Uniform convergence implies convergence in \(\Delta\)-divergence.
\end{proposition}
\begin{proof}
    The \(\Delta\)-divergence is smaller or equal to the \(\L^\infty\) norm,
    \begin{align*}
        \sup_{\gamma\geq0} \frac{\left|H\!\left(\gamma\right) - H_n\!\left(\gamma\right)\right|}{1+\gamma} \leq \sup_{\gamma\geq0} \left|H\!\left(\gamma\right) - H_n\!\left(\gamma\right)\right|.
    \end{align*}
\end{proof}
The following example shows that the implication in Proposition \ref{prop:conv->conv} is strictly one-sided.
\begin{example}\label{ex:no-uniform-conv}
    For \(n\in\Na\), let
    \begin{align*}
        H_n\!\left(\gamma\right) = \begin{cases}
                1 - \frac{\gamma}{n} & \text{if }\gamma\leq n, \\
                0 & \text{otherwise}.
            \end{cases}
    \end{align*}
    Then \(H_n\) converges to \(H=1\) point-wise but not uniformly. A closer examination reveals that for all \(n\), \(\frac{H\left(\gamma\right)-H_n\left(\gamma\right)}{1+\gamma}\) reaches its maximum at \(\gamma=n\). Therefore
    \begin{align*}
        \max_{\gamma\geq0} \frac{H\!\left(\gamma\right)-H_n\!\left(\gamma\right)}{1+\gamma} = \frac{1}{1+n} \xrightarrow[]{n\to\infty} 0,
    \end{align*}
    i.e.~\(H_n\) converges to \(H\) in \(\Delta\)-divergence.
\end{example}
\begin{theorem}\label{thm:unifying-Hconv}
    The following notions of convergence are equivalent for a sequence of privacy profiles \(H_n\) that converges to a privacy profile \(H\).
    \begin{enumerate}
        \item Convergence in \(\Delta\)-divergence.
        \item Point-wise convergence.
        \item Uniform convergence over any compact set.
    \end{enumerate}
\end{theorem}
\begin{proof}
    \textit{1\(\implies\)2}. At any order \(\gamma_0>0\),
    \begin{align*}
        \left|H_n\!\left(\gamma_0\right) - H\!\left(\gamma_0\right)\right| = \left(1+\gamma_0\right)\frac{\left|H_n\!\left(\gamma_0\right) - H\!\left(\gamma_0\right)\right|}{1+\gamma_0} \leq \left(1+\gamma_0\right)\sup_{\gamma\geq0}\frac{\left|H_n\!\left(\gamma\right) - H\!\left(\gamma\right)\right|}{1+\gamma} \xrightarrow[]{n\to\infty}0.
    \end{align*}
    \textit{2\(\implies\)3}. Since \(H_n\) is decreasing for all \(n\) and \(H\) is continuous, uniform convergence over any compact set follows from Pólya's theorem, which is proven e.g.~as Lemma A.7 by \citet{dong_gaussian_2022}.
    
    \textit{3\(\implies\)1}. Fix \(\epsilon>0\). We find \(N\) such that for all \(n>N\),
    \begin{align*}
        \sup_{\gamma\geq0} \frac{\left|H_n\!\left(\gamma\right) - H\!\left(\gamma\right)\right|}{1+\gamma} < \epsilon.
    \end{align*}
    Choose \(\gamma_{\max}=\frac{1}{\epsilon}-1\) so that \(\frac{1}{1+\gamma_{\max}}=\epsilon\). Note that \(H_n\) converges uniformly to \(H\) over \(\gamma\in\left[0,\gamma_{\max}\right]\), which implies that there exists \(N\) such that for all \(n>N\),
    \begin{align*}
        \left|H_n\!\left(\gamma\right) - H\!\left(\gamma\right)\right|<\epsilon\textup{ for all }\gamma\in\left[0,\gamma_{\max}\right].
    \end{align*}
    Therefore, for all \(n>N\),
    \begin{align*}
        \sup_{\gamma\geq0} \frac{\left|H_n\!\left(\gamma\right) - H\!\left(\gamma\right)\right|}{1+\gamma} &= \max\!\left[\sup_{\gamma\in\left[0,\gamma_{\max}\right]} \frac{\left|H_n\!\left(\gamma\right) - H\!\left(\gamma\right)\right|}{1+\gamma},\sup_{\gamma>\gamma_{\max}} \frac{\left|H_n\!\left(\gamma\right) - H\!\left(\gamma\right)\right|}{1+\gamma}\right] \\
        &\leq \max\!\left(\epsilon,\frac{1}{1+\gamma_{\max}}\right) = \epsilon.
    \end{align*}
\end{proof}
\begin{corollary}\label{cor:chain->inc_seq_conv_unified}
    In a Blackwell chain \(\H\) of privacy profiles with a smallest upper bound \(H^\uparrow=\sup_{H\in\H}H\), there exists an everywhere increasing sequence \(\left(H_n\right)_{n\in\Na}\) such that \(H_n\nearrow_{n\to\infty}H^\uparrow\) under any notion of convergence in Theorem \ref{thm:unifying-Hconv}.
\end{corollary}
\begin{proof}
    Combine Lemma \ref{lem:chain->inc_seq_conv} with Theorem \ref{thm:unifying-Hconv}.
\end{proof}

\begin{theorem}\label{thm:unifying-fconv}
    The following notions of convergence are equivalent for a sequence of trade-off functions \(f_n\) that converges to a trade-off function \(f\).
    \begin{enumerate}
        \item Convergence in \(\Delta\)-divergence \cite{kaissis_beyond_2024}:
        \begin{align*}
            \min\!\left\{\Delta\geq0:\forall\alpha\in\left[0,1-\Delta\right]\!,\ 
            \begin{gathered}
                f_n\!\left(\alpha+\Delta\right)-\Delta\leq f\!\left(\alpha\right) \\
                 f\!\left(\alpha+\Delta\right)-\Delta\leq f_n\!\left(\alpha\right)
            \end{gathered}\right\} \xrightarrow[]{n\to\infty}0.
        \end{align*}
        \item Point-wise convergence over \(\left(0,1\right]\).
        \item Uniform convergence over \(\left[\alpha_{\min},1\right]\) for all \(\alpha_{\min}\in\left(0,1\right]\). 
    \end{enumerate}
\end{theorem}
\begin{proof}
    \textit{1\(\implies\)2}. \(f_n\) converges to \(f\) in \(\Delta\)-divergence if and only if there exists a decreasing sequence of positive numbers \(\Delta_n>0\) such that \(\Delta_n\searrow^{n\to\infty}0\) and for all \(n\),
    \begin{align*}
        f\!\left(\alpha+\Delta_n\right) - \Delta_n \leq f_n\!\left(\alpha\right)\,&\textup{for all }\alpha\in\left[0,1-\Delta_n\right], \\
        f_n\!\left(\alpha\right) \leq f\!\left(\alpha-\Delta_n\right) + \Delta_n \ &\textup{for all }\alpha\in\left[\Delta_n,1\right].
    \end{align*}
    Fix \(\alpha\in\left(0,1\right)\). Then once \(n\) is large enough, say, \(n>N\), \(\Delta_n<\alpha<1-\Delta_n\) since \(\Delta_n\searrow^{n\to\infty}0\). Therefore, for all \(n>N\),
    \begin{align*}
        f\!\left(\alpha+\Delta_n\right) - \Delta_n \leq f_n\!\left(\alpha\right) \leq f\!\left(\alpha-\Delta_n\right) + \Delta_n.
    \end{align*}
    Since \(\lim_{n\to\infty}f\!\left(\alpha+\Delta_n\right) - \Delta_n=\lim_{n\to\infty}f\!\left(\alpha-\Delta_n\right) + \Delta_n=f\!\left(\alpha\right)\) due to the continuity of \(f\), \mbox{\(\lim_{n\to\infty}f_n\!\left(\alpha\right) = f\!\left(\alpha\right)\)}.

    \textit{2\(\implies\)3}. Since \(f_n\) is decreasing for all \(n\) and \(f\) is continuous, uniform convergence over \(\left[\alpha_{\min},1\right]\) for all \(\alpha_{\min}\in\left(0,1\right]\) follows from Pólya's theorem, which is proven e.g.~as Lemma A.7 by \citet{dong_gaussian_2022}.

    \textit{3\(\implies\)1}. Fix \(\Delta>0\). Since \(f_n\) converges to \(f\) uniformly over \(\left[\Delta,1\right]\), there exists some \(N\) such that for all \(n>N\),
    \begin{align*}
        -\Delta\leq f_n\!\left(\alpha\right) - f\!\left(\alpha\right) \leq\Delta\ \textup{for all }\alpha\in\left[\Delta,1\right].
    \end{align*}
    The first and second inequalities imply, respectively,
    \begin{align*}
        f_n\!\left(\alpha\right) &\leq f\!\left(\alpha\right)+\Delta \leq f\!\left(\alpha-\Delta\right)+\Delta, \\
        f\!\left(\alpha\right) &\leq f_n\!\left(\alpha\right)+\Delta \leq f_n\!\left(\alpha-\Delta\right)+\Delta
    \end{align*}
    for all \(\alpha\in\left[\Delta,1\right]\). Since this holds for all \(\Delta>0\), \(f_n\) converges to \(f\) in \(\Delta\)-divergence.
\end{proof}
\begin{corollary}
    In a Blackwell chain \(\F\) of trade-off functions with a greatest lower bound \(f^\downarrow=\textup{ce}\inf_{f\in\F}f\), there exists an everywhere decreasing sequence \(\left(f_n\right)_{n\in\Na}\) such that \(f_n\searrow^{n\to\infty}f^\downarrow\) under any notion of convergence in Theorem \ref{thm:unifying-fconv}.
\end{corollary}
\begin{proof}
    Denote \(\H\) the corresponding family of privacy profiles obtained by converting the trade-off functions in \(\F\) using Theorem \ref{thm:f<->H}. Since \(\F\) is a Blackwell chain, \(\H\) is also a Blackwell chain by Theorem \ref{thm:order-f-H}. By Corollary \ref{cor:chain->inc_seq_conv_unified}, there exists an everywhere increasing sequence \(H_n\) that converges to \(H^\uparrow=\sup_{H\in\H}H\) in \(\Delta\)-divergence. 
    
    By Theorem \ref{thm:order-f-H}, the sequence of trade-off functions \(f_n\) that corresponds to \(H_n\) is everywhere decreasing; and furthermore, the smallest upper bound \(H^\uparrow\) of privacy profiles in \(\H\) corresponds to the greatest lower bound \(f^\downarrow\) of trade-off functions in \(\F\). \citet{kaissis_beyond_2024} showed that \(\Delta\)-divergence can be converted between privacy profiles and trade-off functions, and hence \(f_n\) converges to \(f^\downarrow\) in \(\Delta\)-divergence. By Theorem \ref{thm:unifying-fconv}, \(f_n\) converges to \(f^\downarrow\) under all three notions of convergence.
\end{proof}

\section{Proofs for \(f\)-DP Filter Validity via Blackwell Chains}\label{appx:fDP-filter}
\subsection{Blackwell Chain Property Implies \(f\)-DP Filter Validity}
We first restrict ourselves to the simplest case with composition length \(T=2\) as a fundamental result before extending it to the general case of fully adaptive composition of up to \(T\geq3\). The following Proposition \ref{prop:to-use-f2down} is a direct consequence of Theorem \ref{thm:composition_fDP}.
\begin{proposition}\label{prop:to-use-f2down}
    Let \(\M_1\) be \(f_1\)-DP, and let the adaptively chosen \(\M_2\) given output \(y_1\) from \(\M_1\) be \(f_2^{\left(y_1\right)}\)-DP. If there exists a trade-off function \(f_2^\downarrow\) such that
    \begin{enumerate}
        \item\label{propitem:lowerbound} \(f_2^{\left(y_1\right)}\geq f_2^\downarrow\) for all \(y_1\in\O_1\),
        \item\label{propitem:lowerbound-satisfies-filter} \(f_1\otimes f_2^{\downarrow}\geq f_B\),
    \end{enumerate}
    then the composition \(\left(\M_1,\M_2\right)\) is \(f_B\)-DP.
\end{proposition}
Among the lower bounds of the family of options \(f_2^{\left(y_1\right)}\in\F\), i.e.~the trade-off functions that satisfy Condition \ref{propitem:lowerbound} above, it is enough to verify whether the greatest lower bound \(f_2^\downarrow=\textup{ce}\inf_{f_2\in\F_2}f_2\) satisfies Condition \ref{propitem:lowerbound-satisfies-filter}, as otherwise no other lower bound does.

Figure \ref{fig:H_B_vs_H_adapt_vs_H_crit} shows that calibration points cause Condition \ref{propitem:lowerbound-satisfies-filter} to fail. As shown in Appendix \ref{appx:counterexample}, a diverse family of mechanisms exhibiting a calibration point $\gamma_0$ in their privacy profiles may cause 
$\left(H_1\otimes H_2^{\uparrow}\right)\!\left(\gamma_0\right)>H_B\!\left(\gamma_0\right)$, or equivalently $\left(f_1\otimes f_2^{\downarrow}\right)\!\left(\alpha_0\right)<f_B\!\left(\alpha_0\right)$ at the corresponding $\alpha_0$. This suggests that the Blackwell chain structural assumption is necessary for both conditions in Proposition \ref{prop:to-use-f2down} to hold.
\begin{theorem}\label{thm:chain->nonfully-2step-valid}
    Let \(\M_1\) be \(f_1\)-DP, and let the adaptively chosen \(\M_2\) given output \(y_1\) from \(\M_1\) be \(f_2^{\left(y_1\right)}\)-DP. If
    \begin{enumerate}
        \item \(f_1\otimes f_2^{\left(y_1\right)}\geq f_B\) for all \(y_1\in\O_1\),
        \item \(\F = \left\{f_2^{\left(y_1\right)}\right\}_{\!y_1\in\O_1}\) is a Blackwell chain,
    \end{enumerate}
    then the greatest lower bound \(f_2^\downarrow=\textup{ce}\inf_{y_1\in\O_1}f_2^{\left(y_1\right)}\) satisfies \mbox{\(f_1\otimes f_2^\downarrow\geq f_B\)}. As a consequence, the composition \(\left(\M_1,\M_2\right)\) is \(f_B\)-DP by Proposition \ref{prop:to-use-f2down}.

    This can be equivalently formulated using privacy profiles. Let the privacy profile of \(\M_1\) be \(H_1\), and let the privacy profile of the adaptively chosen \(\M_2\) given output \(y_1\) from \(\M_1\) be \(H_2^{\left(y_1\right)}\). If
    \begin{enumerate}
        \item \(H_1\otimes H_2^{\left(y_1\right)}\leq H_B\) for all \(y_1\in\O_1\),
        \item \(\H = \left\{H_2^{\left(y_1\right)}\right\}_{y_1\in\O_1}\) is a Blackwell chain,
    \end{enumerate}
    then the smallest upper bound \(H_2^\uparrow=\sup_{y_1\in\O_1}H_2^{\left(y_1\right)}\) satisfies \(H_1\otimes H_2^\uparrow\leq H_B\). 
\end{theorem}
\begin{proof}
    Let \(\left(P_1,Q_1\right)\) be any pair of distributions that has a privacy profile \(H_1\). For a fixed \(\hat{y}_1\), \(H_1\otimes H_2^{\left(\hat{y}_1\right)}\leq H_B\) can be written  explicitly as
    \begin{align*}
        \int_{\O_1} H_2^{\left(\hat{y}_1\right)}\!\left(\gamma\frac{q_1\!\left(y_1\right)}{p_1\!\left(y_1\right)}\right)p_1\!\left(y_1\right)\d y_1 \leq H_B\!\left(\gamma\right)\ \text{for all } \gamma\geq0.
    \end{align*}
    While \(H_1\otimes H_2^\uparrow\leq H_B\) can be written explicitly as
    \begin{align*}
        \int_{\O_1} H_2^\uparrow\!\left(\gamma\frac{q_1\!\left(y_1\right)}{p_1\!\left(y_1\right)}\right)p_1\!\left(y_1\right)\d y_1 \leq H_B\!\left(\gamma\right)\ \text{for all } \gamma\geq0.
    \end{align*}
    Lemma \ref{lem:chain->inc_seq_conv} states that there exists an everywhere increasing sequence \(\left(H_2^{\left(\hat{y}_{1,n}\right)}\right)_{n\in\Na}\) in \(\H_2\) such that \(H_2^{\left(\hat{y}_{1,n}\right)}\nearrow_{n\to\infty}H_2^\uparrow\) point-wise. By the monotone convergence theorem, for any \(\gamma\geq0\), 
    \begin{align*}
        \left(H_1\otimes H_2^\uparrow\right)\!\left(\gamma\right) = \lim_{n\to\infty} \left(H_1\otimes H_2^{\hat{y}_{1,n}}\right)\!\left(\gamma\right) \leq H_B\!\left(\gamma\right).
    \end{align*}
\end{proof}

In fact, Theorem \ref{thm:chain->nonfully-2step-valid} has already covered the case \(T=2\) in Theorem \ref{thm:validity}, as we can equivalently choose \(f_2^{\left(y_1\right)}\!\left(\alpha\right)=\textup{Id}\!\left(\alpha\right)=1-\alpha\) when \(y_1\) halts the filter early. We now give the remaining proof to Theorem \ref{thm:validity}.
\filtervalidity*
\begin{proof}
    The case \(T=2\) has been proven as Theorem \ref{thm:chain->nonfully-2step-valid}. We present a proof for the case of \(T=3\) for clarity, where \(\F^{\left(\_\otimes\F\right)\left(T-2\right)}=\F^{\_\otimes\F}\). The general case \(T>3\) follows by iteratively applying the same argument.

    Without loss of generality, when the previous outputs halt the filter early, we can set the \(f\)-DP guarantees for remaining mechanisms as the perfectly private trade-off function \(\textup{Id}\!\left(\alpha\right)=1-\alpha\). Hence, the filter inequality becomes, for all outputs \(y_1\in\O_1\) and \(y_2\in\O_2\),
    \begin{align*}
        f_1\otimes f_2^{\left(y_1\right)} \otimes f_3^{\left(y_1,y_2\right)} \geq f_B,
    \end{align*}
    where \(f_t\) is the \(f\)-DP guarantee of \(\M_t\) for \(t=1,2,3\).
    
    Fix \(y_1\in\O_1\). Then \(f_1\otimes f_2^{\left(y_1\right)}\) is fixed and \(f_3^{\left(y_1,y_2\right)}\) only depends on \(y_2\in\O_2\). Since \(\left\{f_3^{\left(y_1,y_2\right)}\right\}_{y_2\in\O_2}\subset\F\subset\F^{\_\otimes\F}\), \(\left\{f_3^{\left(y_1,y_2\right)}\right\}_{y_2\in\O_2}\) is a Blackwell chain, and thus by Theorem \ref{thm:chain->nonfully-2step-valid},
    \begin{align*}
        f_1\otimes f_2^{\left(y_1\right)} \otimes f_3^{\left(y_1,\downarrow\right)} \geq f_B,
    \end{align*}
    where \(f_3^{\left(y_1,\downarrow\right)}=\textup{ce}\inf_{y_2\in\O_2}f_3^{\left(y_1,y_2\right)}\) is the greatest lower bound of \(\left\{f_3^{\left(y_1,y_2\right)}\right\}_{y_2\in\O_2}\). 
    
    The deduction so far holds for all \(y_1\in\O_1\). Now, when we let \(y_1\in\O_1\) change as a variable, \(f_1\) is fixed and \(f_2^{\left(y_1\right)} \otimes f_3^{\left(y_1,\downarrow\right)}\) depends on \(y_1\). Since \(\left\{f_2^{\left(y_1\right)} \otimes f_3^{\left(y_1,\downarrow\right)}\right\}_{y_1\in\O_1}\subset\F^{\_\otimes\F}\), \(\left\{f_2^{\left(y_1\right)} \otimes f_3^{\left(y_1,\downarrow\right)}\right\}_{y_1\in\O_1}\) is a Blackwell chain, and consequently, by Theorem \ref{thm:chain->nonfully-2step-valid} again,
    \begin{align*}
        f_1\otimes f_{2\times3}^{\downarrow\downarrow} \geq f_B,
    \end{align*}
    where \(f_{2\times3}^{\downarrow\downarrow}=\textup{ce}\inf_{y_1\in\O_1}f_2^{\left(y_1\right)} \otimes f_3^{\left(y_1,\downarrow\right)}\) is the greatest lower bound of \(\left\{f_2^{\left(y_1\right)} \otimes f_3^{\left(y_1,\downarrow\right)}\right\}_{y_1\in\O_1}\).
    
    Conditioned on a fixed \(y_1\in\O_1\) as first output, \(\M_2^{\left(y_1\right)}\) is \(f_2^{\left(y_1\right)}\)-DP, and for all \(y_2\in\O_2\), \(\M_3^{\left(y_1,y_2\right)}\) is \(f_3^{\left(y_1,y_2\right)}\)-DP and hence \(f_3^{\left(y_1,\downarrow\right)}\)-DP. Therefore, by Theorem \ref{thm:composition_fDP}, the composition \(\left(\M_2^{\left(y_1\right)},\M_3^{\left(y_1,y_2\right)}\right)\) is \mbox{\(f_2^{\left(y_1\right)} \otimes f_3^{\left(y_1,\downarrow\right)}\)-DP} for all \(y_2\).
    
    It now follows that for all \(y_1\in\O_1\) and all \(y_2\in\O_2\), \(\left(\M_2^{\left(y_1\right)},\M_3^{\left(y_1,y_2\right)}\right)\) is \(f_{2\times3}^{\downarrow\downarrow}\)-DP. Since \(\M_1\) is \(f_1\)-DP, by Theorem \ref{thm:composition_fDP}, \(\left(\M_1,\M_2,\M_3\right)\) is \(f_1\otimes f_{2\times3}^{\downarrow\downarrow}\)-DP and hence \(f_B\)-DP.
\end{proof}
The \(f\)-DP central limit theorem \cite{dong_gaussian_2022} states that a composition of trade-off functions \(f_1\otimes\ldots\otimes f_n\) converges to some Gaussian trade-off function in \(\Delta\)-divergence as \(n\to\infty\). Consequently, as \(T\) increases, the trade-off functions in \(\F^{\left(\_\otimes\F\right)\left(T-2\right)}\) increasingly resemble Gaussian trade-off functions.
\begin{example}
    Let \(f\) be a trade-off function and consider the \(n\)-fold self-compositions \(\F=\left(f^{\otimes n}\right)_{n\in\Na\cup\left\{\infty\right\}}\), in which \(f^{\otimes\infty}=0\) if \(f\) is not \(\textup{Id}\), and \(\textup{Id}^{\otimes\infty}=\textup{Id}\). In this case, \(\F^{\left(\_\otimes\F\right)\left(T-2\right)}=\left(f^{\otimes n}\right)_{n\in\Na\cup\left\{\infty\right\}}\) for all \(T\geq2\), which is a Blackwell chain by Proposition \ref{prop:properties-f-product}. By Theorem \ref{thm:validity}, the \(f\)-DP filter is valid for mechanisms with tight \(f\)-DP guarantees matching \(f^{\otimes n}\) for some \(n\). A construction of such mechanisms when \(f\) is symmetric is given by the canonical noise distribution \cite{awan2021canonical}.
\end{example}

\subsection{Failure of Blackwell Chain Property Implies \(f\)-DP Filter Invalidity}
\begin{definition}
    For a family \(\H\) of privacy profiles, define
    \begin{align*}
        \bar{\H} = \left\{\bar{H}:\exists H_n\in\H\ \textup{s.t.}\,H_n\nearrow\bar{H}\right\}
    \end{align*}
    as the family of limits \(\bar{H}\) of some everywhere increasing sequence \(H_n\in\H\) under all notions of convergence in Theorem \ref{thm:unifying-Hconv}. Furthermore, define
    \begin{align*}
        \H^{\_\otimes\H} = \H\otimes\bar{\H} = \left\{H\otimes\bar{H}:H\in\H,\bar{H}\in\bar{\H}\right\}.
    \end{align*}
\end{definition}
\begin{lemma}\label{lem:exist-real-probs-approx-chain}
    Let \(\H\) be an arbitrary family of privacy profiles, and consider \(\H^{\left(\_\otimes\H\right)\left(T-2\right)}\) for some \(T\geq2\). For all \(H^*\in\H^{\left(\_\otimes\H\right)\left(T-2\right)}\), for all \(\gamma_{\max}>0\) and all \(\eta>0\), there exist \(T-1\) privacy profiles \(H_2,H_3,\ldots,H_T\in\H\) such that
    \begin{align*}
        \sup_{\gamma\in\left[0,\gamma_{\max}\right]}\left|\left(H_2\otimes H_3\otimes\ldots\otimes H_T\right)\!\left(\gamma\right) - H^*\!\left(\gamma\right)\right| \leq \eta.
    \end{align*}
\end{lemma}
\begin{proof}
    For clarity, we present the proof for the case \(T=4\). The general case involves iteratively applying the same kind of argument here.

    Now \(H^*\in\left(\H^{\_\otimes\H}\right)^{\_\otimes\H}\). There exists \(H_2\in\H\) and \(\overline{H_{3\times4}}\in\overline{\H^{\_\otimes\H}}\) such that \(H_2\otimes\overline{H_{3\times4}}=H^*\). Furthermore, there exists an everywhere increasing sequence \(H_{3\times4,n}\in\H^{\_\otimes\H}\) such that \(H_{3\times4,n}\nearrow_{n\to\infty}\overline{H_{3\times4}}\) under all notions of Theorem \ref{thm:unifying-Hconv}. By the monotone convergence theorem, \(H_2\otimes H_{3\times4,n}\nearrow_{n\to\infty}H_2\otimes \overline{H_{3\times4}}\) point-wise, and hence uniformly over any compact set by Theorem \ref{thm:unifying-Hconv}. Therefore, there exists \(n_0\) such that
    \begin{align*}
        \sup_{\gamma\in\left[0,\gamma_{\max}\right]} \left|\left(H_2\otimes H_{3\times4,n_0}\right)\!\left(\gamma\right)-\left(H_2\otimes\overline{H_{3\times4}}\right)\!\left(\gamma\right)\right| \leq \frac{\eta}{2}.
    \end{align*}
    Now \(H_{3\times4,n_0}\in\H^{\_\otimes\H}\). There exists \(H_3\in\H\) and \(\bar{H}_4\in\bar{\H}\) such that \(H_3\otimes\bar{H}_4=H_{3\times4,n_0}\). Furthermore, there exists an everywhere increasing sequence \(H_{4,k}\in\H\) such that \(H_{4,k}\nearrow_{k\to\infty}\bar{H_4}\) under all notions of Theorem \ref{thm:unifying-Hconv}. By the monotone convergence theorem, \(H_2\otimes H_3\otimes H_{4,k}\nearrow_{k\to\infty}H_2\otimes H_3\otimes\bar{H}_4\) point-wise, and hence uniformly over any compact set by Theorem \ref{thm:unifying-Hconv}. Therefore, there exists \(k_0\) such that
    \begin{align*}
        \sup_{\gamma\in\left[0,\gamma_{\max}\right]} \left|\left(H_2\otimes H_3\otimes H_{4,k_0}\right)\!\left(\gamma\right)-\left(H_2\otimes H_3\otimes \bar{H}_4\right)\!\left(\gamma\right)\right| \leq \frac{\eta}{2}.
    \end{align*}
    Since \(H_2\otimes H_3\otimes \bar{H}_4 = H_2\otimes H_{3\times4,n_0}\) and \(H_2\otimes\overline{H_{3\times4}}=H^*\), by the triangle inequality,
    \begin{align*}
        \sup_{\gamma\in\left[0,\gamma_{\max}\right]} \left|\left(H_2\otimes H_3\otimes H_{4,k_0}\right)\!\left(\gamma\right)-H^*\!\left(\gamma\right)\right| \leq \frac{\eta}{2}+\frac{\eta}{2} = \eta.
    \end{align*}
\end{proof}
\nochainisinvalid*
\begin{proof}
    Let \(\H\) be the family of privacy profiles corresponding to \(\F\), and let $H_1^*,H_2^*$ be the pair of profiles corresponding to $f_{21},f_{22}\in\F^{(\_\otimes\F)(T-2)}$. Since $f_{21},f_{22}$ cross at $\alpha_0$ with different slopes $f_{21}'(\alpha_0)$ and $f_{22}'(\alpha_0)$, by Theorem \ref{thm:f<->H}, $H_1^*,H_2^*$ cross and exchange order locally at a calibration point $\gamma_0$ lying between $-\nicefrac{1}{f_{21}'(\alpha_0)}$ and $-\nicefrac{1}{f_{22}'(\alpha_0)}$.

    Let $P_1,Q_1$ be any distribution pair with trade-off function $f_1$. From the assumptions on $f_1$, we show that $\left(\nicefrac{q_1(y_1)}{p_1(y_1)}\right)_{y_1\in\O_1}$ contains some interval in $\real$. By the Neyman-Pearson lemma, each pair of error rates $(\alpha,f_1(\alpha))$ is attained by some likelihood ratio test that uses $\nicefrac{q_1}{p_1}$. The decrease in FNR per unit increase in FPR equals the likelihood-ratio threshold, which is therefore $-f_1'(\alpha)$. Formally, for any $0<r_1<r_2$,
    \begin{align*}
        P_1\!\left(r_1<\frac{q_1}{p_1}<r_2\right)\! &= \text{length}\!\left\{\alpha\in(0,1):r_1<-f_1'(\alpha)<r_2\right\} \\
        &\geq \text{length}\!\left\{\alpha\in(\alpha_1,\alpha_2):r_1<-f_1'(\alpha)<r_2\right\}.
    \end{align*}
    Over $\alpha\in(\alpha_1,\alpha_2)$, since $f_1$ is convex and differentiable, $-f_1'$ is continuous. Together with the assumption that $-f_1'$ is not constant, this implies that $\left(\nicefrac{q_1(y_1)}{p_1(y_1)}\right)_{y_1\in\O_1}$ contains some interval in $\real$.

    As $H_1^*,H_2^*$ cross at $\gamma_0$, we can find $\gamma_\text{scale}$ such that $\gamma_0$ lies in the interior of $\left(\gamma_\text{scale}\nicefrac{q_1(y_1)}{p_1(y_1)}\right)_{y_1\in\O_1}$. Consequently,
    \begin{align*}
        \left(\gamma_\text{scale}\frac{q_1\!\left(y_1\right)}{p_1\!\left(y_1\right)}\right)_{y_1\in\O_1} \textup{overlaps with both }\left\{H^*_1>H^*_2\right\}\textup{ and }\left\{H^*_1<H^*_2\right\}.
    \end{align*}
    
    We apply Lemma \ref{lem:exist-real-probs-approx-chain} with small \(\eta>0\) to find \(H_{2,1},H_{3,1},\ldots,H_{T,1}\) and \(H_{2,2},H_{3,2},\ldots,H_{T,2}\) that respectively approximate \(H_1^*\) and \(H_2^*\) on \(\left[0,\gamma_{\max}\right]\), with $\gamma_{\max}$ large enough. Denote    
    \begin{align*}
        \tilde{H}^*_1 &= H_{2,1}\otimes H_{3,1}\otimes\ldots\otimes H_{T,1},\\
        \tilde{H}^*_2 &= H_{2,2}\otimes H_{3,2}\otimes\ldots\otimes H_{T,2}.
    \end{align*}
    Note that for $\gamma\in[0,\gamma_{\max}]$,
    \begin{align*}
        \left\{\tilde{H}^*_1>\tilde{H}^*_2\right\} &\supset \left\{H^*_1>H^*_2\text{ and }H^*_1-H^*_2>2\eta\right\}, \\
        \left\{\tilde{H}^*_2>\tilde{H}^*_1\right\} &\supset \left\{H^*_2>H^*_1\text{ and }H^*_2-H^*_1>2\eta\right\}.
    \end{align*}
    Since $\bigcup_{\eta>0}\left\{H^*_1-H^*_2>2\eta\right\}=\left\{H^*_1>H^*_2\right\}$ and $\bigcup_{\eta>0}\left\{H^*_2-H^*_1>2\eta\right\}=\left\{H^*_2>H^*_1\right\}$, once \(\eta>0\) is chosen small enough, 
    \begin{align*}
        \left(\gamma_\text{scale}\frac{q_1\!\left(y_1\right)}{p_1\!\left(y_1\right)}\right)_{y_1\in\O_1} \textup{overlaps with both }\left\{\tilde{H}^*_1>\tilde{H}^*_2\right\}\textup{ and }\left\{\tilde{H}^*_1<\tilde{H}^*_2\right\}.
    \end{align*}

    From this, we can construct a counterexample to the \(f\)-DP filter validity. For \(t=2,\ldots,T\), let \(\left(P_{t,1},Q_{t,1}\right)\) and \(\left(P_{t,2},Q_{t,2}\right)\) be the pairs of distributions that have privacy profiles \(H_{t,1}\) and \(H_{t,2}\) respectively. Choose the adaptive selection rule at step \(t\):
    \begin{align*}
        P_{t}^{\left(y_1,\ldots,y_{t-1}\right)} \coloneqq P_{t}^{\left(y_1\right)} &= \begin{cases}
                    P_{t,1} & \textup{if }\gamma_\text{scale}\frac{q_1\!\left(y_1\right)}{p_1\!\left(y_1\right)}\in\left\{\tilde{H}^*_1\geq\tilde{H}^*_2\right\}, \\
                    P_{t,2} & \textup{otherwise,}
                \end{cases} \\
        Q_{t}^{\left(y_1,\ldots,y_{t-1}\right)} \coloneqq Q_{t}^{\left(y_1\right)} &= \begin{cases}
                    Q_{t,1} & \textup{if }\gamma_\text{scale}\frac{q_1\!\left(y_1\right)}{p_1\!\left(y_1\right)}\in\left\{\tilde{H}^*_1\geq\tilde{H}^*_2\right\}, \\
                    Q_{t,2} & \textup{otherwise.}
                \end{cases}
    \end{align*}
    Writing $H_1$ for the privacy profile corresponding to $f_1$, we obtain the following inequality:
    \begin{multline*}
        \int_{y_1\in\O_1}\!\max\!\left[\tilde{H}_1^*\!\left(\gamma_\text{scale}\frac{q_1\!\left(y_1\right)}{p_1\!\left(y_1\right)}\right)\!,\tilde{H}_2^*\!\left(\gamma_\text{scale}\frac{q_1\!\left(y_1\right)}{p_1\!\left(y_1\right)}\right)\right]\!p_1\!\left(y_1\right)\d y_1 \\
        > \max\!\left(\left(H_1\otimes\tilde{H}_1^*\right)\!\left(\gamma_\text{scale}\right),\left(H_1\otimes\tilde{H}_2^*\right)\!\left(\gamma_\text{scale}\right)\right).
    \end{multline*}
    Here, the LHS is the true privacy profile of the composition, while the RHS is chosen as the budget used in the filter.
\end{proof}

\section{Fully Adaptive Central Limit Theorem}\label{appx:adaptCLT-approxGDPfilter}
\subsection{Proof to the Fully Adaptive \(f\)-DP Central Limit Theorem}
First, we present the martingale Berry--Esseen theorem from \citet{FAN20191028}.
\begin{theorem}\label{thm:martingale-BE}
    Let \(\left(\Omega,\G,P\right)\) be a probability space and let \(\G_1\subset\G_2\subset\ldots\subset\G_T\subset\G\) be a filtration of increasing \(\sigma\)-algebras. Let \(\left(\xi_i\right)_{i=1}^T\) be a sequence of martingale differences, which means that \(\sum_{i=1}^t \xi_i\) for \(t=1,\ldots,T\) forms a martingale, or equivalently~\(\E_P\!\left[\xi_i\!\mid\!\G_{i-1}\right] = 0\) for all \(i\). Assume that the following holds.
    \begin{enumerate}
        \item There exists \(\kappa\in\left[0,\frac{1}{4}\right]\) such that \(\left|\sum_{i=1}^T\E_P\!\left[\xi_i^2\!\mid\!\G_{i-1}\right]-1\right|\leq\kappa\) almost surely.
        \item There exists \(\rho\in\left(0,\frac{1}{2}\right]\) such that \(\E_P\!\left[\left|\xi_i\right|^3\!\mid\!\G_{i-1}\right] \leq \rho\E_P\!\left[\xi_i^2\!\mid\!\G_{i-1}\right]\) almost surely for all \(i\).
    \end{enumerate}
    Then the cdf of \(\sum_{i=1}^T\xi_i\) satisfies
    \begin{align*}
        \sup_{x\in\real} \left|P\!\left(\sum_{i=1}^T\xi_i\leq x\right) - \Phi\!\left(x\right)\right| \leq C\!\left(\rho\!\left|\ln\rho\right|+\sqrt{\kappa}\right),
    \end{align*}
    where \(\Phi\) is the cdf of the standard Gaussian distribution, and \(C\) is a universal constant.
\end{theorem}
Note that as $T$ grows, $v$ increases linearly while $\rho$ remains constant, indicating a convergence rate of $O\!\left(\frac{\ln T}{\sqrt{T}}\right)$. \citet{FAN20191028} showed that for martingale difference sequences satisfying the conditions in Theorem~\ref{thm:martingale-BE}, the order of convergence $\rho\!\left|\ln\rho\right|$ is optimal. To the best of our knowledge, this represents the best convergence rate among existing martingale Berry--Esseen variants, though it is weaker than the classical $O\!\left(\frac{1}{\sqrt{T}}\right)$ rate for i.i.d. sequences due to the martingale setting. Additional variants of the martingale Berry--Esseen theorem were derived by \citet{10.1214/aoms/1177696722} and \citet{article}.
\adaptiveCLT*
\begin{proof}
    We apply Theorem \ref{thm:martingale-BE} to the sequences of martingale differences
    \begin{align*}
        \frac{L_t\!\left(y_t\right) - \E_{y\sim P}\!\left[L_t\!\left(y_t\right)\!\mid\!y_{1:t-1}\right]}{\sqrt{v}}, y\sim P \quad\textup{and}\quad \frac{L_t\!\left(y_t\right) - \E_{y\sim Q}\!\left[L_t\!\left(y_t\right)\!\mid\!y_{1:t-1}\right]}{\sqrt{v}}, y\sim Q
    \end{align*}
    for \(t=1,\ldots,T\) to arrive at the inequalities
    \begin{align*}
        \sup_{x\in\real}\left|P\!\left[\frac{\sum_{t=1}^TL_t\!\left(y_t\right) - \sum_{t=1}^T\E_{y\sim P}\!\left[L_t\!\left(y_t\right)\!\mid\!y_{1:t-1}\right]}{\sqrt{v}}\leq x\right] - \Phi\!\left(x\right)\right| &\leq C\!\left(\frac{\rho}{\sqrt{v}}\!\left|\ln\!\frac{\rho}{\sqrt{v}}\right| + \sqrt\frac{\kappa}{v}\right), \\
        \sup_{x\in\real}\left|Q\!\left[\frac{\sum_{t=1}^TL_t\!\left(y_t\right) - \sum_{t=1}^T\E_{y\sim Q}\!\left[L_t\!\left(y_t\right)\!\mid\!y_{1:t-1}\right]}{\sqrt{v}}\leq x\right] - \Phi\!\left(x\right)\right| &\leq C\!\left(\frac{\rho}{\sqrt{v}}\!\left|\ln\!\frac{\rho}{\sqrt{v}}\right| + \sqrt\frac{\kappa}{v}\right).
    \end{align*}
    Denote \(L\!\left(y\right)=\sum_{t=1}^TL_t\!\left(y_t\right)\) and \(\Delta=C\!\left(\frac{\rho}{\sqrt{v}}\!\left|\ln\!\frac{\rho}{\sqrt{v}}\right| + \sqrt\frac{\kappa}{v}\right)\). By Conditions 1 and 2, the sums of conditional means are almost surely bounded,
    \begin{align*}
        m_1-\eta_1\leq \sum_{t=1}^T\E_{y\sim P}\!\left[L_t\!\left(y_t\right)\!\mid\!y_{1:t-1}\right] &\leq m_1+\eta_1, \\
        -m_2-\eta_2 \leq \sum_{t=1}^T\E_{y\sim Q}\!\left[L_t\!\left(y_t\right)\!\mid\!y_{1:t-1}\right] &\leq -m_2+\eta_2.
    \end{align*}
    Therefore for all \(x\in\real\),
    \begin{align}
        P\!\left(\frac{L\!\left(y\right)-m_1+\eta_1}{\sqrt{v}}\leq x\right) &\leq \Phi\!\left(x\right) + \Delta, \label{eq:cdfP_1}\\
        P\!\left(\frac{L\!\left(y\right)-m_1-\eta_1}{\sqrt{v}}\leq x\right) &\geq \Phi\!\left(x\right) - \Delta, \label{eq:cdfP_2}\\
        Q\!\left(\frac{L\!\left(y\right)+m_2+\eta_2}{\sqrt{v}}\leq x\right) &\leq \Phi\!\left(x\right) + \Delta, \label{eq:cdfQ_1}\\
        Q\!\left(\frac{L\!\left(y\right)+m_2-\eta_2}{\sqrt{v}}\leq x\right) &\geq \Phi\!\left(x\right) - \Delta. \label{eq:cdfQ_2}
    \end{align}
    Now, we convert these inequalities to their trade-off function form. Let \(\alpha\in\left(0,1-\Delta\right)\) and let \(x\in\real\) such that the  test
    \begin{align*}
        \phi\!\left(y\right) = \begin{cases}
            1 & \text{if }L\!\left(y\right)\leq x, \\
            0 & \text{otherwise}
        \end{cases}
    \end{align*}
    for the hypothesis testing problem \(H_0:y\sim P\) vs.~\(H_1:y\sim Q\) has an FPR \(P\!\left(L\!\left(y\right)\leq x\right)=\alpha\).\footnote{For those tests with $\phi(y)<1$ when $L(y)=x$, \eqref{eq:CLT-tradeoff} can be proven once \eqref{eq:CLT-tradeoff} is established for those tests with $\phi(y)=1$ when $L(y)=x$.} By \eqref{eq:cdfP_2},
    \begin{align*}
        \Phi\!\left(\frac{x-m_1-\eta_1}{\sqrt{v}}\right) \leq P\!\left(\frac{L\!\left(y\right)-m_1-\eta_1}{\sqrt{v}}\leq \frac{x-m_1-\eta_1}{\sqrt{v}}\right) + \Delta = \alpha+\Delta.
    \end{align*}
    Since \(\Phi^{-1}\) is an increasing function,
    \begin{align*}
        \frac{x-m_1-\eta_1}{\sqrt{v}}\leq \Phi^{-1}\!\left(\alpha+\Delta\right),
    \end{align*}
    which implies that
    \begin{align*}
        \frac{x+m_2+\eta_2}{\sqrt{v}}\leq \Phi^{-1}\!\left(\alpha+\Delta\right) + \frac{m_1+m_2}{\sqrt{v}} + \frac{\eta_1+\eta_2}{\sqrt{v}}.
    \end{align*}
    Combining \eqref{eq:cdfQ_1} with the fact that \(\Phi\) is increasing,
    \begin{multline*}
        Q\!\left(L\!\left(y\right)\leq x\right) = Q\!\left(\frac{L\!\left(y\right)+m_2+\eta_2}{\sqrt{v}}\leq \frac{x+m_2+\eta_2}{\sqrt{v}}\right)
        \leq \Phi\!\left(\frac{x+m_2+\eta_2}{\sqrt{v}}\right)+\Delta \\
        \leq \Phi\!\left(\Phi^{-1}\!\left(\alpha+\Delta\right) + \frac{m_1+m_2}{\sqrt{v}} + \frac{\eta_1+\eta_2}{\sqrt{v}}\right)+\Delta.
    \end{multline*}
    Therefore the FNR \(Q\!\left(L\!\left(y\right)>x\right)\) satisfies
    \begin{align*}
        Q\!\left(L\!\left(y\right)>x\right) \geq 1 - \Phi\!\left(\Phi^{-1}\!\left(\alpha+\Delta\right) + \frac{m_1+m_2}{\sqrt{v}} + \frac{\eta_1+\eta_2}{\sqrt{v}}\right) - \Delta.
    \end{align*}
    We modify this into the final expression using the relationship \(1-\Phi\!\left(x\right) = \Phi\!\left(-x\right)\) and \(\Phi^{-1}\!\left(\alpha\right) = -\Phi^{-1}\!\left(1-\alpha\right)\),
    \begin{align*}
        Q\!\left(L\!\left(y\right)>x\right) \geq  \Phi\!\left(\Phi^{-1}\!\left(1-\left(\alpha+\Delta\right)\right) - \frac{m_1+m_2}{\sqrt{v}} - \frac{\eta_1+\eta_2}{\sqrt{v}}\right) - \Delta = G_{\mu+\phi}\!\left(\alpha+\Delta\right) - \Delta,
    \end{align*}
    where \(\mu=\frac{m_1+m_2}{\sqrt{v}}\) and \(\phi=\frac{\eta_1+\eta_2}{\sqrt{v}}\). Finally, \(T\!\left[P,Q\right]\!\left(\alpha\right) = Q\!\left(L\!\left(y\right)>x\right)\) by the Neyman-Pearson lemma, so we conclude
    \begin{align}
        T\!\left[P,Q\right]\!\left(\alpha\right) \geq G_{\mu+\phi}\!\left(\alpha+\Delta\right) - \Delta \text{ for all }\alpha\in\left[0,1-\Delta\right]. \label{eq:CLT-tradeoff}
    \end{align}
    The bound is extended from \(\alpha\in\left(0,1-\Delta\right)\) to \(\alpha\in\left[0,1-\Delta\right]\) due to the continuity of trade-off functions. The proof for the remaining bound \(T\!\left[P,Q\right]\!\left(\alpha\right) \leq G_{\mu-\phi}\!\left(\alpha-\Delta\right) + \Delta\) is identical, which utilizes the remaining inequalities \eqref{eq:cdfP_1} and \eqref{eq:cdfQ_2}.
\end{proof}

\subsection{Proofs for Demonstration of Approximate GDP Filter}
The limit conditions in Theorem \ref{thm:validity-approx-GDP-filter} are:
\begin{enumerate}
    \item Near \(\tilde{q}=0\), \(T\bar{q}^2=\Theta\!\left(1\right)\) as \(\bar{q}\to0\).
    \item Near \(\tilde{q}=1\), \(T\frac{1}{\bar{\sigma}^2}=\Theta\!\left(1\right)\) and \(\hat{q}=1-\bar{q}=O\!\left(\frac{1}{\bar{\sigma}^2}\right)\) as \(\bar{\sigma}\to\infty\).
\end{enumerate}
\validapproxGDPfilter*
\begin{proof}
    Fix a remove pair \(S,S^-\) that differs in some datapoint \(X\). After \(T\) steps, the sum of conditional expectations of PLRVs between \(S\) and \(S^-\) may remain below \(B\), because
    \begin{enumerate}
        \item The filter takes into account an amount of \(\frac{1}{2}q_t^2\!\left(\e^{\frac{1}{\sigma_t^2}}-1\right)\) or \(\frac{1}{2}q_t^2\frac{1}{\sigma_t^2}\), which upper bounds the actual amount \(\frac{1}{2}q_t^2\!\left(\e^{\frac{\norm{\grave{g_t}\left(X\right)}^2}{\sigma_t^2C_t^2}}-1\right)\) or \(\frac{1}{2}q_t^2\frac{\norm{\grave{g_t}\left(X\right)}^2}{\sigma_t^2C_t^2}\) spent by the \(S,S^-\).
        \item After \(T\) steps, the total budget accounted for may not have reached \(B\).
    \end{enumerate}
    In this case, we can artificially create an extra mechanism \(\M_{T+1}\) at step \(T+1\), tailored to the pair \(S,S^-\), that adaptively consumes the remaining budget, and then apply projection as a post-processing to extract the first \(T\) outputs from \(T+1\) outputs. Therefore, w.l.o.g.~we can assume that the total budget \(B\) is almost surely exhausted by \(S,S^-\) after \(T\) steps.

    We apply Theorem \ref{thm:fully-adaptive-CLT} in combination with Theorems \ref{thm:approx-q=0} and \ref{thm:approx-q=1} to the pair \(S,S^-\):
    \begin{enumerate}
        \item In the regime near \(\tilde{q}=0\), as \(\bar{q}\to0\), the approximate GDP guarantee is 
        \begin{align*}
            \frac{B+B}{\sqrt{2B}}+T\cdot O\!\left(\bar{q}^3\right)=\sqrt{2B}+O\!\left(\bar{q}\right),
        \end{align*}
        with \(\Delta\)-error
        \begin{align*}
            O\!\left(\bar{q}\right)\left|\ln O\!\left(\bar{q}\right)\right|+\sqrt{\frac{T\cdot O\!\left(\bar{q}^3\right)}{2B}} = O\!\left(\sqrt{\bar{q}}\right).
        \end{align*}
        \item In the regime near \(\tilde{q}=1\), as \(\bar{\sigma}\to\infty\) the approximate GDP guarantee is 
        \begin{align*}
            \frac{B+B}{\sqrt{2B}}+T\hat{q}\cdot O\!\left(\frac{1}{\bar{\sigma}^4}\right)+T\cdot O\!\left(\hat{q}^2\right)=\sqrt{2B}+O\!\left(\frac{1}{\bar{\sigma}^2}\right),
        \end{align*} 
        with \(\Delta\)-error
        \begin{align*}
            O\!\left(\frac{1}{\bar{\sigma}}\right)\left|\ln O\!\left(\frac{1}{\bar{\sigma}}\right)\right|+\sqrt{\frac{T\hat{q}\cdot O\!\left(\frac{1}{\bar{\sigma}^4}\right)+T\cdot O\!\left(\hat{q}^2\right)}{2B}} = O\!\left(\frac{\ln\bar{\sigma}}{\bar{\sigma}}\right).
        \end{align*}
    \end{enumerate}
    Since the \(\Delta\)-approximate \(\mu\)-GDP guarantee is symmetric, the privacy guarantees above apply to both trade-off functions \(T\!\left[\M_{1:t}\!\left(S\right),\M_{1:t}\!\left(S^-\right)\right]\) and \(T\!\left[\M_{1:t}\!\left(S^-\right),\M_{1:t}\!\left(S\right)\right]\).
\end{proof}

\individualapproxGDPfilter*
\begin{proof}
    The proof mirrors the proof idea of \citet{feldman_individual_2021}. Fix a remove pair \(S=\left(X_1,\ldots,X_n,X\right)\) and \(S^-=\left(X_1,\ldots,X_n\right)\). At each step \(t\), the subset of datapoints in the input dataset that have positive remaining budget, and thus are still participating as proper inputs, is called the active subset. Denote \(\tilde{S}_t\) the active subset of \(S\) at step \(t\), and denote \(\tilde{S}^-_t\) that of \(S^-\).

    At step \(t\), conditioned on a given sequence of previous outputs \(y_{1:t-1}\), the total budget spent so far for each individual \(X_j\) can be computed, and hence it is deterministic which datapoint remains in \(\tilde{S}_t\) and \(\tilde{S}^-_t\). Consequently, \(\tilde{S}_t\) and \(\tilde{S}^-_t\) are either identical, or differ only in the datapoint \(X\).

    Therefore, the pair of PLRVs of \(S,S^-\) at step \(t\), conditioned on a given sequence of previous outputs \(y_{1:t-1}\), is the pair of PLRVs of \(\tilde{S}_t,\tilde{S}^-_t\), which implies that Theorems \ref{thm:approx-q=0} and \ref{thm:approx-q=1} still apply. We then apply Theorem \ref{thm:fully-adaptive-CLT} and obtain identical results to Theorem \ref{thm:validity-approx-GDP-filter}.

    This proof can be seen as an extension of the proof to Theorem \ref{thm:validity-approx-GDP-filter} by allowing the pair of active subsets adaptively shrink through composition.
\end{proof}

\section{Privacy Loss Moments: Approximations}\label{appx:approximations}

For a subsampled Gaussian mechanism $\M$, let $P = \M(S)$ and $Q = \M(S^-)$ denote
the output distributions for an arbitrary remove pair $S, S^-$ differing by a single
datapoint $X$. In the regimes $q \ll 1$ and $q \approx 1$, the trade-off function is
nearly symmetric and close to Gaussian, which gives rise to the asymptotic relationships
\[
  \E_P\!\left[\ln\tfrac{P}{Q}\right] \approx -\E_Q\!\left[\ln\tfrac{P}{Q}\right],
  \qquad
  \Var_P\!\left[\ln\tfrac{P}{Q}\right] \approx \Var_Q\!\left[\ln\tfrac{P}{Q}\right],
\]
and, most notably,
\[
  \frac{\E_P\!\left[\ln\frac{P}{Q}\right]}{\Var_P\!\left[\ln\frac{P}{Q}\right]}
  \;\approx\;
  \frac{-\E_Q\!\left[\ln\frac{P}{Q}\right]}{\Var_Q\!\left[\ln\frac{P}{Q}\right]}
  \;\approx\; \frac{1}{2}.
\]

To build intuition, consider each regime in turn. When $q \approx 1$, the subsampled
Gaussian mechanism is close to the standard Gaussian mechanism, whose PLRV has
mean-to-variance ratio exactly $\nicefrac{1}{2}$. When $q \ll 1$, the PLRV mean is
small and the variance $\Var[L]$ is well approximated by the second moment $\E[L^2]$. Setting $L =
\ln\nicefrac{P}{Q}$ and applying the Taylor expansions $e^y \approx 1 + y +
\frac{y^2}{2}$ and $e^{-y} \approx 1 - y + \frac{y^2}{2}$ yields
\begin{align*}
  1 = \E_P\!\left[\tfrac{Q}{P}\right] = \E_P\!\left[e^{-L}\right]
    &\approx 1 - \E_P[L] + \tfrac{1}{2}\E_P[L^2], \\
  1 = \E_Q\!\left[\tfrac{P}{Q}\right] = \E_Q\!\left[e^{L}\right]
    &\approx 1 + \E_Q[L] + \tfrac{1}{2}\E_Q[L^2].
\end{align*}

\begin{theorem}\label{thm:approx-q=0}
  Fix a datapoint $X$. Let $h$ be a deterministic function on datasets with
  $h(\emptyset) = 0$, and let $\mu = \|h((X))\|$. Let $\M$ be the subsampled Gaussian
  mechanism derived from $h$, with sampling rate $q$ and noise level $\sigma$. For an
  arbitrary remove-neighbor pair $S = (X_1, \ldots, X_n, X)$ and $S^- = (X_1, \ldots,
  X_n)$, set $P = \M(S)$ and $Q = \M(S^-)$. As $q \to 0$:
  \begin{compactenum}
    \item $\E_P\!\left[\ln\frac{P}{Q}\right]
          = \frac{1}{2}q^2\!\left(e^{\mu^2/\sigma^2} - 1\right) + O(q^3)$,
    \item $\E_Q\!\left[\ln\frac{P}{Q}\right]
          = -\frac{1}{2}q^2\!\left(e^{\mu^2/\sigma^2} - 1\right) + O(q^3)$,
    \item $\Var_P\!\left[\ln\frac{P}{Q}\right]
          = q^2\!\left(e^{\mu^2/\sigma^2} - 1\right) + O(q^3)$,
    \item $\Var_Q\!\left[\ln\frac{P}{Q}\right]
          = q^2\!\left(e^{\mu^2/\sigma^2} - 1\right) + O(q^3)$,
    \item $\E_P\!\left|\ln\frac{P}{Q} - \E_P\!\left[\ln\frac{P}{Q}\right]\right|^3
          \leq q^3\!\left(e^{3\mu^2/\sigma^2} + 1\right) + O(q^4)$,
    \item $\E_Q\!\left|\ln\frac{P}{Q} - \E_Q\!\left[\ln\frac{P}{Q}\right]\right|^3
          \leq q^3\!\left(e^{3\mu^2/\sigma^2} + 1\right) + O(q^4)$.
  \end{compactenum}
  Each error term $O(q^3)$ or $O(q^4)$ may depend on $q$, $\sigma$, $X$, and
  $X_1, \ldots, X_n$.
\end{theorem}

\begin{theorem}\label{thm:approx-q=1}
  Fix a datapoint $X$. Let $h$ be a summation $h(X_1, \ldots, X_n) = \sum_i g(X_i)$
  for some function $g$, and let $\mu = \|g(X)\|$. Let $\M$ be the subsampled Gaussian
  mechanism derived from $h$, with sampling rate $q$ and noise level $\sigma$, and write
  $\hat{q} = 1 - q$. For an arbitrary remove-neighbor pair $S = (X_1, \ldots, X_n, X)$
  and $S^- = (X_1, \ldots, X_n)$, set $P = \M(S)$ and $Q = \M(S^-)$. Then:
  \begin{compactenum}
    \item $\E_P\!\left[\ln\frac{P}{Q}\right]
          = \frac{1}{2}q^2\frac{\mu^2}{\sigma^2}
            + \hat{q}\cdot O\!\left(\sigma^{-4}\right) + O(\hat{q}^2)$,
    \item $\E_Q\!\left[\ln\frac{P}{Q}\right]
          = -\frac{1}{2}q^2\frac{\mu^2}{\sigma^2}
            + \hat{q}\cdot O\!\left(\sigma^{-4}\right) + O(\hat{q}^2)$,
    \item $\Var_P\!\left[\ln\frac{P}{Q}\right]
          = q^2\frac{\mu^2}{\sigma^2}
            + \hat{q}\cdot O\!\left(\sigma^{-4}\right) + O(\hat{q}^2)$,
    \item $\Var_Q\!\left[\ln\frac{P}{Q}\right]
          = q^2\frac{\mu^2}{\sigma^2}
            + \hat{q}\cdot O\!\left(\sigma^{-4}\right) + O(\hat{q}^2)$,
    \item $\E_P\!\left|\ln\frac{P}{Q} - \E_P\!\left[\ln\frac{P}{Q}\right]\right|^3
          \leq (2 - q^3)\sqrt{\tfrac{8}{\pi}}\,\frac{\mu^3}{\sigma^3}
            + \hat{q}\cdot O\!\left(\sigma^{-4}\right) + O(\hat{q}^2)$,
    \item $\E_Q\!\left|\ln\frac{P}{Q} - \E_Q\!\left[\ln\frac{P}{Q}\right]\right|^3
          \leq (2 - q^3)\sqrt{\tfrac{8}{\pi}}\,\frac{\mu^3}{\sigma^3}
            + \hat{q}\cdot O\!\left(\sigma^{-4}\right) + O(\hat{q}^2)$.
  \end{compactenum}
  Here $O(\hat{q}^2)$ is in the limit $\hat{q} \to 0$, and $O(\sigma^{-4})$ is in the
  limit $\sigma \to \infty$. All error terms may depend on $q$, $\sigma$, $X$, and
  $X_1, \ldots, X_n$.
\end{theorem}

\subsection{Numerical Illustration of Moment Approximations}

We numerically illustrate the asymptotic approximations of Theorems~\ref{thm:approx-q=0} and \ref{thm:approx-q=1} for the means and variances of the PLRVs $L = \ln(P/Q)$ of the Poisson subsampled Gaussian mechanism. The accurate values are computed by direct numerical integration on a fine grid, using the one-dimensional distributions $P= q \cdot \mathcal{N}(1,\sigma^2) + (1-q) \cdot \mathcal{N}(0,\sigma^2)$ and $Q= \mathcal{N}(0,\sigma^2)$.

\para{Small-$q$ regime.}
Figures~\ref{fig:smallq_EP} and~\ref{fig:smallq_VarP} illustrate the accuracy of approximations 1.~and 3.~of Theorem~\ref{thm:approx-q=0}, respectively. We plot the relative error between the numerical quantity and its $O(q^2)$ approximation as $q\to0$. As predicted by the theorem, the error decreases polynomially in $q$, and the accuracy improves for larger noise levels $\sigma$.

\para{Large-$q$ regime.}
Figures~\ref{fig:largeq_EP} and~\ref{fig:largeq_VarP} illustrate the accuracy of approximations 1.~and 3.~of Theorem~\ref{thm:approx-q=1}. Here $q$ is close to $1$, and we again plot relative errors as functions of $q$ for several values of $\sigma$. As predicted by the theorem, the relative error decreases as $\hat q = 1-q \to 0$ and $\sigma$ increases.

The numerical illustrations of Figures~\ref{fig:smallq_EP} to~\ref{fig:largeq_VarP} confirm that the PLRV mean and variance satisfy the predicted scalings, and that the resulting mean-to-variance ratios approach $1/2$, as stated by Theorems~\ref{thm:approx-q=0} and \ref{thm:approx-q=1}.

\begin{figure}[H]
	\centering
	\includegraphics[width=0.85\linewidth]{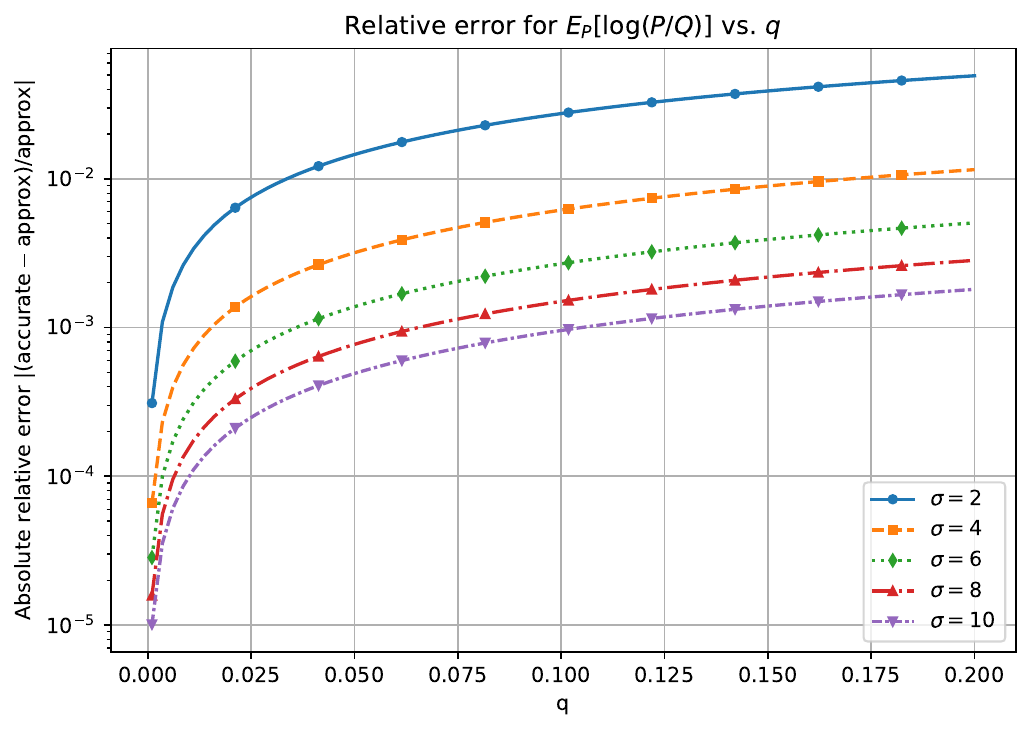}
	\caption{Relative error of
	$\E_P[\ln(P/Q)]$ versus the approximation
	$\frac12 q^2(\exp(\mu^2/\sigma^2)-1)$ as $q\to0$,
	illustrating the approximation of 1.~in Theorem~\ref{thm:approx-q=0}.}
	\label{fig:smallq_EP}
\end{figure}

\vspace{15pt}
\begin{figure}[H]
	\centering
	\includegraphics[width=0.85\linewidth]{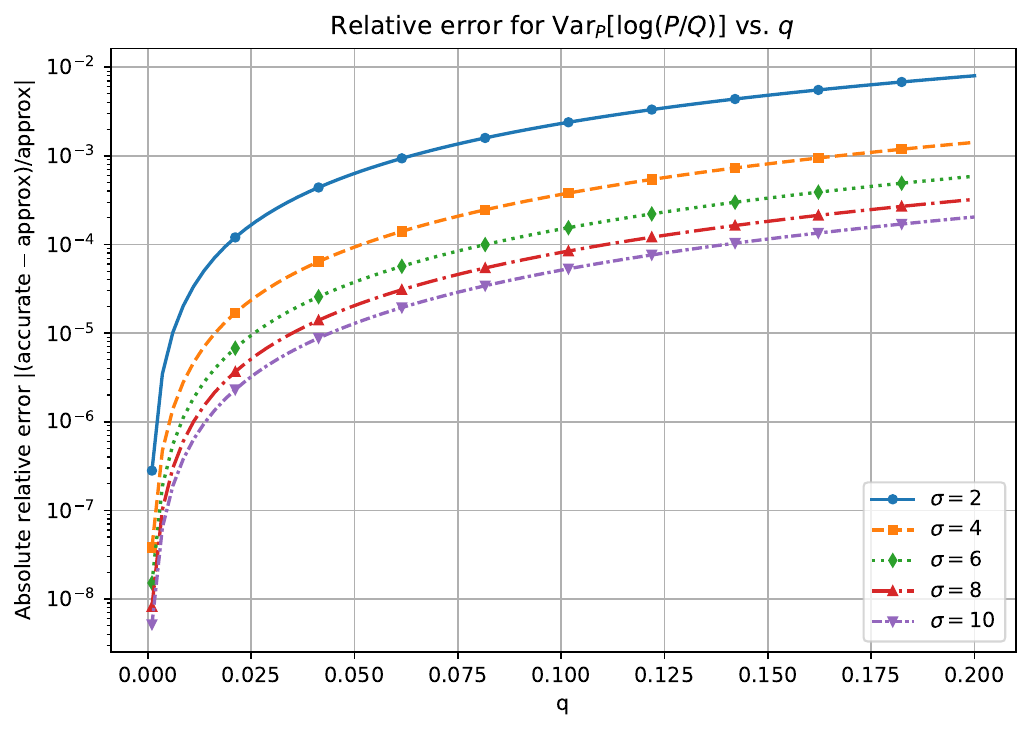}
	\caption{Relative error of
	$\Var_P[\ln(P/Q)]$ versus the approximation
	$q^2(\exp(\mu^2/\sigma^2)-1)$ as $q\to0$,
	illustrating the approximation 3.~in Theorem~\ref{thm:approx-q=0}.}
	\label{fig:smallq_VarP}
\end{figure}

\begin{figure}[H]
	\centering
	\includegraphics[width=0.85\linewidth]{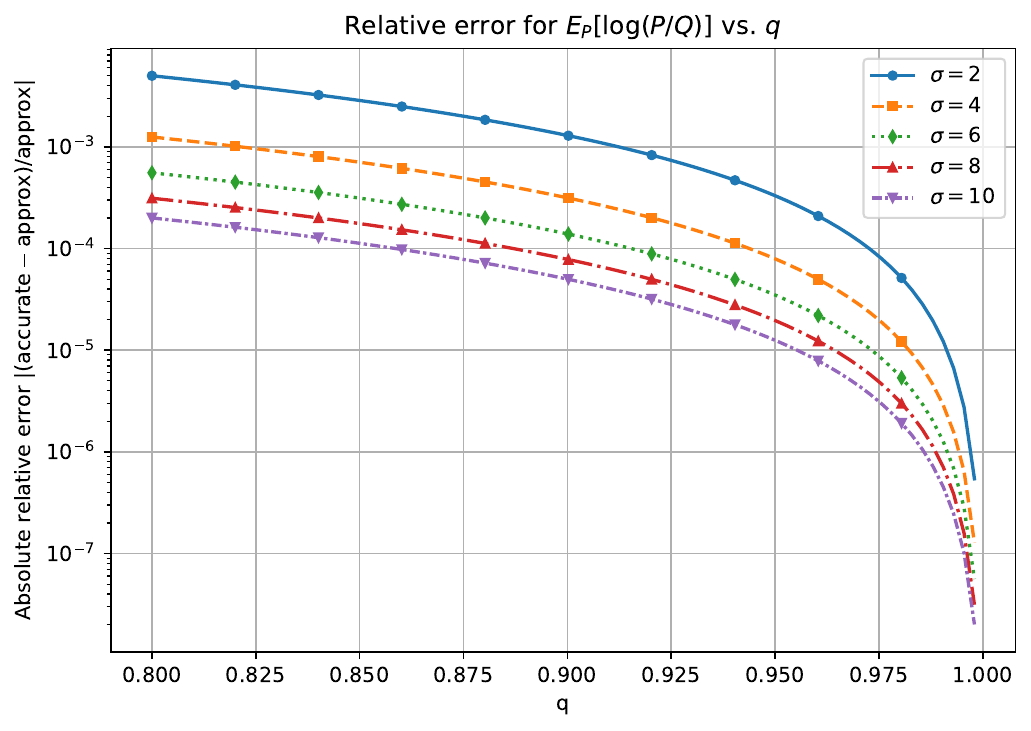}
	\caption{Relative error of
	$\E_P[\ln(P/Q)]$ versus the approximation
	$\frac12 q^2 \mu^2/\sigma^2$ for $q\approx1$,
	illustrating the approximation 1.~in Theorem~\ref{thm:approx-q=1}.}
	\label{fig:largeq_EP}
\end{figure}

\vspace{15pt}
\begin{figure}[H]
	\centering
	\includegraphics[width=0.85\linewidth]{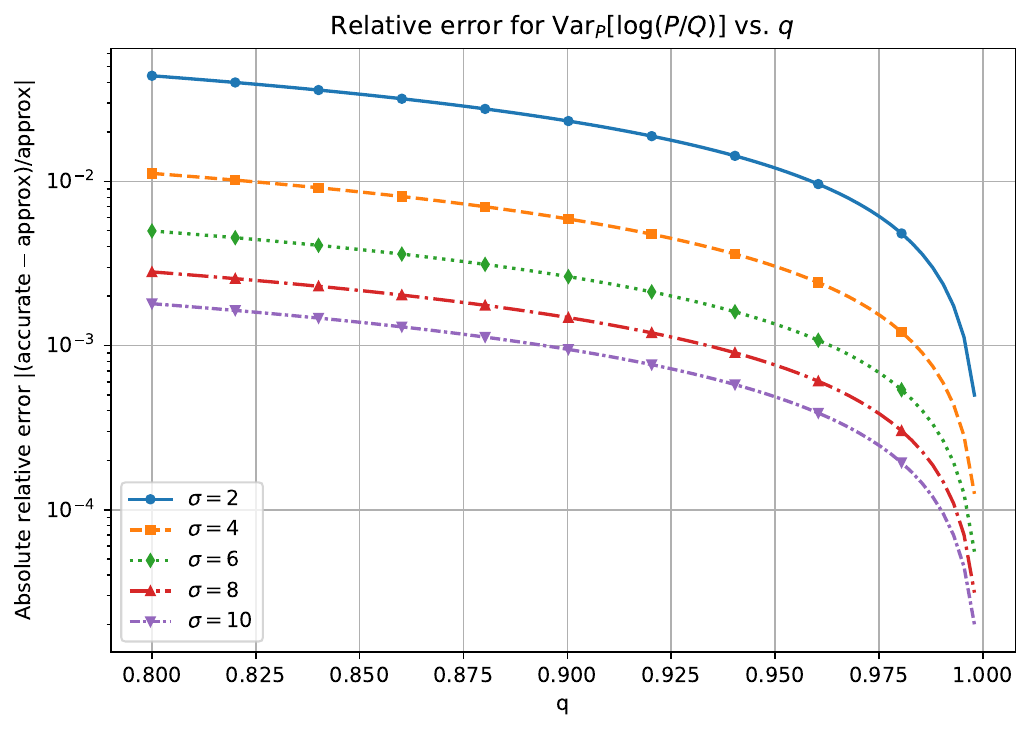}
	\caption{Relative error of
	$\Var_P[\log(P/Q)]$ versus the approximation
	$q^2 \mu^2/\sigma^2$ for $q\approx1$,
	illustrating the approximation 3.~in Theorem~\ref{thm:approx-q=1}.}
	\label{fig:largeq_VarP}
\end{figure}

\subsection{Derivations Using Taylor Expansions}

\begin{proof}[Proof of Theorem~\ref{thm:approx-q=0}]
All big-$O$ terms in this proof are taken in the limit $q\to 0$.
Throughout we use the following notation.
\begin{itemize}
  \item $P_0 = \mathcal{N}(0, \sigma^2 \mathbf{I})$.
  \item $a(y) = \dfrac{\mathcal{N}(h((X)),\,\sigma^2\mathbf{I})(y)}{P_0(y)} - 1$.
  \item $a_j(y) = \dfrac{\mathcal{N}(h((X_j)),\,\sigma^2\mathbf{I})(y)}{P_0(y)} - 1$,
        \quad $j = 1,\ldots,n$.
  \item $b_k(y) = \dfrac{\mathcal{N}(h((X_{j_1},X_{j_2})),\,\sigma^2\mathbf{I})(y)}{P_0(y)} - 1$,
        where $X_{j_1}, X_{j_2}$ is an arbitrary distinct pair from
        $(X,X_1,\ldots,X_n)$. The index $k$ runs from $1$ to $\tfrac{n(n+1)}{2}$
        to cover all such pairs; by convention, $k = 1, \ldots, \tfrac{n(n-1)}{2}$
        covers pairs that exclude $X$, and the remaining indices cover pairs that
        include $X$.
\end{itemize}

\medskip
\noindent\textbf{Step 1: Density approximations.}
The density of $P = \mathcal{M}(S)$ is a Gaussian mixture whose weights depend on
$q$. Applying the Taylor expansion
$(1-q)^\alpha = 1 - \alpha q + \tfrac{\alpha(\alpha-1)}{2}q^2 + q^3 O(1)$
gives two useful approximations:
\begin{align}
  \mathcal{M}(S)(y)
  &= \Bigl[1 + q\!\left(\sum_{j=1}^n a_j(y) + a(y)\right)\notag\\
  &\qquad + q^2\!\left(-n\sum_{j=1}^n a_j(y) - na(y)
                    + \sum_{k=1}^{\frac{n(n+1)}{2}} b_k(y)\right)
       + q^3 O(1)
     \Bigr] P_0(y),
  \label{eq:approx_M(S)_q~0_o(q^2)} \\[4pt]
  \mathcal{M}(S)(y)
  &= \Bigl[
       1 + q\!\left(\sum_{j=1}^n a_j(y) + a(y)\right) + q^2 O(1)
     \Bigr] P_0(y).
  \label{eq:approx_M(S)_q~0_o(q)}
\end{align}
The same expansion applied to $Q = \mathcal{M}(S^-)$ yields
\begin{align}
  \mathcal{M}(S^-)(y)
  &= \Bigl[1 + q\sum_{j=1}^n a_j(y)\notag\\
  &\qquad + q^2\!\left(-(n-1)\sum_{j=1}^n a_j(y)
                    + \sum_{k=1}^{\frac{n(n-1)}{2}} b_k(y)\right)
       + q^3 O(1)
     \Bigr] P_0(y),
  \label{eq:approx_M(S-)_q~0_o(q^2)} \\[4pt]
  \mathcal{M}(S^-)(y)
  &= \Bigl[
       1 + q\sum_{j=1}^n a_j(y) + q^2 O(1)
     \Bigr] P_0(y).
  \label{eq:approx_M(S-)_q~0_o(q)}
\end{align}

\medskip
\noindent\textbf{Step 2: Approximation of $\E_P[\ln(P/Q)]$.}
Using \eqref{eq:approx_M(S)_q~0_o(q^2)},
\eqref{eq:approx_M(S)_q~0_o(q)}, and
\eqref{eq:approx_M(S-)_q~0_o(q^2)},
\begin{align*}
  \E_P\!\left[\ln\frac{P}{Q}\right]
  &= \E_{y\sim P_0}\!\left[\frac{P(y)}{P_0(y)}\ln\frac{P(y)}{Q(y)}\right] \\
  &= \E_{y\sim P_0}\!\left[\left(1 + q\!\left(\sum_{j=1}^n a_j + a\right) + q^2 O(1)\right)\right.\\
  &\qquad\left.\cdot\ln\frac{
      1 + q\!\left(\sum_j a_j + a\right)
        + q^2\!\left(-n\sum_j a_j - na + \sum_{k=1}^{\frac{n(n+1)}{2}} b_k\right)
        + q^3 O(1)
    }{
      1 + q\sum_j a_j
        + q^2\!\left(-(n-1)\sum_j a_j + \sum_{k=1}^{\frac{n(n-1)}{2}} b_k\right)
        + q^3 O(1)
    }
  \right].
\end{align*}
Because $\E_{y\sim P_0}[a(y)] = \E_{y\sim P_0}[a_j(y)] = \E_{y\sim P_0}[b_k(y)] = 0$
for all $j,k$, the $q^2$ corrections in the numerator and denominator contribute zero
expectation and may be dropped, giving
\begin{align*}
  \E_P\!\left[\ln\frac{P}{Q}\right]
  &= \E_{y\sim P_0}\!\left[
    \left(1 + q\!\left(\sum_j a_j + a\right) + q^2 O(1)\right)
    \ln\frac{
      1 + q\!\left(\sum_j a_j + a\right) + q^3 O(1)
    }{
      1 + q\sum_j a_j + q^3 O(1)
    }
  \right].
\end{align*}

\pagebreak
We simplify the log in two sub-steps.
\begin{enumerate}
  \item Apply $\tfrac{1}{1+x} = 1 - x + x^2 + O(x^3)$ to the denominator.
  \item Apply $\ln(1+x) = x - \tfrac{1}{2}x^2 + O(x^3)$ to the resulting expression.
\end{enumerate}
Sub-step~1 yields
\begin{multline*}
  \E_P\!\left[\ln\frac{P}{Q}\right]
  = \E_{y\sim P_0}\!\left[\left(1 + q\!\left(\sum_j a_j + a\right) + q^2 O(1)\right)\right.\\
  \qquad\qquad\ \left.\cdot\ln\!\left(
       \bigl(1 + q\!\left(\sum_j a_j + a\right) + q^3 O(1)\bigr)
       \bigl(1 - q\sum_j a_j + q^2\!\left(\sum_j a_j\right)^{\!2} + q^3 O(1)\bigr)
     \right)
   \right] \\
  = \E_{y\sim P_0}\!\left[
     \left(1 + q\!\left(\sum_j a_j + a\right) + q^2 O(1)\right)
     \ln\!\left(1 + qa(y) - q^2 a(y)\sum_j a_j(y) + q^3 O(1)\right)
   \right].
\end{multline*}
Sub-step~2 then gives
\begin{align}
  \E_P\!\left[\ln\frac{P}{Q}\right]
  &= \E_{y\sim P_0}\!\left[
     \left(1 + q\!\left(\sum_j a_j + a\right) + q^2 O(1)\right)\right.\\
  &\qquad\qquad\left.\cdot
     \left(
       qa(y)
       - q^2\!\left(a(y)\sum_j a_j(y) + \tfrac{1}{2}a(y)^2\right)
       + q^3 O(1)
     \right)
   \right] \nonumber\\
  &= q\,\E_{y\sim P_0}[a(y)]
     + \tfrac{1}{2}q^2\,\E_{y\sim P_0}[a(y)^2]
     + O(q^3) \nonumber\\
  &= \tfrac{1}{2}q^2\,\E_{y\sim P_0}[a(y)^2] + O(q^3) \nonumber\\
  &= \tfrac{1}{2}q^2\!\left(e^{\mu^2/\sigma^2} - 1\right) + O(q^3).
  \label{eq:target_EP_q->0}
\end{align}
In the second equality we used that $O(1)$ is integrable, so
$\E_{y\sim P_0}[q^3 O(1)] = O(q^3)$. In the third equality we used
$\E_{y\sim P_0}[a(y)] = 0$.

The identical argument applied to $\E_Q[\ln(P/Q)]$ yields
\begin{align*}
  \E_Q\!\left[\ln\frac{P}{Q}\right]
  &= \E_{y\sim P_0}\!\left[
     \left(1 + q\sum_j a_j + q^2 O(1)\right)\right.\\
  &\qquad\qquad\left.\cdot
     \left(
       qa(y)
       - q^2\!\left(a(y)\sum_j a_j(y) + \tfrac{1}{2}a(y)^2\right)
       + q^3 O(1)
     \right)
   \right] \\
  &= q\,\E_{y\sim P_0}[a(y)]
     - \tfrac{1}{2}q^2\,\E_{y\sim P_0}[a(y)^2]
     + O(q^3) \\
  &= -\tfrac{1}{2}q^2\!\left(e^{\mu^2/\sigma^2} - 1\right) + O(q^3).
\end{align*}

\medskip
\noindent\textbf{Step 3: Variance.}
Since $\E_P[\ln(P/Q)] = O(q^2)$, we have
\[
  \Var_P\!\left[\ln\frac{P}{Q}\right]
  = \E_P\!\left[\left(\ln\frac{P}{Q}\right)^{\!2}\right]
    - \left(\E_P\!\left[\ln\frac{P}{Q}\right]\right)^{\!2}
  = \E_P\!\left[\left(\ln\frac{P}{Q}\right)^{\!2}\right] + O(q^4),
\]
so it suffices to approximate $\E_P[(\ln(P/Q))^2]$. The same two sub-steps give
\begin{align*}
  \E_P\!\left[\left(\ln\frac{P}{Q}\right)^{\!2}\right]
  &= \E_{y\sim P_0}\!\left[
     \left(1 + q\!\left(\sum_j a_j + a\right) + q^2 O(1)\right)\right.\\
  &\qquad\qquad\left.\cdot
     \left(
       qa(y)
       - q^2\!\left(a(y)\sum_j a_j(y) + \tfrac{1}{2}a(y)^2\right)
       + q^3 O(1)
     \right)^{\!2}
   \right] \\
  &= \E_{y\sim P_0}\!\left[
     \left(1 + q\!\left(\sum_j a_j + a\right) + q^2 O(1)\right)
     \left(q^2 a(y)^2 + q^3 O(1)\right)
   \right] \\
  &= q^2\,\E_{y\sim P_0}[a(y)^2] + O(q^3) \\
  &= q^2\!\left(e^{\mu^2/\sigma^2} - 1\right) + O(q^3).
\end{align*}
Hence $\Var_P[\ln(P/Q)] = q^2(e^{\mu^2/\sigma^2}-1) + O(q^3)$, and the case of
$\Var_Q[\ln(P/Q)]$ is identical.

\medskip
\noindent\textbf{Step 4: Third absolute central moment.}
\begin{align*}
  \E_{P}\!\left|\ln\frac{P}{Q} - \E_P\!\left[\ln\frac{P}{Q}\right]\right|^3
  &= \E_{y\sim P_0}\!\left[(1 + qO(1))\cdot\left|
       \ln\frac{
         1 + q(\sum_j a_j + a) + q^2 O(1)
       }{
         1 + q\sum_j a_j + q^2 O(1)
       }
  \right.\right.\\
  &\qquad\qquad\qquad\qquad\qquad\qquad\left.\left. - \tfrac{1}{2}q^2(e^{\mu^2/\sigma^2}-1) + q^3 O(1)
     \right|^3
   \right] \\
  &= \E_{y\sim P_0}\!\left[
     (1 + qO(1))\,\bigl|qa(y) + q^2 O(1)\bigr|^3
   \right] \\
  &= \E_{y\sim P_0}\!\left[\bigl|qa(y) + q^2 O(1)\bigr|^3\right] + O(q^4).
\end{align*}
Applying the triangle inequality and using $|O(1)| = O(1)$,
\begin{align*}
  \E_{P}\!\left|\ln\frac{P}{Q} - \E_P\!\left[\ln\frac{P}{Q}\right]\right|^3
  &\leq \E_{y\sim P_0}\!\left[\left(q|a(y)| + q^2|O(1)|\right)^3\right] + O(q^4) \\
  &= q^3\,\E_{y\sim P_0}|a(y)|^3 + O(q^4).
\end{align*}
Applying the inequality $|x-1|^3 \leq x^3+1$ for $x\geq0$ to $a(y) = \exp\!\left(\tfrac{\mu\cdot y}{\sigma^2} - \tfrac{\mu^2}{2\sigma^2}\right) - 1$,
\begin{align*}
  \E_{P}\!\left|\ln\frac{P}{Q} - \E_P\!\left[\ln\frac{P}{Q}\right]\right|^3
  &\leq q^3\,\E_{y\sim P_0}\!\left|e^{\frac{\mu\cdot y}{\sigma^2} - \frac{\mu^2}{2\sigma^2}} - 1\right|^3
     + O(q^4) \\
  &\leq q^3\,\E_{y\sim P_0}\!\left[e^{\frac{3\mu\cdot y}{\sigma^2} - \frac{3\mu^2}{2\sigma^2}} + 1\right]
     + O(q^4) \\
  &= q^3\!\left(e^{3\mu^2/\sigma^2} + 1\right) + O(q^4).
\end{align*}
The case of $\E_{Q}|\ln(P/Q) - \E_Q[\ln(P/Q)]|^3$ is identical.
\end{proof}

\begin{proof}[Proof of Theorem~\ref{thm:approx-q=1}]
For clarity we give the proof in the one-dimensional setting, where $g$ outputs real
numbers; the argument extends to the multidimensional case without essential change.

All big-$O$ terms are taken in the limit $\hat{q} = 1-q \to 0$ or $\sigma\to\infty$,
unless otherwise stated. We introduce the following notation.
\begin{itemize}
  \item $\mu = g(X)$, and $\mu_j = g(X_j)$ for $j = 1,\ldots,n$.
  \item $\bar{P} = \mathcal{N}\!\left(\sum_j \mu_j + \mu,\,\sigma^2\right)$
        and $\bar{Q} = \mathcal{N}\!\left(\sum_j \mu_j,\,\sigma^2\right)$.
  \item For $i = 1,\ldots,n$,
        \begin{align}
          A_i(y)
          &= \frac{\mathcal{N}(\sum_j\mu_j + \mu - \mu_i,\,\sigma^2)(y)}{\bar{P}(y)} - 1
           = \exp\!\left(-\frac{\mu_i(y - \sum_j\mu_j - \mu)}{\sigma^2}
                          - \frac{\mu_i^2}{2\sigma^2}\right) - 1.
          \label{eq:Ai}
        \end{align}
  \item Similarly,
        \begin{align}
          A(y)
          = \frac{\bar{Q}(y)}{\bar{P}(y)} - 1
          &= \frac{\mathcal{N}(\sum_j\mu_j,\,\sigma^2)(y)}{\bar{P}(y)} - 1
           = \exp\!\left(-\frac{\mu(y - \sum_j\mu_j - \mu)}{\sigma^2}
                          - \frac{\mu^2}{2\sigma^2}\right) - 1.
          \label{eq:A}
        \end{align}
  \item For $i = 1,\ldots,n$,
        \begin{align*}
          B_i(y)
          = \frac{\mathcal{N}(\sum_j\mu_j - \mu_i,\,\sigma^2)(y)}{\bar{Q}(y)} - 1
          = \exp\!\left(-\frac{\mu_i(y - \sum_j\mu_j)}{\sigma^2}
                         - \frac{\mu_i^2}{2\sigma^2}\right) - 1.
        \end{align*}
\end{itemize}

\medskip
\noindent\textbf{Step 1: Density approximations.}
Applying the Taylor expansion
$q^\alpha = (1-\hat{q})^\alpha = 1 - \alpha\hat{q} + O(\hat{q}^2)$
to the mixture weights gives
\begin{align}
  \mathcal{M}(S)(y)
  &= \Bigl[
       1
       + \hat{q}\!\left(\sum_{j=1}^n A_j(y) + A(y)\right)
       + \hat{q}^2 O(1)
     \Bigr]\bar{P}(y),
  \label{eq:approx_M(S)_q~1} \\[4pt]
  \mathcal{M}(S^-)(y)
  &= \Bigl[
       1 + \hat{q}\sum_{j=1}^n B_j(y) + \hat{q}^2 O(1)
     \Bigr]\bar{Q}(y).
  \label{eq:approx_M(S-)_q~1}
\end{align}

\medskip
\noindent\textbf{Step 2: Approximation of $\E_P[\ln(P/Q)]$.}
By \eqref{eq:approx_M(S)_q~1} and \eqref{eq:approx_M(S-)_q~1},
\begin{align*}
  \E_P\!\left[\ln\frac{P}{Q}\right]
  &= \E_{y\sim\bar{P}}\!\left[\left(1 + \hat{q}\!\left(\sum_j A_j + A\right) + \hat{q}^2 O(1)\right)\right.\\
  &\qquad\left.\cdot\left(
      \ln\frac{\bar{P}(y)}{\bar{Q}(y)}
      + \ln\frac{
          1 + \hat{q}\!\left(\sum_j A_j + A\right) + \hat{q}^2 O(1)
        }{
          1 + \hat{q}\sum_j B_j + \hat{q}^2 O(1)
        }
    \right)
  \right].
\end{align*}
We split this into two parts, $E_1$ and $E_2$.

\smallskip
\textit{Part $E_1$.}
\begin{align*}
  E_1
  &= \E_{y\sim\bar{P}}\!\left[
     \left(1 + \hat{q}\!\left(\sum_j A_j + A\right) + \hat{q}^2 O(1)\right)
     \ln\frac{\bar{P}(y)}{\bar{Q}(y)}
   \right] \\
  &= \frac{\mu^2}{2\sigma^2}
     + \hat{q}\,\E_{y\sim\bar{P}}\!\left[
         \left(\sum_j A_j + A\right)\ln\frac{\bar{P}}{\bar{Q}}
       \right]
     + O(\hat{q}^2).
\end{align*}
Expanding $A_j$, $A$ via \eqref{eq:Ai}--\eqref{eq:A} and
$e^x - 1 = \sum_{k=1}^\infty x^k/k!$, and using
$\E_{y\sim\bar{P}} y = \sum_j\mu_j + \mu$ and $\Var_{\bar{P}} y = \sigma^2$,
the dominant contribution comes from terms of order $\hat{q}\,y^2/\sigma^4$:
\begin{align}
  M_1
  &\;=\; \hat{q}\,\E_{y\sim\bar{P}}\!\left[
     \left(\sum_j A_j + A\right)\ln\frac{\bar{P}}{\bar{Q}}
   \right]
   = \hat{q}\,\E_{y\sim\bar{P}}\!\left[
       \frac{\mu(-\sum_j\mu_j - \mu)\,y^2}{\sigma^4}
     \right]
     + \hat{q}\cdot O(\sigma^{-4}) \nonumber\\
  &\;=\; -\hat{q}\frac{\mu(\sum_j\mu_j + \mu)}{\sigma^2}
     + \hat{q}\cdot O(\sigma^{-4}), \nonumber\\[4pt]
  E_1
  &= \frac{\mu^2}{2\sigma^2}
     - \hat{q}\frac{\mu(\sum_j\mu_j + \mu)}{\sigma^2}
     + \hat{q}\cdot O(\sigma^{-4})
     + O(\hat{q}^2).
  \label{eq:E1}
\end{align}

\smallskip
\textit{Part $E_2$.}
Applying $\ln(1+q) = q + O(q^2)$ to the second log,
\begin{align*}
  E_2
  &= \E_{y\sim\bar{P}}\!\left[
     \left(1 + \hat{q}O(1)\right)
     \left(
       \hat{q}\!\left(\sum_j A_j + A - \sum_j B_j\right) + \hat{q}^2 O(1)
     \right)
   \right] \\
  &= \hat{q}\,\E_{y\sim\bar{P}}\!\left[
       \sum_j A_j + A - \sum_j B_j
     \right]
     + O(\hat{q}^2).
\end{align*}
Retaining only the $y^2/\sigma^4$ terms from the Taylor expansion of $A_j$, $A$,
$B_j$ and using the moments of $\bar{P}$,
\begin{align}
  E_2
  &= \hat{q}\frac{\mu\sum_j\mu_j}{\sigma^2}
     + \hat{q}\cdot O(\sigma^{-4})
     + O(\hat{q}^2).
  \label{eq:E2}
\end{align}
Combining \eqref{eq:E1} and \eqref{eq:E2},
\begin{align*}
  \E_P\!\left[\ln\frac{P}{Q}\right]
  = E_1 + E_2
  &= \left(1 - 2\hat{q}\right)\frac{\mu^2}{2\sigma^2}
     + \hat{q}\cdot O(\sigma^{-4})
     + O(\hat{q}^2) \\
  &= (1-\hat{q})^2\frac{\mu^2}{2\sigma^2}
     + \hat{q}\cdot O(\sigma^{-4})
     + O(\hat{q}^2).
\end{align*}
The case of $\E_Q[\ln(P/Q)]$ is analogous.

\medskip
\noindent\textbf{Step 3: Variance.}
We approximate $\E_P[(\ln(P/Q))^2]$ by splitting it into two parts, $V_1$ and $V_2$.
\begin{align*}
  \E_P\!\left[\left(\ln\frac{P}{Q}\right)^{\!2}\right]
  &= \E_{y\sim\bar{P}}\!
  \left[\left(1 + \hat{q}\!\left(\sum_j A_j + A\right) + \hat{q}^2 O(1)\right)\right.\\
  &\qquad\left.\cdot\left[
      \left(\ln\frac{\bar{P}}{\bar{Q}}\right)^{\!2}
      + 2\ln\frac{\bar{P}}{\bar{Q}}
        \ln\frac{
          1 + \hat{q}\!\left(\sum_j A_j + A\right) + \hat{q}^2 O(1)
        }{
          1 + \hat{q}\sum_j B_j + \hat{q}^2 O(1)
        }
      + \hat{q}^2 O(1)
    \right]
  \right].
\end{align*}

\smallskip
\textit{Part $V_1$.}
\begin{align}
  V_1
  &= \E_{y\sim\bar{P}}\!\left[
     \left(1 + \hat{q}\!\left(\sum_j A_j + A\right) + \hat{q}^2 O(1)\right)
     \left(\ln\frac{\bar{P}}{\bar{Q}}\right)^{\!2}
   \right] \\
  &= \E_{\bar{P}}\!\left[\left(\ln\frac{\bar{P}}{\bar{Q}}\right)^{\!2}\right]
     + \hat{q}\cdot O(\sigma^{-4})
     + O(\hat{q}^2).
  \label{eq:V1}
\end{align}

\smallskip
\textit{Part $V_2$.}
Retaining only terms of order $\hat{q}\,y^2/\sigma^4$,
\begin{align}
  V_2
  &= \E_{y\sim\bar{P}}\!\left[
     (1 + \hat{q}O(1))
     \cdot 2\ln\frac{\bar{P}(y)}{\bar{Q}(y)}
       \ln\frac{
         1 + \hat{q}(\sum_j A_j + A) + \hat{q}^2 O(1)
       }{
         1 + \hat{q}\sum_j B_j + \hat{q}^2 O(1)
       }
   \right] \nonumber\\
  &= \E_{y\sim\bar{P}}\!\left[
     2\hat{q}\!\left(\ln\frac{\bar{P}(y)}{\bar{Q}(y)}\right)
     \!\left(\sum_j A_j + A - \sum_j B_j\right)
   \right] + O(\hat{q}^2) \nonumber\\
  &= 2\hat{q}\,\E_{y\sim\bar{P}}\!\left[
     \frac{\mu(-\sum_j\mu_j - \mu + \sum_j\mu_j)\,y^2}{\sigma^4}
   \right]
   + \hat{q}\cdot O(\sigma^{-4})
   + O(\hat{q}^2) \nonumber\\
  &= -2\hat{q}\frac{\mu^2}{\sigma^2}
     + \hat{q}\cdot O(\sigma^{-4})
     + O(\hat{q}^2).
  \label{eq:V2}
\end{align}
Combining \eqref{eq:V1} and \eqref{eq:V2},
\begin{align*}
  \Var_P\!\left[\ln\frac{P}{Q}\right]
  &= V_1 + V_2 - \left(\E_P\!\left[\ln\frac{P}{Q}\right]\right)^{\!2} \\
  &= \E_{\bar{P}}\!\left[\left(\ln\frac{\bar{P}}{\bar{Q}}\right)^{\!2}\right]
     - 2\hat{q}\frac{\mu^2}{\sigma^2}
     - \left(1 - 4\hat{q}\right)\!\left(\frac{\mu^2}{2\sigma^2}\right)^{\!2}
     + \hat{q}\cdot O(\sigma^{-4})
     + O(\hat{q}^2) \\
  &= \left(\frac{\mu^2}{\sigma^2} + \left(\frac{\mu^2}{2\sigma^2}\right)^{\!2}\right)
     - 2\hat{q}\frac{\mu^2}{\sigma^2}
     - \left(\frac{\mu^2}{2\sigma^2}\right)^{\!2}
     + \hat{q}\cdot O(\sigma^{-4})
     + O(\hat{q}^2) \\
  &= \left(1 - 2\hat{q}\right)\frac{\mu^2}{\sigma^2}
     + \hat{q}\cdot O(\sigma^{-4})
     + O(\hat{q}^2)
   = (1-\hat{q})^2\frac{\mu^2}{\sigma^2}
     + \hat{q}\cdot O(\sigma^{-4})
     + O(\hat{q}^2).
\end{align*}
The case of $\Var_Q[\ln(P/Q)]$ is identical.

\medskip
\noindent\textbf{Step 4: Third absolute central moment.}
We have
\begin{align*}
  & \E_{P}\!\left|\ln\frac{P}{Q} - \E_P\!\left[\ln\frac{P}{Q}\right]\right|^3 \\
  &= \E_{y\sim\bar{P}}\!
  \left[
    \left(1 + \hat{q}\!\left(\sum_j A_j + A\right) + \hat{q}^2 O(1)\right)
    \left|
      \ln\frac{\bar{P}(y)}{\bar{Q}(y)}
      + \ln\frac{
          1 + \hat{q}(\sum_j A_j + A) + \hat{q}^2 O(1)
        }{
          1 + \hat{q}\sum_j B_j + \hat{q}^2 O(1)
        }\right.\right.\\
  &\left.\left.\hspace{7.5cm} - (1-2\hat{q})\frac{\mu^2}{2\sigma^2}
      + \hat{q}\cdot O(\sigma^{-4})
      + \hat{q}^2 O(1)
    \right|^3
  \right] \\
  &= \E_{y\sim\bar{P}}\!
  \left[
    \left(1 + \hat{q}\!\left(\sum_j A_j + A\right) + \hat{q}^2 O(1)\right)
    \left|
      \ln\frac{\bar{P}(y)}{\bar{Q}(y)} - \frac{\mu^2}{2\sigma^2}\right.\right.\\
  &\left.\left.\hspace{5cm}
      + \hat{q}\!\left(
          \sum_j A_j + A - \sum_j B_j + \frac{\mu^2}{\sigma^2} + O(\sigma^{-4})
        \right)
      + \hat{q}^2 O(1)
    \right|^3
  \right].
\end{align*}
Applying the triangle inequality, noting $|O(1)| = O(1)$,
\begin{align*}
  &\E_{P}\!\left|\ln\frac{P}{Q} - \E_P\!\left[\ln\frac{P}{Q}\right]\right|^3 \\
  &\leq \E_{y\sim\bar{P}}\!
  \left[
    \left(1 + \hat{q}\!\left(\sum_j A_j + A\right) + \hat{q}^2 O(1)\right)
    \cdot\left(
      \left|\ln\frac{\bar{P}(y)}{\bar{Q}(y)} - \frac{\mu^2}{2\sigma^2}\right|
    \right.\right.\\
  &\left.\left.\hspace{6cm}
      + \hat{q}\!\left|\sum_j A_j + A - \sum_j B_j + O(\sigma^{-2})\right|
      + \hat{q}^2 O(1)
    \right)^3
  \right] \\
  &= \E_{y\sim\bar{P}}\!
  \left[\left(1 + \hat{q}\!\left(\sum_j A_j + A\right) + \hat{q}^2 O(1)\right)\cdot\left(
      \left|\ln\frac{\bar{P}(y)}{\bar{Q}(y)} - \frac{\mu^2}{2\sigma^2}\right|^3
  \right.\right.\\
  &\left.\left.\hspace{3.3cm}
      + 3\hat{q}\!\left(\ln\frac{\bar{P}(y)}{\bar{Q}(y)} - \frac{\mu^2}{2\sigma^2}\right)^{\!2}
        \left|\sum_j A_j + A - \sum_j B_j + O(\sigma^{-2})\right|
      + \hat{q}^2 O(1)
    \right)
  \right].
\end{align*}
We split this into $T_1$ and $T_2$.

\smallskip
\textit{Part $T_1$.}
\begin{align}
  T_1
  &= \E_{y\sim\bar{P}}\!\left[
     \left(1 + \hat{q}\!\left(\sum_j A_j + A\right) + \hat{q}^2 O(1)\right)
     \left|\ln\frac{\bar{P}(y)}{\bar{Q}(y)} - \frac{\mu^2}{2\sigma^2}\right|^3
   \right] \nonumber\\
  &= \E_{y\sim\bar{P}}\!\left|\ln\frac{\bar{P}(y)}{\bar{Q}(y)} - \frac{\mu^2}{2\sigma^2}\right|^3
     + \hat{q}\cdot O(\sigma^{-5})
     + O(\hat{q}^2) \nonumber\\
  &= \sqrt{\tfrac{8}{\pi}}\,\frac{\mu^3}{\sigma^3}
     + \hat{q}\cdot O(\sigma^{-5})
     + O(\hat{q}^2).
  \label{eq:T1}
\end{align}

\smallskip
\textit{Part $T_2$.}
Applying the triangle inequality and retaining terms of order $\hat{q}/\sigma^3$,
\begin{align}
  T_2
  &= \E_{y\sim\bar{P}}\!\left[
     (1 + \hat{q}O(1))
     \cdot 3\hat{q}\!
     \left(\ln\frac{\bar{P}(y)}{\bar{Q}(y)} - \frac{\mu^2}{2\sigma^2}\right)^{\!2}
     \left|\sum_j A_j + A - \sum_j B_j + O(\sigma^{-2})\right|
   \right] \nonumber\\
  &= \E_{y\sim\bar{P}}\!\left[
     3\hat{q}\!
     \left(\ln\frac{\bar{P}(y)}{\bar{Q}(y)} - \frac{\mu^2}{2\sigma^2}\right)^{\!2}
     \left|\sum_j A_j + A - \sum_j B_j + O(\sigma^{-2})\right|
   \right] + O(\hat{q}^2) \nonumber\\
  &\leq \E_{y\sim\bar{P}}\!\left[
     3\hat{q}\,\frac{\mu^2(y - \sum_j\mu_j - \mu)^2}{\sigma^4}
     \left|-\frac{\mu(y - \sum_j\mu_j - \mu)}{\sigma^2}\right|
   \right]
   + \hat{q}\cdot O(\sigma^{-4})
   + O(\hat{q}^2) \nonumber\\
  &= \frac{3\hat{q}\,\mu^3\,\E_{y\sim\bar{P}}\!|y - \sum_j\mu_j - \mu|^3}{\sigma^6}
     + \hat{q}\cdot O(\sigma^{-4})
     + O(\hat{q}^2) \nonumber\\
  &= 3\hat{q}\sqrt{\tfrac{8}{\pi}}\,\frac{\mu^3}{\sigma^3}
     + \hat{q}\cdot O(\sigma^{-4})
     + O(\hat{q}^2).
  \label{eq:T2}
\end{align}
Combining \eqref{eq:T1} and \eqref{eq:T2},
\begin{align*}
  \E_{P}\!\left|\ln\frac{P}{Q} - \E_P\!\left[\ln\frac{P}{Q}\right]\right|^3
  \leq T_1 + T_2
  &\leq (1 + 3\hat{q})\sqrt{\tfrac{8}{\pi}}\,\frac{\mu^3}{\sigma^3}
     + \hat{q}\cdot O(\sigma^{-4})
     + O(\hat{q}^2) \\
  &= (2 - q^3)\sqrt{\tfrac{8}{\pi}}\,\frac{\mu^3}{\sigma^3}
     + \hat{q}\cdot O(\sigma^{-4})
     + O(\hat{q}^2).
\end{align*}
The case of $\E_Q|\ln(P/Q) - \E_Q[\ln(P/Q)]|^3$ is identical.
\end{proof}

\section{Further Monte Carlo Simulations for the Approximate GDP Filter}%
\label{appx:experiments}

\subsection{Conversion of RDP Order~1 to $(\varepsilon,\delta)$-DP}

The standard algebraic conversion from RDP to $(\varepsilon,\delta)$-DP
\citep{canonne_discrete_2020} is undefined at RDP order~1, because the formula
involves a factor that diverges there. This is not merely an algebraic singularity --- the underlying proof breaks down at order~1. We adapt that proof to derive a conversion dedicated to RDP order~1. 

We note that this algebraic conversion is not tight. Optimal conversion is available through the Bernoulli privacy-region characterization of \cite{riess_optimal_2026}. The algebraic conversion is nonetheless a standard choice in the literature, and we follow that convention here.

\begin{theorem}\label{thm:rdp-order1-conversion}
	If a mechanism satisfies $(1,B)$-RDP, then for all $\varepsilon>0$ it satisfies
	$(\varepsilon,\delta)$-DP with $\delta=\min\{1,cB\}$, where $c>0$ is the
	unique solution to
	\begin{equation*}
		c\!\left(e^{1/c}-1\right)=e^\varepsilon.
	\end{equation*}
\end{theorem}

\begin{proof}
	We find the smallest $c$ such that
	\begin{equation}
		\left[1-e^{\varepsilon-l}\right]_+
		\leq c\!\left(l-1+e^{-l}\right)
		\quad\text{for all }l\in\real. \label{eq:rdp-order1-pointwise}
	\end{equation}
	The right-hand side is nonnegative since $e^{-l}\geq1-l$. This inequality
	implies that for every remove pair $(P,Q)$ and PLRV $L(y)=\ln(P(y)/Q(y))$,
	\begin{align*}
        \E_{y\sim P}\!\left[\left[1-e^{\varepsilon-L(y)}\right]_+\right]
        &\leq c\cdot\E_{y\sim P}\!\left[L(y)-1+e^{-L(y)}\right] \\
        &=c\cdot\E_{y\sim P}[L(y)]\leq c\cdot B.
	\end{align*}
	Here $\E_{y\sim P}[e^{-L(y)}]=1$, since the finite RDP bound in both
	directions implies that $P$ and $Q$ are mutually absolutely continuous.
	Taking the supremum over all neighboring pairs, the LHS becomes the smallest $\delta$ for which the mechanism is $(\varepsilon,\delta)$-DP.

	Note that the smallest $c$ such that \eqref{eq:rdp-order1-pointwise} holds is
	\begin{equation*}
		c = \max_{l>\epsilon} \frac{1-e^{\varepsilon-l}}{l-1+e^{-l}}.
	\end{equation*}
	The inner ratio tends to zero as $l\searrow\varepsilon$ or $l\nearrow\infty$,
	and its derivative has the sign of $1+e^\varepsilon l-e^l$. It therefore
	has a unique maximum, characterized by the system
	\begin{equation}
		1-e^{\varepsilon-l}=c\!\left(l-1+e^{-l}\right),
		\qquad e^{\varepsilon-l}=c\!\left(1-e^{-l}\right).
	\end{equation}
	These equations give $l=1/c$ and hence $c(e^{1/c}-1)=e^\varepsilon$ with $c>0$.
\end{proof}

\subsection{Details of Monte Carlo Simulations}

\begingroup
\setlength{\intextsep}{8pt}
\renewcommand{\theHtable}{G2.\arabic{table}}
\renewcommand{\theHfigure}{G2.\arabic{figure}}

We simulate Algorithm~\ref{alg:approx-GDP-filter} on a one-dimensional remove pair with independent Poisson subsampling at a fixed rate $q$. Conditional on the released history, the removed datapoint contributes $\mu_t=\min(C_t,C_{B,t})$, giving
\begin{align*}
	P_t &= (1-q)\N(0,\sigma_t^2 C_t^2)+q\N(\mu_t,\sigma_t^2 C_t^2), &
	Q_t &= \N(0,\sigma_t^2 C_t^2),
\end{align*}
and
\begin{equation}
	L_t(y_t) = \ln\!\left[1-q+q\exp\!\left(\frac{\mu_t(2y_t-\mu_t)}{2\sigma_t^2 C_t^2}\right)\right].\label{eq:mc-step-loss}
\end{equation}

The first step uses $(\sigma_1,C_1)=(\sigma_0,C_0)$. For $t\geq2$, we use the AdaGrad-inspired rule \citep{JMLR:v12:duchi11a}
\begin{align*}
	\sigma_t &= \operatorname{clamp}\!\left(\frac{\sigma_0}{\sqrt{\sum_{i<t}y_i^2}};[\sigma_{\min},\sigma_0]\right), &
	C_t &= \frac{C_0}{\sqrt{\sum_{i<t}y_i^2}},
\end{align*}
with noise multipliers $\sigma_0=8$ and $\sigma_{\min}=2$, and $C_0=1$. The declining $C_t$ models decreasing update norms during training, motivated by adaptive clipping \citep{andrew2021differentially}. Although the one-step privacy-loss distribution does not depend on $C_t$, clipping scales the released outputs and hence affects subsequent parameter choices.

For each of the five configurations in Table~\ref{tab:delta-gdp}, we simulate $N=10^{10}$ independent trajectories under each of $P$ and $Q$. We choose $T$ and $B$ so that the filter exhausts its budget by $T$ with probability at least $98\%$ in each direction. We use $\tilde q=0$ for $q\in\{0.01,0.1,0.199\}$ and $\tilde q=1$ for $q\in\{0.801,0.95\}$, with the budget charges and final clipping rule of Algorithm~\ref{alg:approx-GDP-filter}.

\para{Direct privacy-curve estimation.}
Let $L_i^P$ and $L_i^Q$ denote the total privacy losses from the $i$-th independently simulated trajectory under the transcript laws $P$ and $Q$, respectively. Using the step loss in~\eqref{eq:mc-step-loss}, define $L_{i,t}=L_t(y_{i,t})$, so that
\begin{equation*}
	L_i=\ln(P/Q)=\sum_{t=1}^{T}L_{i,t}=\sum_{t=1}^{\tau_i}L_{i,t}.
\end{equation*}
Here $y_{i,t}$ is the output of trajectory $i$ at step $t$, and $\tau_i$ is its budget-binding step, or $T$ if the horizon is reached first; we set $L_{i,t}=0$ for $t>\tau_i$. The directed hockey-stick identities yield the estimators
\begin{align*}
	\widehat\delta_P(\varepsilon) &= \frac1N\sum_{i=1}^N\left[1-\e^{\varepsilon-L_i^P}\right]_+, \\
	\widehat\delta_Q(\varepsilon) &= \frac1N\sum_{i=1}^N\left[1-\e^{\varepsilon+L_i^Q}\right]_+.
\end{align*}
The $P$-expectation and $Q$-expectation of the respective summands give the remove- and add-direction privacy curves. We estimate the add/remove curve by $\widehat\delta(\varepsilon)=\max\{\widehat\delta_P(\varepsilon),\widehat\delta_Q(\varepsilon)\}$.

We construct a simultaneous confidence band using the empirical-Bernstein bound of \citet[Theorem~4]{maurer_empirical_2009}. Let $\widehat V_P(\varepsilon)$ and $\widehat V_Q(\varepsilon)$ be the sample variances of the summands defining $\widehat\delta_P(\varepsilon)$ and $\widehat\delta_Q(\varepsilon)$, respectively, each with denominator $N-1$. For total Monte Carlo failure probability $\beta=10^{-9}$ and at most $M=5\times2\times700{,}000$ candidate $\varepsilon$ values on the evaluation grid across all five $q$ values and both directions, the radii are
\begin{align*}
	r_P(\varepsilon)&=\sqrt{\frac{2\widehat V_P(\varepsilon)\ln(4M/\beta)}{N}}+\frac{7\ln(4M/\beta)}{3(N-1)},\\
	r_Q(\varepsilon)&=\sqrt{\frac{2\widehat V_Q(\varepsilon)\ln(4M/\beta)}{N}}+\frac{7\ln(4M/\beta)}{3(N-1)}.
\end{align*}
By a union bound, with probability at least $1-\beta$, both bounds
\begin{align*}
	\left|\widehat\delta_P(\varepsilon)-\delta_P(\varepsilon)\right|&\leq r_P(\varepsilon), &
	\left|\widehat\delta_Q(\varepsilon)-\delta_Q(\varepsilon)\right|&\leq r_Q(\varepsilon)
\end{align*}
hold simultaneously at every candidate $\varepsilon$ across all five $q$ values. The shaded confidence band extends from $\max\{0,\widehat\delta_P-r_P,\widehat\delta_Q-r_Q\}$ to $\min\{1,\max[\widehat\delta_P+r_P,\widehat\delta_Q+r_Q]\}$. We plot the EVR-style pessimistic estimate \citep{wang_randomized_2023}
\begin{equation*}
	\min\!\left\{1,\max\!\left[\widehat\delta_P(\varepsilon)+r_P(\varepsilon),\widehat\delta_Q(\varepsilon)+r_Q(\varepsilon)\right]+\beta\right\}.
\end{equation*}

The candidate $\varepsilon$ values are the positive points of a fixed uniform grid dividing $[B-15\sqrt{2B},B+20\sqrt{2B}]$ into $700{,}000$ equal intervals.

\para{Gaussian comparison and results.}
We compare the empirical CDFs of $L$ under $P$ and $Q$ with $\N(B,2B)$ and $\N(-B,2B)$, respectively. Writing $\widehat F_P,\widehat F_Q$ for the empirical CDFs, the reported discrepancies are
\begin{align*}
	\widehat\Delta_P &= \max_x\left|\widehat F_P(x)-\Phi\!\left(\frac{x-B}{\sqrt{2B}}\right)\right|, \\
	\widehat\Delta_Q &= \max_x\left|\widehat F_Q(x)-\Phi\!\left(\frac{x+B}{\sqrt{2B}}\right)\right|, \\
	\widehat\Delta &= \max(\widehat\Delta_P,\widehat\Delta_Q),
\end{align*}
where $x$ ranges over the evaluated points. The privacy-curve panels compare the simulation with the exact $\sqrt{2B}$-GDP curve used when applying Theorem~\ref{thm:validity-approx-GDP-filter} in practice with $\Delta$ omitted, and with the matched $(1,B)$-RDP baseline from Appendix~G.1.

Despite output-dependent noise and clipping, the empirical PLRV means and variances remain close to the Gaussian predictions $\pm B$ and $2B$ (Table~\ref{tab:plrv-moments}), supporting the moment approximations underlying the filter. At $\delta=10^{-6}$, the empirical $\varepsilon$ is approximately $2\%$--$11\%$ higher than the $\sqrt{2B}$-GDP prediction across the five settings; for $q=0.1$, the values are $\varepsilon\approx2.81$ and $\varepsilon\approx2.56$, respectively.

\begin{figure}[p]
	\centering
	\includegraphics{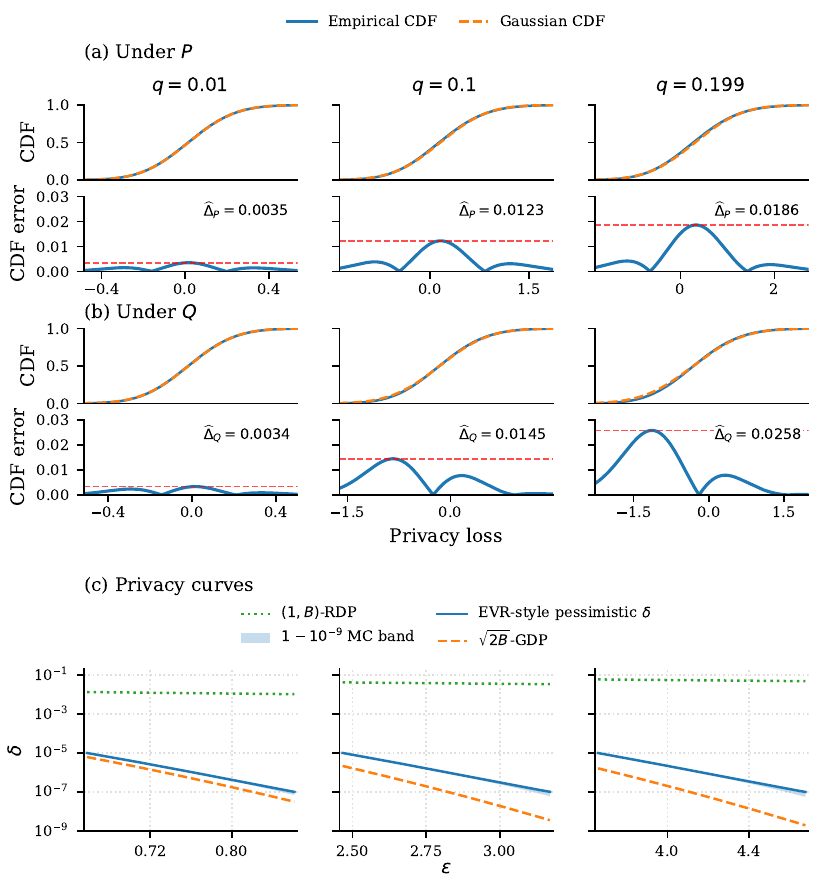}
	\caption{Monte Carlo evaluation in the low-subsampling regime, with $10^{10}$ trajectories per direction at each $q$ (Table~\ref{tab:delta-gdp}). \textbf{(a) Under $P$:} empirical CDF of $L$ versus $\N(B,2B)$ and absolute CDF discrepancy. \textbf{(b) Under $Q$:} the corresponding comparison with $\N(-B,2B)$. \textbf{(c) Privacy curves:} the EVR-style pessimistic add/remove $\delta$ estimate and simultaneous $1-10^{-9}$ confidence band compared with $\sqrt{2B}$-GDP with $\Delta$ omitted and the matched $(1,B)$-RDP baseline. At $q=0.01$, halting rates under $(P,Q)$ are $(100.00\%,100.00\%)$, with conditional mean stopping steps $(1111.2,1111.2)$ or $(11.11,11.11)$ epochs. At $q=0.1$, halting rates under $(P,Q)$ are $(98.87\%,98.80\%)$, with conditional mean stopping steps $(113.1,113.1)$ or $(11.31,11.31)$ epochs. At $q=0.199$, halting rates under $(P,Q)$ are $(99.24\%,99.13\%)$, with conditional mean stopping steps $(58.0,58.0)$ or $(11.53,11.54)$ epochs. }
	\label{fig:mc-low-q}
	\label{fig:mc-q001}
	\label{fig:mc-q01}
	\label{fig:mc-q0199}
\end{figure}

\begin{figure}[p]
	\centering
	\includegraphics{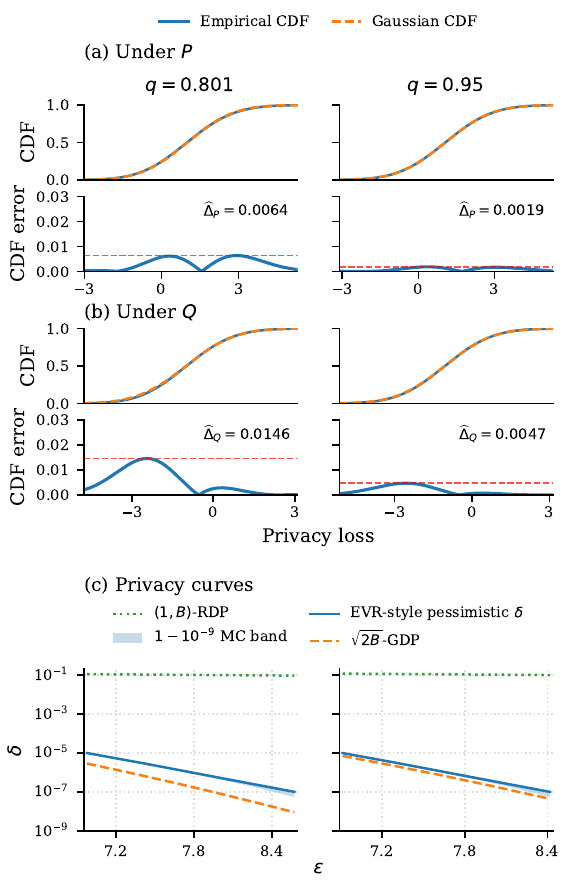}
	\caption{Monte Carlo evaluation in the high-subsampling regime, with $10^{10}$ trajectories per direction at each $q$ (Table~\ref{tab:delta-gdp}). \textbf{(a) Under $P$:} empirical CDF of $L$ versus $\N(B,2B)$ and absolute CDF discrepancy. \textbf{(b) Under $Q$:} the corresponding comparison with $\N(-B,2B)$. \textbf{(c) Privacy curves:} the EVR-style pessimistic add/remove $\delta$ estimate and simultaneous $1-10^{-9}$ confidence band compared with $\sqrt{2B}$-GDP with $\Delta$ omitted and the matched $(1,B)$-RDP baseline. At $q=0.801$, halting rates under $(P,Q)$ are $(98.61\%,98.06\%)$, with conditional mean stopping steps $(14.8,14.8)$ or $(11.85,11.88)$ epochs. At $q=0.95$, halting rates under $(P,Q)$ are $(99.98\%,99.97\%)$, with conditional mean stopping steps $(11.8,11.9)$ or $(11.22,11.29)$ epochs. }
	\label{fig:mc-high-q}
	\label{fig:mc-q0801}
	\label{fig:mc-q095}
\end{figure}
\clearpage

\begin{table}[htbp]
	\centering
	\caption{Empirical PLRV moments and the Gaussian reference. The reference mean is $+B$ under $P$ and $-B$ under $Q$, with variance $2B$ under both.}
	\label{tab:plrv-moments}
	\small
	\begin{tabular}{c c cc cc cc}
		\toprule
		& & \multicolumn{2}{c}{Gaussian reference} & \multicolumn{2}{c}{Empirical under $P$} & \multicolumn{2}{c}{Empirical under $Q$} \\
		\cmidrule(lr){3-4}\cmidrule(lr){5-6}\cmidrule(lr){7-8}
		$q$ & $B$ & Mean & Variance & Mean & Variance & Mean & Variance \\
		\midrule
		$0.01$ & $0.01575$ & $\pm0.01575$ & $0.0315$ & $0.0157$ & $0.0315$ & $-0.01565$ & $0.03121$ \\
		$0.1$ & $0.1575$ & $\pm0.1575$ & $0.315$ & $0.1532$ & $0.3142$ & $-0.1494$ & $0.2913$ \\
		$0.199$ & $0.3118$ & $\pm0.3118$ & $0.6236$ & $0.2967$ & $0.6187$ & $-0.2846$ & $0.5464$ \\
		$0.801$ & $1.003$ & $\pm1.003$ & $2.005$ & $1.006$ & $2.087$ & $-0.9693$ & $1.866$ \\
		$0.95$ & $1.058$ & $\pm1.058$ & $2.115$ & $1.058$ & $2.139$ & $-1.046$ & $2.066$ \\
		\bottomrule
	\end{tabular}
\end{table}

\begin{table}[htbp]
	\centering
	\caption{Simulation configurations and empirical CDF discrepancies. The overall discrepancy is $\widehat\Delta=\max(\widehat\Delta_P,\widehat\Delta_Q)$.}
	\label{tab:delta-gdp}
	\small
	\setlength{\tabcolsep}{4pt}
	\begin{tabular}{c c c c c c c}
		\toprule
		$q$ & $T$ & $B$ & $\widehat\Delta_P$ & $\widehat\Delta_Q$ & $\widehat\Delta$ & $\sqrt{2B}$ \\[4pt]
		\midrule
		$0.01$ & $1200$ & $0.01575$ & $0.003533$ & $0.003377$ & $0.003533$ & $0.1775$ \\
		$0.1$ & $120$ & $0.1575$ & $0.01228$ & $0.0145$ & $0.0145$ & $0.5612$ \\
		$0.199$ & $65$ & $0.3118$ & $0.01862$ & $0.02584$ & $0.02584$ & $0.7897$ \\
		$0.801$ & $20$ & $1.003$ & $0.006391$ & $0.01461$ & $0.01461$ & $1.416$ \\
		$0.95$ & $20$ & $1.058$ & $0.001889$ & $0.004735$ & $0.004735$ & $1.454$ \\
		\bottomrule
	\end{tabular}
\end{table}
\clearpage
\endgroup


\newpage
\section*{NeurIPS Paper Checklist}

\begin{enumerate}

\item {\bf Claims}
    \item[] Question: Do the main claims made in the abstract and introduction accurately reflect the paper's contributions and scope?
    \item[] Answer: \answerYes{} 
    \item[] Justification: The abstract and introduction each state the four contributions, which correspond directly to Corollary~\ref{cor:SG_counterexample}, Theorems~\ref{thm:validity} and~\ref{thm:nochain=>invalid}, Theorem~\ref{thm:fully-adaptive-CLT}, and Algorithm~\ref{alg:approx-GDP-filter} respectively. The suggested general methodology in Contribution 3 is framed as a consequence of the CLT rather than a proven general result, so no overclaiming occurs.

\item {\bf Limitations}
    \item[] Question: Does the paper discuss the limitations of the work performed by the authors?
    \item[] Answer: \answerYes{} 
    \item[] Justification: Limitations are discussed throughout the paper and summarized in Section~5: the approximate GDP filter underestimates privacy loss at realistic subsampling rates, a known limitation of CLT-based Gaussian approximations acknowledged explicitly in the abstract, Figure~\ref{fig:CLT-validation}, and Table~\ref{tab:delta-gdp-conv}; and the filter's validity guarantees are restricted to the asymptotic regimes $q\ll1$ and $q\approx1$, with the gap to exact accounting quantified by $\Delta$ in Theorem~\ref{thm:validity-approx-GDP-filter}.

\item {\bf Theory assumptions and proofs}
    \item[] Question: For each theoretical result, does the paper provide the full set of assumptions and a complete (and correct) proof?
    \item[] Answer: \answerYes{} 
    \item[] Justification: All theorems state assumptions explicitly in their statement; full proofs appear in the appendices with each theorem cross-referenced to the appendix where the proof lies. Proof sketches and intuition are provided in the main text alongside each result, and intermediate lemmas are stated and proved in the appendix to patch any gaps in the main arguments.

    \item {\bf Experimental result reproducibility}
    \item[] Question: Does the paper fully disclose all the information needed to reproduce the main experimental results of the paper to the extent that it affects the main claims and/or conclusions of the paper (regardless of whether the code and data are provided or not)?
    \item[] Answer: \answerYes{} 
	\item[] Justification: All mechanism parameters and integral formulas for the counterexample are given in Appendix~\ref{appx:counterexample}. Appendix~\ref{appx:experiments} specifies the Monte Carlo composition rule, five configurations, $10^{10}$ independent trajectories per direction and configuration, direct privacy-curve estimators, and simultaneous empirical-Bernstein confidence bands.


\item {\bf Open access to data and code}
    \item[] Question: Does the paper provide open access to the data and code, with sufficient instructions to faithfully reproduce the main experimental results, as described in supplemental material?
    \item[] Answer: \answerNo{} 
	\item[] Justification: The paper does not provide a public code or data link. The counterexample is specified in Appendix~\ref{appx:counterexample}, and the Monte Carlo policy, parameters, estimators and confidence-band construction are specified in Appendix~\ref{appx:experiments}.

\item {\bf Experimental setting/details}
    \item[] Question: Does the paper specify all the training and test details (e.g., data splits, hyperparameters, how they were chosen, type of optimizer) necessary to understand the results?
    \item[] Answer: \answerYes{} 
	\item[] Justification: The paper contains no model training. Appendix~\ref{appx:counterexample} gives the counterexample parameters and privacy-profile integrals. Appendix~\ref{appx:experiments} gives the adaptive noise and clipping rule, horizons and budgets, subsampling rates, sample sizes, and privacy-curve estimation procedure.

\item {\bf Experiment statistical significance}
    \item[] Question: Does the paper report error bars suitably and correctly defined or other appropriate information about the statistical significance of the experiments?
    \item[] Answer: \answerYes{} 
	\item[] Justification: Figure~\ref{fig:CLT-validation} and the privacy-curve panels in Appendix~\ref{appx:experiments} report simultaneous $1-10^{-9}$ confidence bands for Monte Carlo sampling uncertainty across the candidate $\varepsilon$ values, all five configurations and both add/remove directions. Appendix~\ref{appx:experiments} gives the empirical-Bernstein bounds and union-bound allocation. The counterexample privacy profiles are computed by deterministic numerical integration.

\item {\bf Experiments compute resources}
    \item[] Question: For each experiment, does the paper provide sufficient information on the computer resources (type of compute workers, memory, time of execution) needed to reproduce the experiments?
    \item[] Answer: \answerNo{} 
	\item[] Justification: The paper does not report the hardware, memory requirements or execution times for the $10^{10}$-trajectory Monte Carlo simulations.
    
\item {\bf Code of ethics}
    \item[] Question: Does the research conducted in the paper conform, in every respect, with the NeurIPS Code of Ethics \url{https://neurips.cc/public/EthicsGuidelines}?
    \item[] Answer: \answerYes{} 
    \item[] Justification: The paper is purely theoretical work on privacy accounting and raises no ethical concerns.

\item {\bf Broader impacts}
    \item[] Question: Does the paper discuss both potential positive societal impacts and negative societal impacts of the work performed?
    \item[] Answer: \answerYes{} 
    \item[] Justification: The work improves the theoretical foundations of differential privacy accounting, enabling tighter and more accurate privacy guarantees for sensitive data processing — a direct positive societal impact. No negative societal impacts are foreseen, as the work strengthens rather than weakens privacy protections and has no direct path to misuse.
    
\item {\bf Safeguards}
    \item[] Question: Does the paper describe safeguards that have been put in place for responsible release of data or models that have a high risk for misuse (e.g., pre-trained language models, image generators, or scraped datasets)?
    \item[] Answer: \answerNA{} 
    \item[] Justification: The paper releases no models, datasets, or other artifacts that pose misuse risks.

\item {\bf Licenses for existing assets}
    \item[] Question: Are the creators or original owners of assets (e.g., code, data, models), used in the paper, properly credited and are the license and terms of use explicitly mentioned and properly respected?
    \item[] Answer: \answerNA{} 
    \item[] Justification: The paper uses no existing assets such as code, data, or models.

\item {\bf New assets}
    \item[] Question: Are new assets introduced in the paper well documented and is the documentation provided alongside the assets?
    \item[] Answer: \answerNA{} 
    \item[] Justification: The paper releases no new assets.

\item {\bf Crowdsourcing and research with human subjects}
    \item[] Question: For crowdsourcing experiments and research with human subjects, does the paper include the full text of instructions given to participants and screenshots, if applicable, as well as details about compensation (if any)? 
    \item[] Answer: \answerNA{} 
    \item[] Justification: The paper involves no crowdsourcing or research with human subjects.

\item {\bf Institutional review board (IRB) approvals or equivalent for research with human subjects}
    \item[] Question: Does the paper describe potential risks incurred by study participants, whether such risks were disclosed to the subjects, and whether Institutional Review Board (IRB) approvals (or an equivalent approval/review based on the requirements of your country or institution) were obtained?
    \item[] Answer: \answerNA{} 
    \item[] Justification: The paper involves no crowdsourcing or research with human subjects.

\item {\bf Declaration of LLM usage}
    \item[] Question: Does the paper describe the usage of LLMs if it is an important, original, or non-standard component of the core methods in this research? Note that if the LLM is used only for writing, editing, or formatting purposes and does \emph{not} impact the core methodology, scientific rigor, or originality of the research, declaration is not required.
    \item[] Answer: \answerYes{} 
    \item[] Justification: LLM assistance with manuscript preparation, theoretical analysis, and numerical experiments is described in the acknowledgements.

\end{enumerate}
\addcontentsline{toc}{section}{NeurIPS Paper Checklist}

\end{document}